\documentclass[3p,11pt]{elsarticle}
\usepackage{amssymb}
\usepackage{amsmath}
\usepackage{amsthm}
\usepackage{xcolor}

\DeclareMathAlphabet{\mathcal}{OMS}{cmsy}{m}{n}

\makeatletter
\def\ps@pprintTitle{%
 \let\@oddhead\@empty
 \let\@evenhead\@empty
 \def\@oddfoot{\centerline{\thepage}}%
 \let\@evenfoot\@oddfoot}
\makeatother

\newcommand{\bbC}{\mathbb{C}}
\newcommand{\bbF}{\mathbb{F}}
\newcommand{\bbH}{\mathbb{H}}
\newcommand{\bbQ}{\mathbb{Q}}
\newcommand{\bbR}{\mathbb{R}}
\newcommand{\bbT}{\mathbb{T}}
\newcommand{\bbZ}{\mathbb{Z}}

\newcommand{\bfA}{\mathbf{A}}
\newcommand{\bfC}{\mathbf{C}}
\newcommand{\bfD}{\mathbf{D}}
\newcommand{\bfG}{\mathbf{G}}

\newcommand{\bfI}{\mathbf{I}}
\newcommand{\bfJ}{\mathbf{J}}
\newcommand{\bfM}{\mathbf{M}}
\newcommand{\bfP}{\mathbf{P}}
\newcommand{\bfS}{\mathbf{S}}

\newcommand{\bfu}{\mathbf{u}}
\newcommand{\bfU}{\mathbf{U}}
\newcommand{\bfv}{\mathbf{v}}
\newcommand{\bfV}{\mathbf{V}}
\newcommand{\bfx}{\mathbf{x}}

\newcommand{\bfy}{\mathbf{y}}
\newcommand{\bfZ}{\mathbf{Z}}

\newcommand{\bfone}{\boldsymbol{1}}
\newcommand{\bfzero}{\boldsymbol{0}}
\newcommand{\bfdelta}{\boldsymbol{\delta}}
\newcommand{\bfDelta}{\boldsymbol{\Delta}}
\newcommand{\bfGamma}{\boldsymbol{\Gamma}}
\newcommand{\bfphi}{\boldsymbol{\varphi}}
\newcommand{\bfPhi}{\boldsymbol{\Phi}}
\newcommand{\bfPsi}{\boldsymbol{\Psi}}

\newcommand{\calD}{\mathcal{D}}
\newcommand{\calE}{\mathcal{E}}
\newcommand{\calG}{\mathcal{G}}

\newcommand{\calJ}{\mathcal{J}}
\newcommand{\calK}{\mathcal{K}}
\newcommand{\calM}{\mathcal{M}}
\newcommand{\calN}{\mathcal{N}}
\newcommand{\calR}{\mathcal{R}}
\newcommand{\calS}{\mathcal{S}}
\newcommand{\calU}{\mathcal{U}}
\newcommand{\calV}{\mathcal{V}}
\newcommand{\calW}{\mathcal{W}}

\newcommand{\rmc}{\mathrm{c}}
\newcommand{\rmH}{\mathrm{H}}
\newcommand{\rmi}{\mathrm{i}}
\newcommand{\rms}{\mathrm{s}}
\newcommand{\rmT}{\mathrm{T}}

\newcommand{\Tr}{\operatorname{Tr}}

\newcommand{\Fro}{\mathrm{Fro}}
\newcommand{\dist}{\operatorname{dist}}
\newcommand{\rank}{\operatorname{rank}}
\newcommand{\Span}{\operatorname{span}}
\newcommand{\real}{\operatorname{Re}}
\newcommand{\imag}{\operatorname{Im}}

\newcommand{\TFF}{{\operatorname{TFF}}}
\newcommand{\ETF}{{\operatorname{ETF}}}
\newcommand{\ECTFF}{{\operatorname{ECTFF}}}
\newcommand{\EITFF}{{\operatorname{EITFF}}}

\newcommand{\abs}[1]{|{#1}|}

\newcommand{\bigparen}[1]{\bigl({#1}\bigr)}
\newcommand{\Bigparen}[1]{\Bigl({#1}\Bigr)}
\newcommand{\biggparen}[1]{\biggl({#1}\biggr)}

\newcommand{\bigbracket}[1]{\bigl[{#1}\bigr]}
\newcommand{\Bigbracket}[1]{\Bigl[{#1}\Bigr]}
\newcommand{\biggbracket}[1]{\biggl[{#1}\biggr]}

\newcommand{\set}[1]{\{{#1}\}}

\newcommand{\Bigset}[1]{\Bigl\{{#1}\Bigr\}}

\newcommand{\norm}[1]{\|{#1}\|}

\newcommand{\ip}[2]{\langle{#1},{#2}\rangle}

\newtheorem{theorem}{Theorem}[section]

\newtheorem{corollary}[theorem]{Corollary}

\theoremstyle{definition}
\newtheorem{definition}[theorem]{Definition}
\newtheorem{example}[theorem]{Example}
\newtheorem{remark}[theorem]{Remark}

\begin{document}
\begin{frontmatter}
\title{Harmonic Grassmannian codes}

\author[AFIT]{Matthew Fickus}
\ead{Matthew.Fickus@gmail.com}
\author[IA]{Joseph W.\ Iverson}
\author[SDSU]{John Jasper}
\author[OSU]{Dustin G.\ Mixon}

\address[AFIT]{Department of Mathematics and Statistics, Air Force Institute of Technology, Wright-Patterson AFB, OH 45433}
\address[IA]{Department of Mathematics, Iowa State University, Ames, IA 50011}
\address[SDSU]{Department of Mathematics and Statistics, South Dakota State University, Brookings, SD 57007}
\address[OSU]{Department of Mathematics, The Ohio State University, Columbus, OH 43210}

\begin{abstract}
An equi-isoclinic tight fusion frame (EITFF) is a type of Grassmannian code,
being a sequence of subspaces of a finite-dimensional Hilbert space of a given dimension with the property that the smallest spectral distance between any pair of them is as large as possible.
EITFFs arise in compressed sensing, yielding dictionaries with minimal block coherence.
Their existence remains poorly characterized.
Most known EITFFs have parameters that match those of one that arose from an equiangular tight frame (ETF) in a rudimentary, direct-sum-based way.
In this paper, we construct new infinite families of non-``tensor-sized" EITFFs in a way that generalizes the one previously known infinite family of them as well as the celebrated equivalence between harmonic ETFs and difference sets for finite abelian groups.
In particular, we construct EITFFs consisting of $Q$ planes in $\mathbb{C}^Q$ for each prime power $Q\geq 4$, of $Q-1$ planes in $\mathbb{C}^Q$ for each odd prime power $Q$, and of $11$ three-dimensional subspaces in $\mathbb{R}^{11}$.
The key idea is that every harmonic EITFF---one that is the orbit of a single subspace under the action of a unitary representation of a finite abelian group---arises from a smaller tight fusion frame with a nicely behaved ``Fourier transform."
Our particular constructions of harmonic EITFFs exploit the properties of Gauss sums over finite fields.
\end{abstract}

\end{frontmatter}

\section{Introduction}

Let $D$, $N$ and $R$ be integers with $1\leq R\leq D\leq NR$ and $N>1$.
For each $n=1,\dotsc,N$,
let $\bfPhi_n$ be any $D\times R$ complex isometry,
that is, any $D\times R$ complex matrix with $\bfPhi_n^*\bfPhi_n^{}=\bfI$,
meaning its columns form an orthonormal basis for some $R$-dimensional subspace $\calU_n$ of $\bbC^D$.
Any such sequence \smash{$\set{\bfPhi_n}_{n=1}^{N}$} satisfies a generalized \textit{Welch bound}~\cite{Welch74,ConwayHS96,StrohmerH03,DhillonHST08,CalderbankTX15}:
\begin{equation}
\setlength{\abovedisplayskip}{5pt}
\label{eq.Welch bound}
\smash{\max_{n_1\neq n_2}\norm{\bfPhi_{n_1}^*\bfPhi_{n_2}^{}}_2
\geq\tfrac1{\sqrt{R}}\max_{n_1\neq n_2}\norm{\bfPhi_{n_1}^*\bfPhi_{n_2}^{}}_\Fro
\geq\bigbracket{\tfrac{NR-D}{D(N-1)}}^{\frac12}.}
\end{equation}
Equality in the right-hand inequality of~\eqref{eq.Welch bound} holds if and only if $\set{\calU_n}_{n=1}^N$ is an \textit{equichordal tight fusion frame} (ECTFF) for $\bbC^D$,
namely a \textit{tight fusion frame} (TFF) for $\bbC^D$, meaning $\sum_{n=1}^N\bfPhi_n^{}\bfPhi_n^*=A\bfI$ for some $A>0$, that is \textit{equichordal} (EC) in the sense that $\norm{\bfPhi_{n_1}^*\bfPhi_{n_2}^{}}_\Fro$ is constant over all $n_1\neq n_2$.
Such $\set{\calU_n}_{n=1}^N$ are precisely those that achieve equality in Conway, Hardin and Sloane's \textit{simplex bound}~\cite{ConwayHS96},
and are \textit{chordal Grassmannian codes},
that is, sequences of $N$ members of the \textit{Grassmannian} that consists of all $R$-dimensional subspaces of $\bbC^D$ whose smallest pairwise \textit{chordal distance}
$\dist_\rmc(\calU_{n_1},\calU_{n_2})
:=(R-\norm{\bfPhi_{n_1}^*\bfPhi_{n_2}^{}}_\Fro^2)^{\frac12}$ is as large as possible.
Meanwhile, equality holds throughout~\eqref{eq.Welch bound} if and only if $\set{\calU_n}_{n=1}^{N}$ is an \textit{equi-isoclinic tight fusion frame} (EITFF) for $\bbC^D$,
namely a TFF for $\bbC^D$ that is \textit{equi-isoclinic} (EI),
meaning that the singular values $\set{\cos(\theta_{n_1,n_2,r})}_{r=1}^{R}$ of  $\bfPhi_{n_1}^*\bfPhi_{n_2}^{}$ are constant over all $n_1\neq n_2$ and all $r$.
Every EITFF is an ECTFF, and is moreover a Grassmannian code with respect to the \textit{spectral distance}
\smash{$\dist_\rms(\calU_{n_1},\calU_{n_2})
:=(1-\norm{\bfPhi_{n_1}^*\bfPhi_{n_2}^{}}_2^2)^{\frac12}$}~\cite{DhillonHST08}.

EITFFs are well-suited for \textit{block compressed sensing},
yielding isometries $\set{\bfPhi_n}_{n=1}^{N}$ with minimal \textit{block coherence},
and so providing the strongest guarantee via~\cite{EldarKB10} on the recovery of a sparse linear combination of members of $\set{\calU_n}_{n=1}^{N}$ using \textit{block orthogonal matching pursuit}, for example.
Because of this, a lot of recent work on EITFFs has focused on the \textit{existence problem}: for what $(D,N,R)$ does there exist an $\EITFF(D,N,R)$, namely an EITFF for $\bbC^D$ that consists of $N$ subspaces of dimension $R$?
Some basic facts here are as follows: if an $\EITFF(D,N,R)$ exists, then $R\leq D\leq NR$;
conversely, if $R=D$ or $D=NR$ then a trivial $\EITFF(D,N,R)$ exists;
every $\EITFF(D,N,R)$ with $D<NR$ has a \textit{Naimark complement} which is an $\EITFF(NR-D,N,R)$;
directly summing an $\EITFF(D_1,N,R_1)$ and an $\EITFF(D_2,N,R_2)$ with \smash{$\frac{D_1}{R_1}=\frac{D_2}{R_2}$} yields an $\EITFF(D_1+D_2,N,R_1+R_2)$.
Due to this last fact, $\EITFF(D,N,R)$ in which $D$ and $R$ are coprime are particularly noteworthy.
In particular, every $\EITFF(D,N,R)$ with $R=1$ equates to an $N$-vector \textit{equiangular tight frame} (ETF) for $\bbC^D$,
and these have been the subject of much attention;
see~\cite{FickusM16} for a survey.
To date, remarkably few constructions of $\EITFF(D,N,R)$ with $\gcd(D,R)=1$ and $R>1$ are known:
a real $\EITFF(N,N,2)$ exists whenever a complex symmetric conference matrix of size $N$ does, including whenever $N$ is a prime power with $N\equiv1\bmod 4$~\cite{EtTaoui18,BlokhuisBE18},
and a real $\EITFF(8,6,3)$ exists~\cite{EtTaouiR09,IversonKM21}.
In fact, most known $\EITFF(D,N,R)$ are ``tensor-sized,"
having parameters that match those of one obtained by directly summing $R$ copies of an $N$-vector ETF for a space of dimension $\frac{D}{R}$~\cite{FickusMW21}.

In this paper, we generalize the celebrated equivalence~\cite{Konig99,StrohmerH03,XiaZG05,DingF07} between \textit{harmonic} ETFs and \textit{difference sets} for finite abelian groups so as to produce an infinite number of new EITFFs, including an infinite number of new $\EITFF(N,N,2)$ with odd $N$,
and an $\EITFF(11,11,3)$.
A finite frame is \textit{harmonic} with respect to the Pontryagin dual $\widehat{\calG}$ of some finite abelian group $\calG$ if it has any one of three equivalent properties~\cite{Waldron18}:
(i) it arises by extracting rows from the $\calG\times\widehat{\calG}$ character table of $\calG$;
(ii) it has a $\widehat{\calG}$-circulant Gram matrix~\cite{ValeW04};
(iii) it is the orbit of a single vector under the action of a unitary representation of $\widehat{\calG}$~\cite{ValeW04}.
As detailed in Section~3,
there is a straightforward generalization of (ii) and (iii) to the fusion-frame context, where any TFF with a $\widehat{\calG}$-block-circulant Gram matrix arises as the orbit of $R$ orthonormal vectors under such an action.
The corresponding generalization of (i) is less intuitive,
with one's choice of rows generalizing to some $\calG$-indexed TFF for $\bbC^R$ whose subspaces are of perhaps mixed dimension; see Definition~\ref{def.harmonic matrix ensemble} and Theorems~\ref{thm.small TFF yields harmonic} and~\ref{thm.harmonic TFF characterization}.
In essence, we find that every \textit{harmonic TFF} (Definition~\ref{def.harmonic}) is generated by some other TFF with smaller parameters.
In Theorem~\ref{thm.EITFF},
we moreover characterize when a harmonic TFF is either equichordal or equi-isoclinic.
In particular, harmonic EITFFs arise from generating TFFs whose projections have nicely behaved operator-valued Fourier transforms (Definition~\ref{def.difference projections}).
In Section~4, we then exploit this characterization to construct new infinite families of harmonic $\EITFF(N-1,N,2)$ (Theorem~\ref{thm.EITFF(Q-1,Q,2)}),
$\EITFF(N,N,2)$ (Theorem~\ref{thm.EITFF(Q,Q,2)}),
as well as a harmonic $\EITFF(11,11,3)$ (Theorem~\ref{thm.EITFF(11,11,3)}).
In the fifth section we obtain ``combinatorial" (conference-matrix-based) generalizations of some of these harmonic EITFFs; see Theorem~\ref{thm.combinatorial}.
We find that some of the EITFFs of~\cite{EtTaoui18} are harmonic,
and help motivate the ongoing search for symmetric and skew-symmetric complex conference matrices~\cite{EtTaouiM21}; see Theorem~\ref{thm.conference}.
We begin in the next section by introducing notation and explaining some known facts that we use later on.

\section{Preliminaries}

For the sake of simplicity, we sometimes consider vectors, matrices, etc., whose entries are indexed by members of arbitrary finite sets.
In particular, for any finite $N$-element set $\calN$,
we equip $\bbC^\calN:=\set{\bfx:\calN\rightarrow\bbC}$ with the inner product
$\ip{\bfx_1}{\bfx_2}:=\sum_{n\in\calN}\overline{\bfx_1(n)}\bfx_2(n)$,
and denote its standard basis as $\set{\bfdelta_n}_{n\in\calN}$ where $\bfdelta_n(n):=1$ and $\bfdelta_n(n'):=0$ whenever $n'\neq n$.
(All of our complex inner products are conjugate-linear in their first arguments.)
For any finite sets $\calM$ and $\calN$,
we identify a linear operator $\bfA:\bbC^\calN\rightarrow\bbC^\calM$ with a member of $\bbC^{\calM\times\calN}$ in the usual way,
being the ``$\calM\times\calN$ matrix" $\bfA$ with entries $\bfA(m,n)=\ip{\bfdelta_m}{\bfA\bfdelta_n}$ for all $m\in\calM$ and $n\in\calN$.
The adjoint $\bfA^*$ of such a matrix $\bfA$ is its $\calN\times\calM$ conjugate transpose.

The \textit{synthesis operator} of vectors $\set{\bfphi_n}_{n\in\calN}$ in a $D$-dimensional complex Hilbert space $\calV$ is $\bfPhi:\bbC^\calN\rightarrow\calV$,
$\bfPhi\bfx:=\sum_{n\in\calN}\bfx(n)\bfphi_n$.
The image of such an operator $\bfPhi$ is $\bfPhi(\bbC^\calD)=\Span\set{\bfphi_n}_{n\in\calN}$,
and its adjoint is the corresponding \textit{analysis operator} $\bfPhi^*:\calV\rightarrow\bbC^\calN$,
$(\bfPhi^*\bfy)(n)=\ip{\bfphi_n}{\bfy}$.
We sometimes identify a single vector $\bfphi\in\calV$ with the synthesis operator $\bfphi:\bbC\rightarrow\calV$, $\bfphi(z):=z\bfphi$ whose adjoint is the linear functional $\bfphi^*:\calV\rightarrow\bbC$, $\bfphi^*\bfy=\ip{\bfphi}{\bfy}$.
Composing the synthesis and analysis operators of $\set{\bfphi_n}_{n\in\calN}$ gives the \textit{frame operator} $\bfPhi\bfPhi^*:\calV\rightarrow\calV$,
$\bfPhi\bfPhi^*=\sum_{n\in\calN}\bfphi_n^{}\bfphi_n^*$ and the $\calN\times\calN$ \textit{Gram matrix} $\bfPhi^*\bfPhi$ whose $(n_1,n_2)$th entry is $(\bfPhi^*\bfPhi)(n_1,n_2)=\ip{\bfphi_{n_1}}{\bfphi_{n_2}}$.
In particular, $\bfPhi$ is an isometry (i.e., satisfies $\bfPhi^*\bfPhi=\bfI$) if and only if $\set{\bfphi_n}_{n\in\calN}$ is orthonormal,
and in this case $\bfPhi\bfPhi^*$ is the orthogonal projection operator onto the image of $\bfPhi$; in the sequel, we use \textit{projection} synonymously with \textit{orthogonal projection operator}.
As a degenerate case,
$\bbC^\emptyset$ is the singleton vector space $\set{\bfzero}$;
there is thus a unique linear operator $\bfPhi:\bbC^{\emptyset}\rightarrow\calV$,
and it is an isometry: $\bfPhi^*\bfPhi$ is the identity operator on $\set{\bfzero}$ and $\bfPhi\bfPhi^*$ is the rank-$0$ projection (zero operator) on $\calV$.
When $\calV=\bbC^\calD$ for some finite set $\calD$,
$\bfPhi$ is the $\calD\times\calN$ matrix that has $\bfphi_n$ as its $n$th column,
$\bfPhi^*$ is its $\calN\times\calD$ conjugate transpose,
and $\bfPhi\bfPhi^*$ and $\bfPhi^*\bfPhi$ are their $\calD\times\calD$ and $\calN\times\calN$ products, respectively.

Vectors $\set{\bfphi_n}_{n\in\calN}$ in $\calV$ are \textit{unit norm} if $\norm{\bfphi_n}=1$ for all $n$,
and form a \textit{tight frame} for $\calV$ if
$\bfPhi\bfPhi^*
=A\bfI$ for some $A>0$.
If both of these properties hold then they form a \textit{unit norm tight frame} (UNTF) for $\calV$, and taking the rank and trace of $\sum_{n\in\calN}\bfphi_n^{}\bfphi_n^*
=A\bfI$ gives that necessarily $A=\frac ND\geq1$.
An $\calN\times\calN$ matrix $\bfG$ is the Gram matrix $\bfPhi^*\bfPhi$ of some UNTF $\set{\bfphi_n}_{n\in\calN}$ for some complex Hilbert space $\calV$ if and only if each of its diagonal entries is $1$ and $\frac1A\bfG$ is a projection for some $A>0$; in this case, $D:=\dim(\calV)=\frac{N}{A}$.
If, in this case, we further have that $D\neq N$ then $\frac{A}{A-1}(\bfI-\tfrac1A\bfPhi^*\bfPhi)$ has these same two properties,
meaning that it is the Gram matrix \smash{$(\bfPhi^\flat)^*\bfPhi^\flat$} of some $N$-vector UNTF \smash{$\set{\bfphi_n^\flat}_{n\in\calN}$} for some $(N-D)$-dimensional complex Hilbert space.
When this occurs, $\set{\bfphi_n}_{n\in\calN}$ and $\set{\bfphi_n^\flat}_{n\in\calN}$ are \textit{Naimark complements} of each other.

\subsection{Tight fusion frames (TFFs)}

Now let $\set{\calR_n}_{n\in\calN}$ be a sequence of finite sets with $R_n:=\#(\calR_n)\geq0$ for all $n$.
For each $n\in\calN$, let $\bfPhi_n:\bbC^{\calR_n}\rightarrow\calV$ be an isometry,
namely the synthesis operator of some orthonormal basis \smash{$\set{\bfphi_{n,r}}_{r\in\calR_n}$} of some $R_n$-dimensional subspace \smash{$\calU_n:=\bfPhi_n(\bbC^{\calR_n})=\Span\set{\bfphi_{n,r}}_{r\in\calR_n}$} of $\calV$.
The \textit{fusion} synthesis $\bfPhi$, analysis $\bfPhi^*$, frame $\bfPhi\bfPhi^*$ and Gram $\bfPhi^*\bfPhi$ operators of $\set{\bfPhi_n}_{n\in\calN}$ are the (traditional) synthesis, analysis, frame and Gram operators of the concatenation $\set{\bfphi_{n,r}}_{(n,r)\in\calN\times\set{\calR_n}}$ of these bases, respectively;
here, ``$\calN\times\set{\calR_n}$" is shorthand notation for the set
$\set{(n,r): n\in\calN, r\in\calR_n}=\cup_{n\in\calN}(\set{n}\times\calR_n)$ of all pairs of such indices.
In particular,
$\bfPhi\bfPhi^*
=\sum_{(n,r)\in\calN\times\set{\calR_n}}\bfphi_{n,r}^{}\bfphi_{n,r}^*
=\sum_{n\in\calN}\bfPhi_n^{}\bfPhi_n^*$.
Meanwhile,
$\bfPhi^*\bfPhi$
has an $((n_1,r_1),(n_2,r_2))$th entry of
\begin{equation*}
(\bfPhi^*\bfPhi)((n_1,r_1),(n_2,r_2))
=\ip{\bfphi_{n_1,r_1}}{\bfphi_{n_2,r_2}}
=\ip{\bfPhi_{n_1}\bfdelta_{r_1}}{\bfPhi_{n_2}\bfdelta_{r_2}}
=(\bfPhi_{n_1}^*\bfPhi_{n_2}^{})(r_1,r_2),
\end{equation*}
and so is naturally regarded as an $\calN\times\calN$ array whose $(n_1,n_2)$th submatrix is the $\calR_{n_1}\times\calR_{n_2}$ \textit{cross-Gram} matrix
$\bfPhi_{n_1}^*\bfPhi_{n_2}^{}$.
We say isometries are \textit{equivalent} to $\set{\bfPhi_n}_{n\in\calN}$ if they are of the form $\set{\bfU\bfPhi_n\bfV_n}_{n\in\calN}$ for some unitary operator $\bfU$ on $\calV$ and some $\calR_n\times\calR_n$ unitaries $\set{\bfV_n}_{n\in\calN}$;
their fusion frame operator is $\bfU\bfPhi\bfPhi^*\bfU^*$,
while their fusion Gram matrix is obtained by conjugating $\bfPhi^*\bfPhi$ by a block-diagonal matrix whose $n$th diagonal block is $\bfV_n$.
We say isometries $\set{\bfPhi_n}_{n\in\calN}$ are \textit{real} if they are equivalent to isometries whose fusion Gram matrix has all real entries;
in this case we can without loss of generality regard $\set{\calU_n}_{n\in\calN}$ to be subspaces of $\bbR^D$ where $D=\rank(\bfPhi)$.

We say that such isometries $\set{\bfPhi_n}_{n\in\calN}$ form a $\TFF(D,N,\set{R_n}_{n\in\calN})$ for $\calV$ if $\bfPhi\bfPhi^*=A\bfI$ for some $A>0$.
This property is preserved by equivalence,
and so is often regarded as a statement about $\set{\calU_n}_{n\in\calN}$ rather than $\set{\bfPhi_n}_{n\in\calN}$.
Regardless, it holds if and only if the concatenation $\set{\bfphi_{n,r}}_{(n,r)\in\calN\times\set{\calR_n}}$ of some orthonormal bases for $\set{\calU_n}_{n\in\calN}$ is a UNTF for $\calV$.
As such, we necessarily have $1\leq A=\tfrac1D\sum_{n\in\calN}R_n\leq N$ and moreover that a matrix $\bfG$ whose rows and columns are indexed by $\calN\times\set{\calR_n}$ is the fusion Gram matrix $\bfPhi^*\bfPhi$ of such a TFF for some complex Hilbert space $\calV$ if and only if $\frac1A\bfG$ is a projection for some $A>0$ and, for each $n$, the $n$th diagonal block of $\bfG$ is the $\calR_n\times\calR_n$ identity matrix; in this case, \smash{$D:=\dim(\calV)=\tfrac1A\sum_{n\in\calN}R_n$}.
From this, it follows that the isometries $\set{\bfPhi_n}_{n\in\calN}$ of any $\TFF(D,N,\set{R_n}_{n\in\calN})$ with $D<\sum_{n\in\calN}R_n$ have a Naimark complement whose isometries \smash{$\set{\bfPhi_n^\flat}_{n\in\calN}$} form a \smash{$\TFF(\sum_{n\in\calN}R_n-D,N,\set{R_n}_{n\in\calN})$}.
When $\set{\bfPhi_n}_{n\in\calN}$ are alternatively the isometries of a $\TFF(D,N,\set{R_n}_{n\in\calN})$ for $\calV$ with $R_n<D$ for at least one $n$,
then $A<N$ and any isometries $\set{\bfPhi_n^\sharp}_{n\in\calN}$ for $\set{\calU_n^\perp}_{n\in\calN}$ satisfy
$\sum_{n\in\calN}\bfPhi_n^\sharp(\bfPhi_n^\sharp)^*
=\sum_{n\in\calN}(\bfI-\bfPhi_n^{}\bfPhi_n^*)
=(N-A)\bfI$ and so form a \textit{spatial complementary} $\TFF(D,N,\set{D-R_n}_{n\in\calN})$ for $\calV$.

One can also directly sum some TFFs to form others provided their parameters are consistent in a certain way.
To explain,
let \smash{$\set{\calD^{(j)}}_{j\in\calJ}$} be a partition of some $D$-element set $\calD$ into nonempty subsets,
and, for each $n\in\calN$, let
\smash{$\set{\calR_n^{(j)}}_{j\in\calJ}$} be a partition of $\calR_n$.
For each $j\in\calJ$,
also let
\smash{$\set{\bfPhi_n^{(j)}}_{n\in\calN}$} be  \smash{$\set{\calD^{(j)}\times\set{\calR_n^{(j)}}}_{n\in\calN}$} isometries, respectively, of some \smash{$\TFF(D^{(j)},N,\set{R_n^{(j)}}_{n\in\calN})$} for \smash{$\bbC^{\calD^{(j)}}$}.
For each $n\in\calN$, let $\bfPhi_n$ be a $\calD\times\calR_n$ block-diagonal matrix whose $j$th diagonal block is \smash{$\bfPhi_n^{(j)}$}.
Then each cross-Gram matrix $\bfPhi_{n_1}^*\bfPhi_{n_2}^{}$ is a block-diagonal matrix whose $j$th diagonal block is \smash{$[\bfPhi_{n_1}^{(j)}]^*\bfPhi_{n_2}^{(j)}$}.
In particular, each $\bfPhi_n$ is an isometry.
Moreover, the fusion frame operator $\sum_{n\in\calN}\bfPhi_n^{}\bfPhi_n^*$ of these isometries is also a block diagonal matrix whose $j$th diagonal block is \smash{$\sum_{n\in\calN}\bfPhi_n^{(j)}[\bfPhi_n^{(j)}]^*=A^{(j)}\bfI$} where
\smash{$A^{(j)}=\frac1{D^{(j)}}\sum_{n\in\calN}R_n^{(j)}$}.
In particular, $\set{\bfPhi_n}_{n\in\calN}$ are the isometries of a $\TFF(D,N,\set{R_n}_{n\in\calN})$ if and only if $A^{(j)}=A$ for all $j\in\calJ$,
that is, if and only if
\begin{equation}
\label{eq.direct sum consistency}
\tfrac1{D^{(j)}}\sum_{n\in\calN}R_n^{(j)}
=\tfrac1{D}\sum_{n\in\calN}R_n,
\quad\forall\, j\in\calJ.
\end{equation}

\subsection{Equichordal (EC) and equi-isoclinic (EI) tight fusion frames (TFFs)}

We say isometries $\set{\bfPhi_n}_{n\in\calN}$, $\bfPhi_n:\bbC^{\calR_n}\rightarrow\calV$ are \textit{uniform} if there is some set $\calR$ such that $\calR_n=\calR$ for all $n$.
In this case, $\calN\times\set{\calR_n}=\calN\times\calR$.
If uniform isometries form a $\TFF(D,N,\set{R_n}_{n\in\calN})$ for $\calV$ we often call it simply a $\TFF(D,N,R)$ for $\calV$ where $R=\#(\calR)$.
In this paper, we only consider nonuniform TFFs as a ``means to an end" for constructing uniform ones.
For any uniform isometries $\set{\bfPhi_n}_{n\in\calN}$ onto $R$-dimensional subspaces $\set{\calU_n}_{n\in\calN}$ of the $D$-dimensional complex Hilbert space $\calV$, cycling traces gives
\begin{equation*}
0\leq\norm{\bfPhi\bfPhi^*-\tfrac{NR}{D}\bfI}_{\Fro}^2
=\Tr[(\bfPhi^*\bfPhi)^2]-\tfrac{N^2R^2}{D}
=\sum_{n_1\in\calN}\sum_{n_2\neq n_1}\norm{\bfPhi_{n_1}^*\bfPhi_{n_2}^{}}_\Fro^2
-\tfrac{NR(NR-D)}{D}.
\end{equation*}
In particular, any such $\set{\bfPhi_n}_{n\in\calN}$ satisfy the following inequality,
and moreover achieve equality in it if and only if they form a $\TFF(D,N,R)$ for $\calV$:
\begin{equation}
\label{eq.proto Welch}
\tfrac{NR-D}{D(N-1)}
\leq\tfrac{1}{N(N-1)}\sum_{n_1\in\calN}
\sum_{n_2\neq n_1}\tfrac1R\norm{\bfPhi_{n_1}^*\bfPhi_{n_2}^{}}_\Fro^2.
\end{equation}
Here, for any $n_1\neq n_2$,
$\norm{\bfPhi_{n_1}^*\bfPhi_{n_2}^{}}_2\leq\norm{\bfPhi_{n_1}}_2\norm{\bfPhi_{n_2}^{}}_2=1$,
and so the singular values of $\bfPhi_{n_1}^*\bfPhi_{n_2}^{}$ can be expressed as $\set{\cos(\theta_{n_1,n_2,r})}_{r=1}^R$ for some nondecreasing sequence \smash{$\set{\theta_{n_1,n_2,r}}_{r=1}^R\subset[0,\frac\pi2]$} of \textit{principal angles} between $\calU_{n_1}$ and $\calU_{n_2}$.
In particular,
$\norm{\bfPhi_{n_1}^*\bfPhi_{n_2}^{}}_\Fro^2
=\sum_{r=1}^R\cos^2(\theta_{n_1,n_2,r})$,
and the chordal distance~\cite{ConwayHS96} between $\calU_{n_1}$ and $\calU_{n_2}$ is the (scaled) Frobenius distance between their projections:
\begin{equation}
\label{eq.chordal distance}
\dist_\rmc(\calU_{n_1},\calU_{n_2})
:=\tfrac1{\sqrt{2}}\norm{\bfPhi_{n_1}^{}\bfPhi_{n_1}^{*}-\bfPhi_{n_2}^{}\bfPhi_{n_2}^{*}}_\Fro
=(R-\norm{\bfPhi_{n_1}^*\bfPhi_{n_2}^{}}_\Fro^2)^{\frac12}
=\biggparen{\,\sum_{r=1}^R\sin^2(\theta_{n_1,n_2,r})}^{\frac12}.
\end{equation}
Bounding the summands in~\eqref{eq.proto Welch} by their largest member gives the right-hand inequality of~\eqref{eq.Welch bound}:
\begin{equation}
\label{eq.chordal Welch}
\tfrac{NR-D}{D(N-1)}
\leq\tfrac{1}{N(N-1)}\sum_{n_1\in\calN}
\sum_{n_2\neq n_1}\tfrac1R\norm{\bfPhi_{n_1}^*\bfPhi_{n_2}^{}}_\Fro^2
\leq\tfrac1R\max_{n_1\neq n_2}\norm{\bfPhi_{n_1}^*\bfPhi_{n_2}^{}}_\Fro^2.
\end{equation}
Isometries $\set{\bfPhi_n}_{n\in\calN}$ achieve equality throughout~\eqref{eq.chordal Welch} if and only if they form a $\TFF(D,N,R)$ for $\calV$ in which~\eqref{eq.chordal distance} is constant over all $n_1\neq n_2$,
namely an $\ECTFF(D,N,R)$ for $\calV$.
That is, a $\TFF(D,N,R)$ is an $\ECTFF(D,N,R)$ if and only if all of the off-diagonal blocks of its fusion Gram matrix have the same Frobenius norm.
Meanwhile,
the left-hand inequality of~\eqref{eq.Welch bound} arises from the fact that $\frac1R\norm{\bfPhi_{n_1}^*\bfPhi_{n_2}^{}}_\Fro^2
\leq\norm{\bfPhi_{n_1}^*\bfPhi_{n_2}^{}}_2^2$
for any $n_1,n_2\in\calN$,
with equality holding if and only if $\calU_{n_1}$ and $\calU_{n_2}$ are \textit{isoclinic} in the sense that $\theta_{n_1,n_2,r}$ is independent of $r$.
This relates to the spectral distance~\cite{DhillonHST08} between $\calU_{n_1}$ and $\calU_{n_2}$:
\begin{equation}
\label{eq.spectral distance}
\dist_\rms(\calU_{n_1},\calU_{n_2})
:=(1-\norm{\bfPhi_{n_1}^*\bfPhi_{n_2}^{}}_2^2)^{\frac12}
=[1-\max_{r}\cos^2(\theta_{n_1,n_2,r})]^{\frac12}
=\min_{r}\sin(\theta_{n_1,n_2,r}).
\end{equation}
In particular, one may continue~\eqref{eq.proto Welch} in the following way,
which is alternative to~\eqref{eq.chordal Welch}:
\begin{equation*}
\tfrac{NR-D}{D(N-1)}
\leq\tfrac{1}{N(N-1)}\sum_{n_1\in\calN}
\sum_{n_2\neq n_1}\tfrac1R\norm{\bfPhi_{n_1}^*\bfPhi_{n_2}^{}}_\Fro^2
\leq\tfrac{1}{N(N-1)}\sum_{n_1\in\calN}
\sum_{n_2\neq n_1}\norm{\bfPhi_{n_1}^*\bfPhi_{n_2}^{}}_2^2
\leq\max_{n_1\neq n_2}\norm{\bfPhi_{n_1}^*\bfPhi_{n_2}^{}}_2^2.
\end{equation*}
Equality holds throughout~\eqref{eq.Welch bound} if and only if it holds throughout here,
and this occurs if and only if $\set{\bfPhi_n}_{n\in\calN}$ forms a $\TFF(D,N,R)$ for $\calV$ with the property that $\theta_{n_1,n_2,r}$ is independent of both $n_1\neq n_2$ and $r$, namely an $\EITFF(D,N,R)$ for $\calV$.
That is, a $\TFF(D,N,R)$ is an $\EITFF(D,N,R)$ if and only if all of the off-diagonal blocks of its fusion Gram matrix are a common scalar multiple of some unitaries.
When rewritten in terms of~\eqref{eq.chordal distance} and~\eqref{eq.spectral distance},
\eqref{eq.Welch bound} equivalently becomes the following extension of Conway, Hardin and Sloane's \textit{simplex bound}:
\begin{equation*}
\min_{n_1\neq n_2}\dist_\rms(\calU_{n_1},\calU_{n_2})
\leq\tfrac1{\sqrt{R}}\min_{n_1\neq n_2}\dist_\rmc(\calU_{n_1},\calU_{n_2})
\leq\sqrt{\tfrac{N}{N-1}\tfrac{D-R}{D}}.
\end{equation*}
In particular, every $\ECTFF(D,N,R)$ is a chordal Grassmannian code,
while every $\EITFF(D,N,R)$ is a spectral Grassmannian code,
is necessarily an $\ECTFF(D,N,R)$,
and has minimal \textit{block coherence}
$\min_{n_1\neq n_2}\tfrac1R\norm{\bfPhi_{n_1}^*\bfPhi_{n_2}^{}}_2$
(and so is well-suited for sensing block-sparse signals~\cite{EldarKB10}).

Since EITFFs are more applicable---and rarer---than ECTFFs,
we tend to prioritize their consideration throughout this paper.
There are two trivial types of $\EITFF(D,N,R)$:
those with $D=R$ and those with $D=NR$:
the former consist of $N$ copies of the entire space,
while the isometries of the latter are formed by partitioning an orthonormal basis for the entire space into $N$ subcollections of $R$ vectors apiece.
There are also at least two trivial ways to construct some EITFFs from other ones,
namely via Naimark complements and via direct sums.
In fact,
by exploiting the aforementioned fusion-Gram-matrix-based characterizations of such objects,
it is straightforward to show that a $\TFF(D,N,R)$ with $D<NR$ is real and/or equichordal and/or equi-isoclinic if and only if its Naimark complement is as well.
Similarly, if a $\TFF(D,N,R)$ is the direct sum of $\TFF(D^{(j)},N,R^{(j)})$ over all $j\in\calJ$---which by~\eqref{eq.direct sum consistency} can only happen if \smash{$\frac{R^{(j)}}{D^{(j)}}=\frac{R}{D}$} for all $j\in\calJ$---then this $\TFF(D,N,R)$ is real and/or equichordal and/or equi-isoclinic if each $\TFF(D^{(j)},N,R^{(j)})$ has that same property for all $j\in\calJ$.
We caution however that the spatial complement of an $\EITFF(D,N,R)$ with $R<D$ is not an $\EITFF(D,N,R)$ unless $R=\frac{D}{2}$.
In fact, if $R<D$ then an $\EITFF(D,N,R)$ can only exist if $R\leq\frac D2$,
since otherwise at least one but not all of the principal angles between any two subspaces of an $\ECTFF(D,N,R)$ are necessarily $0$.
That said, for any $\TFF(D,N,R)$, \eqref{eq.chordal distance} implies that $\dist_\rmc(\calU_{n_1}^\perp,\calU_{n_2}^\perp)
=\dist_\rmc(\calU_{n_1},\calU_{n_2})$ for all $n_1\neq n_2$,
and so this TFF is real and/or equichordal if and only if its spatial complement is as well; see~\cite{FickusMW21} for a more thorough discussion of these well-known phenomena.

Beyond this,
it is known that entrywise applying the mapping $x+\rmi y\mapsto[\begin{smallmatrix}x&-y\\y&\phantom{-}x\end{smallmatrix}]$ to a synthesis matrix of a complex $\operatorname{(EC/EI)TFF}(D,N,R)$ converts it into a real $\operatorname{(EC/EI)TFF}(2D,N,2R)$~\cite{Hoggar77}.
Moreover, applying a similar mapping to a quaternionic $\operatorname{(EC/EI)TFF}(D,N,R)$ converts it into a complex $\operatorname{(EC/EI)TFF}(2D,N,2R)$~\cite{Hoggar77,Waldron20}.
The community is particularly interested in EITFFs whose parameters do not match those of one that can be produced from other known EITFFs via any combination of these \textit{Hoggar-type} methods with Naimark complements and direct sums.
Along these lines,
$\EITFF(D,N,R)$ in which $D$ and $R$ are coprime are particularly sought after.
These include all $\EITFF(D,N,R)$ in which $R=1$:
in this case, a rank-$R$ isometry $\bfPhi_n$ equates to a unit vector $\bfphi_n$,
\eqref{eq.Welch bound} reduces to Welch's classical lower bound~\cite{Welch74} on the coherence of unit vectors $\set{\bfphi_n}_{n\in\calN}$, namely
\smash{$\max_{n_1\neq n_2}\abs{\ip{\bfphi_{n_1}}{\bfphi_{n_2}}}
\geq\bigbracket{\tfrac{N-D}{D(N-1)}}^{\frac12}$},
and they achieve equality in it precisely when they form an $\ETF(D,N)$ for $\bbC^\calD$,
namely a UNTF with the additional property that $\abs{\ip{\bfphi_{n_1}}{\bfphi_{n_2}}}$ is constant over all $n_1\neq n_2$.
To date, the only apparently known $\EITFF(D,N,R)$ with $R>1$ whose parameters are not reproducible by applying some combination of Naimark, direct sum and Hoggar-type methods to known (possibly quaternionic~\cite{CohnKM16,EtTaoui20,Waldron20}) ETFs or trivial EITFFs are those of~\cite{LemmensS73b}, \cite{EtTaoui18,BlokhuisBE18}, \cite{EtTaouiR09} and~\cite{IversonKM21}.
Specifically, \cite{LemmensS73b} gives a real $\EITFF(2R,\rmH(R)+2,R)$ whenever $R=2^{4J+K}$ for some nonnegative integers $J$ and $K\in\set{0,1,2,3}$ and $\rmH(R)=8J+2^K$ is its corresponding \textit{Hurwitz number},
\cite{EtTaoui18,BlokhuisBE18} constructs a real $\EITFF(N,N,2)$ from any complex symmetric conference matrix of size $N$, including whenever $N$ is a prime power such that $N\equiv1\bmod 4$,
\cite{EtTaouiR09} constructs a real $\EITFF(8,6,3)$,
and~\cite{IversonKM21} constructs a real $\EITFF(21,15,3)$ (and also recovers an $\EITFF(8,6,3)$).

\subsection{The signature matrix of an equi-isoclinic tight fusion frame}

For any isometries $\set{\bfPhi_n}_{n\in\calN}$ of an $\EITFF(D,N,R)$ with $D<NR$,
the corresponding \textit{signature matrix} $\bfS$ is $\bfS=\frac1{B}(\bfPhi^*\bfPhi-\bfI)$ where \smash{$B=\bigbracket{\tfrac{NR-D}{D(N-1)}}^{\frac12}>0$}.
This is a self-adjoint block matrix whose diagonal blocks are zero and whose off-diagonal blocks are unitaries.
Moreover, since $\bfPhi^*\bfPhi=\bfI+B\bfS$ where $(\bfPhi^*\bfPhi)^2=A\bfPhi^*\bfPhi$ with $A=\frac{NR}{D}$, we have
\begin{equation}
\label{eq.signature matrix tightness}
A(\bfI+B\bfS)
=(\bfI+B\bfS)^2
=\bfI+2B\bfS+B^2\bfS^2,
\quad\text{i.e.,}\quad
\bfS^2
=\tfrac{A-2}{B}\bfS+\tfrac{A-1}{B^2}\bfI.
\end{equation}
That is, $\bfS$ is an $\calN\times\calN$ array of $\calR\times\calR$ matrices $\set{\bfS_{n_1,n_2}}_{n_1,n_2\in\calN}$ that satisfies
\begin{enumerate}
\renewcommand{\labelenumi}{(\roman{enumi})}
\item $\bfS_{n,n}=\bfzero$ for all $n\in\calN$,
\item $\bfS_{n_1,n_2}^*=\bfS_{n_1,n_2}^{-1}=\bfS_{n_2,n_1}$ for $n_1,n_2\in\calN$, $n_1\neq n_2$,
\item $\bfS^2$ is a linear combination of $\bfS$ and $\bfI$.
\end{enumerate}
Conversely, any such $\bfS$ arises from an EITFF.
Indeed, having $\bfS_{n_2,n_1}=\bfS_{n_1,n_2}^*=(\bfS^*)_{n_2,n_1}$ for all $n_1,n_2\in\calN$ implies that $\bfS^*=\bfS$,
and so $\bfS$ is unitarily diagonalizable and has real eigenvalues.
Moreover, since $\Tr(\bfS)=0$ and $\bfS^2=C\bfS+E\bfI$ for some scalars $C$ and $E$,
$\bfS$ has exactly two distinct eigenvalues, one positive and the other negative,
which are roots of $\lambda^2-C\lambda-E$.
(In particular, $C$ and $E$ are necessarily real.)
Denoting its negative eigenvalue as $-F$ for some $F>0$,
we thus have that $\bfI+\frac1{F}\bfS$ is the Gram matrix $\bfPhi^*\bfPhi$ of some UNTF $\set{\bfphi_{n,r}}_{(n,r)\in\calN\times\calR}$ for some complex Hilbert space $\calV$ whose dimension $D$ is the multiplicity of the positive eigenvalue of $\bfS$.
For each $n\in\calN$,
letting $\bfPhi_n$ be the $\calN\times\calR$ matrix that has $\set{\bfphi_{n,r}}_{r\in\calR}$ as its columns,
we thus have that $\bfPhi_n^*\bfPhi_n^{}$ is the $n$th diagonal block of $\bfPhi^*\bfPhi=\bfI+\frac1{F}\bfS$, namely an $\calR\times\calR$ identity matrix.
As such, $\bfPhi^*\bfPhi$ is in fact the fusion Gram matrix of the isometries $\set{\bfPhi_n}_{n\in\calN}$ of an $\EITFF(D,N,R)$ for $\calV$.

In practice, this theory is often used as follows.
Let $\bfS$ be any $\calN\times\calN$ array of $\calR\times\calR$ matrices where $\bfS_{n,n}=\bfzero$ for all $n\in\calN$ and $\bfS_{n_1,n_2}$ is unitary for all $n_1\neq n_2$ (and is typically chosen from some explicitly parameterized subset of $\operatorname{O}(\calR)$ or $\operatorname{U}(\calR)$).
Then $\bfS$ is the signature matrix of some $\EITFF(D,N,R)$ if and only if both $\bfS_{n_1,n_2}=\bfS_{n_2,n_1}^*$ for all $n_1\neq n_2$ and $\bfS$ satisfies~\eqref{eq.signature matrix tightness} where \smash{$A=\frac{NR}{D}$} and \smash{$B=\bigbracket{\tfrac{NR-D}{D(N-1)}}^{\frac12}$}.
When written block-wise, this latter condition becomes
\begin{equation}
\label{eq.signature matrix tightness (blockwise)}
\sum_{n_3\in\calN}\bfS_{n_1,n_3}\bfS_{n_3,n_2}
=(\bfS^2)_{n_1,n_2}
=(\tfrac{A-2}{B}\bfS+\tfrac{A-1}{B^2}\bfI)_{n_1,n_2}
=\left\{\begin{array}{cl}
\tfrac{A-1}{B^2}\bfI_{\calR\times\calR},\smallskip&\ n_1=n_2,\\
\tfrac{A-2}{B}\bfS_{n_1,n_2},&\ n_1\neq n_2.
\end{array}\right.
\end{equation}
Here, when $\bfS_{n_1,n_2}=\bfS_{n_2,n_1}^*$ for all $n_1\neq n_2$,
the fact that \smash{$\frac{A-1}{B^2}=(\tfrac{NR}{D}-1)\tfrac{D(N-1)}{NR-D}=(N-1)$} implies that~\eqref{eq.signature matrix tightness (blockwise)} is automatically satisfied whenever $n_1=n_2$:
\begin{equation*}
\sum_{n_3\in\calN}\bfS_{n_1,n_3}\bfS_{n_3,n_1}
=\sum_{n_3\in\calN}\bfS_{n_1,n_3}\bfS_{n_1,n_3}^*
=\bfzero+\sum_{\substack{n_3\in\calN\\n_3\neq n_1}}\bfI_{\calR\times\calR}
=(N-1)\bfI_{\calR\times\calR}
=\tfrac{A-1}{B^2}\bfI_{\calR\times\calR}.
\end{equation*}
Since
$\frac{A-2}{B}
=(\tfrac{NR}{D}-2)\bigbracket{\tfrac{D(N-1)}{NR-D}}^{\frac12}
=(NR-2D)\bigbracket{\tfrac{N-1}{D(NR-D)}}^{\frac12}$,
we thus have that such a matrix $\bfS$ is the signature matrix of an $\EITFF(D,N,R)$ if and only if
\begin{equation}
\label{eq.signature matrix conditions}
\bfS_{n_1,n_2}=\bfS_{n_2,n_1}^*,
\qquad
\sum_{n_3\in\calN}\bfS_{n_1,n_3}\bfS_{n_3,n_2}
=(NR-2D)\bigbracket{\tfrac{N-1}{D(NR-D)}}^{\frac12}\bfS_{n_1,n_2},
\qquad\forall\,n_1\neq n_2.
\end{equation}
Clearly, the resulting EITFF is moreover real if $\bfS$ is real.
In the special case where $\bfS$ is real and $R=1$, \eqref{eq.signature matrix conditions} reduces to a classical method for constructing real ETFs that is intimately related to \textit{two-graphs}.
In the final section of this paper, we generalize that approach.
There, inspired by some discoveries made in Section~4, we make educated guesses for choices of $\bfS$ that might yield $\EITFF(N,N,2)$ and $\EITFF(N-1,N,2)$, and use~\eqref{eq.signature matrix conditions} to characterize when those guesses are valid.

\section{Harmonic equichordal and equi-isoclinic tight fusion frames}

We begin this section with a method for \textit{modulating} a possibly nonuniform TFF into a uniform one with larger parameters.
This generalizes the method of Theorem~14 in~\cite{CasazzaFMWZ11} to the case of nonuniform TFFs and noncyclic abelian groups.
Later on in this section, we show that a TFF is produced by this method if and only if it is ``harmonic" in a certain sense that itself is the natural fusion-frame-based generalization of the concept of a harmonic frame.
We then characterize when such harmonic TFFs are equichordal and/or equi-isoclinic.
In the next section, we exploit these characterizations to construct new (harmonic) ECTFFs and EITFFs.

Here and throughout, let $\calG$ be a finite abelian group whose operation is denoted as addition.
A \textit{character} of $\calG$ is a homomorphism $\gamma:\calG\rightarrow\bbT:=\set{z\in\bbC: \abs{z}=1}$.
The \textit{Pontryagin dual} of $\calG$ is the set $\widehat{\calG}$ of all such characters,
which is itself a group under pointwise multiplication.
In fact, it is well known that since $\calG$ is finite,
$\widehat{\calG}$ is isomorphic to $\calG$ and moreover that the characters \smash{$\set{\gamma}_{\gamma\in\widehat{\calG}}$} form an equal-norm orthogonal basis for $\bbC^\calG$.
As such, the synthesis operator $\bfGamma$ of this basis satisfies $\bfGamma^{-1}=\frac1G\bfGamma^*$.
This $\calG\times\widehat{\calG}$ matrix $\bfGamma$ is the \textit{character table} of $\bfG$, having $(g,\gamma)$th entry $\bfGamma(g,\gamma)=\gamma(g)$.
Its adjoint $\bfGamma^*$ is the \textit{discrete Fourier transform} (DFT) over $\calG$.
For a fixed character $\gamma$,
the function $\bfM_\gamma:\bbC^\calG\rightarrow\bbC^\calG$,
$(\bfM_\gamma\bfy)(g):=\gamma(g)\bfy(g)$ is known as \textit{modulation} by $\gamma$.
We use this same term to more generally refer to other types of pointwise multiplication by characters.
In particular, in what follows,
we construct the isometries $\set{\bfPhi_\gamma}_{\gamma\in\widehat{\calG}}$ of a uniform TFF by modulating the adjoint $\bfPsi^*$ of the fusion synthesis operator $\bfPsi$ of the isometries $\set{\bfPsi_g}_{g\in\calG}$ of some $\calG$-indexed (possibly nonuniform) TFF for $\bbC^\calR$.

\begin{definition}
\label{def.harmonic matrix ensemble}
Let $\calG$ be a finite abelian group of order $G$, and let $\calR$ be a finite set of order $R>0$.
Consider any ensemble $\set{\bfPsi_g}_{g\in\calG}$ of $\calR\times\calD_g$ matrices, where each finite set $\calD_g$ has order $D_g$ and $D:=\sum_{g\in\cal G}D_g>0$.
Put $\calD:=\calG\times\set{\calD_g}=\cup_{g\in\calG}(\set{g}\times\calD_g)$.
The \textit{harmonic matrix ensemble} generated by \smash{$\set{\bfPsi_g}_{g\in\calG}$} is the sequence \smash{$\set{\bfPhi_\gamma}_{\gamma\in\widehat{\calG}}$} of $\calD\times\calR$ matrices defined by
\begin{equation}
\label{eq.def of big frame}
\bfPhi_\gamma((g,d),r):=(\tfrac RD)^{\frac12}\gamma(g)\overline{\bfPsi_{g}(r,d)},
\quad\forall\,(g,d)\in\calD,\, r\in\calR.
\end{equation}
\end{definition}

We will soon see in Theorem~\ref{thm.small TFF yields harmonic} that the isometries of a TFF generate a harmonic matrix ensemble that forms another TFF.
To distinguish between these two TFFs, we will say the \textit{generating TFF} generates the \textit{harmonic TFF}; see Definition~\ref{def.harmonic} for a slightly more general notion following the characterization in Theorem~\ref{thm.harmonic TFF characterization}.
The following example uses the technology in Definition~\ref{def.harmonic matrix ensemble} to obtain the second-known construction of an $\EITFF(4,5,2)$.

\begin{example}
\label{ex.EITFF(4,5,2)}
Let $\calG$ be the cyclic group $\bbZ_5$ of integers modulo $5$.
Let $\calR=\set{1,2}$, which has cardinality $R=2$.
Let $\calD_0$ be empty, and let $\calD_1=\calD_2=\calD_3=\calD_4=\set{1}$ so that
$\calD=\cup_{n\in\bbZ_5}(\set{n}\times\calD_n)$ is naturally identified with the subset $\set{1,2,3,4}$ of $\bbZ_5$,
which has cardinality $D=4$.
Now consider five isometries $\set{\bfPsi_n}_{n\in\bbZ_5}$ with
$\bfPsi_n:\bbC^{\calD_n}\rightarrow\bbC^\calR\cong\bbC^2$ for all $n$,
as well as each one's corresponding $2\times 2$ projection  $\bfP_n:=\bfPsi_n^{}\bfPsi_n^*$.
When $n=0$ we have $\calD_0=\emptyset$, and so $\bfPsi_0$ is informally a ``matrix with two rows and zero columns," which by formal convention means that $\bfP_0$ is a $2\times 2$ zero matrix (a rank-$0$ projection).
When instead $n\in\set{1,2,3,4}$, $D_n=\#(\calD_n)=1$,
meaning $\bfPsi_n$ is a $2\times 1$ matrix whose sole column is unit norm.
For example, we may choose
\begin{equation*}
\label{eq.EITFF(4,5,2) 1}
\bfPsi
=\left[\begin{array}{c|c|c|c|c}
\bfPsi_0&\bfPsi_1&\bfPsi_2&\bfPsi_3&\bfPsi_4
\end{array}\right]
=\tfrac1{\sqrt{2}}\left[\begin{array}{c|c|c|c|c}
\,\cdot\ &\ 1\ &\    1\ &\ \phantom{-}1\ &\ \phantom{-}1\ \\
\,\cdot\ &\ 1\ &\ \rmi\ &\ -\rmi\ &\ -1\
\end{array}\right],
\end{equation*}
where ``$\cdot$" is a placeholder for an entry whose column index belongs to the empty set.  Here,
\begin{equation}
\label{eq.EITFF(4,5,2) 2}
\bfP_0=\left[\begin{array}{cc}0&0\\0&0\end{array}\right],\quad
\bfP_1=\tfrac12\left[\begin{array}{cc}1&1\\1&1\end{array}\right],\quad
\bfP_2=\tfrac12\left[\begin{array}{cr}1&-\rmi\\\rmi&1\end{array}\right],\quad
\bfP_3=\tfrac12\left[\begin{array}{rc}\phantom{-}1&\rmi\\-\rmi&1\end{array}\right],\quad
\bfP_4=\tfrac12\left[\begin{array}{rr}1&-1\\-1&1\end{array}\right].
\end{equation}
Since these five projections sum to $2\bfI$,
$\set{\bfPsi_n}_{n\in\bbZ_5}$ are the isometries of a $\TFF(2,5,\set{0,1,1,1,1})$ for $\bbC^2$.
As we shall see, they generate the isometries
\smash{$\set{\bfPhi_\gamma}_{\gamma\in\widehat{\bbZ}_5}$} of a harmonic $\TFF(4,5,2)$ for $\bbC^\calD\cong\bbC^4$.
Here, we elect to identify $\bbZ_5$ with its Pontryagin dual by identifying any $m\in\bbZ_5$ with the character $\gamma_m:\bbZ_5\rightarrow\bbT$,
$\gamma_m(n)=\omega^{mn}$ where $\omega=\exp(\frac{2\pi\rmi}5)$,
and denote $\bfPhi_{\gamma_m}$ as simply $\bfPhi_m$.
In particular, by~\eqref{eq.def of big frame},
each $\bfPhi_m$ is obtained by multiplying $\bfPsi^*$ of the left by a diagonal (modulation) matrix and then scaling by a factor of
\smash{$(\frac{R}{D})^{\frac12}=(\frac{2}{4})^{\frac12}=\frac1{\sqrt{2}}$}:
\begin{equation}
\label{eq.EITFF(4,5,2) 3}
\bfPhi_m
=\tfrac1{\sqrt{2}}\left[\begin{array}{cccc}
\omega^m&0&0&0\\
0&\omega^{2m}&0&0\\
0&0&\omega^{3m}&0\\
0&0&0&\omega^{4m}
\end{array}\right]
\tfrac1{\sqrt{2}}\left[\begin{array}{cc}
1&\phantom{-}1\\
1&-\rmi\\
1&\phantom{-}\rmi\\
1&-1\end{array}\right].
\end{equation}
Essentially, since $\set{\bfPsi_n}_{n\in\bbZ_5}$ are the isometries of a TFF,
the rows of $\bfPsi$ are equal-norm orthogonal,
and so when suitably scaled,
the columns of $\bfPsi^*$ form an orthonormal basis for some two-dimensional subspace of $\bbC^4$;
we then act on this subspace by a diagonal unitary matrix of order $5$ to obtain orthonormal bases for five distinct two-dimensional subspaces of $\bbC^4$.
Remarkably, these five isometries $\set{\bfPhi_m}_{m\in\bbZ_5}$
themselves form a TFF for $\bbC^4$;
as seen in the proof of Theorem~\ref{thm.small TFF yields harmonic},
this is basically due to the fact that every off-diagonal entry of
$\bfPhi\bfPhi^*=\sum_{m\in\bbZ_5}\bfPhi_m^{}\bfPhi_m^*$ is a vanishing geometric sum.
The group-based structure of~\eqref{eq.EITFF(4,5,2) 3} further implies that the fusion Gram matrix $\bfPhi^*\bfPhi$ of this $\TFF(4,5,2)$ is block-circulant,
having its $(m_1,m_2)$th cross-Gram matrix $\bfPhi_{m_1}^*\bfPhi_{m_2}^{}$ depend only on the difference of $m_1$ and $m_2$,
and moreover does so in a way that is essentially entrywise DFT of the generating TFF's projections~\eqref{eq.EITFF(4,5,2) 2}:
\begin{equation*}
\tfrac12\biggparen{\omega^{m_2-m_1}\tfrac12\left[\begin{array}{cc}1&1\\1&1\end{array}\right]
+\omega^{2(m_2-m_1)}\tfrac12\left[\begin{array}{cr}1&-\rmi\\\rmi&1\end{array}\right]
+\omega^{3(m_2-m_1)}\tfrac12\left[\begin{array}{rc}1&\rmi\\-\rmi&1\end{array}\right]
+\omega^{4(m_2-m_1)}\tfrac12\left[\begin{array}{rr}1&-1\\-1&1\end{array}\right]}.
\end{equation*}
Because of this, this $\TFF(4,5,2)$ is equi-isoclinic precisely when $\bfPhi_0^*\bfPhi_m^{}$ is
\smash{$[\frac{NR-D}{D(N-1)}]^{\frac12}
=\sqrt{\frac{3}{8}}$}
times some unitary for each $m\in\set{1,2,3,4}$.
As explained in the next section, the theory of Gauss sums over finite fields implies that this is indeed the case here.
This is the second known construction of a complex $\EITFF(4,5,2)$;
as shown in~\cite{Waldron20},
an $\EITFF(4,5,2)$ can also be obtained by applying Hoggar's $\bbH$-to-$\bbC$ method~\cite{Hoggar77} to a quaternionic $\ETF(2,5)$~\cite{CohnKM16,EtTaoui20,Waldron20}.

\end{example}

The next example illustrates how harmonic matrix ensembles generalize harmonic frames.

\begin{example}
\label{ex.EITFF(3,7,1)}
When $\calR$ is a singleton set, \eqref{eq.def of big frame} reduces to the traditional construction of harmonic (unit norm tight) frames.
To explain,
if $\bfPsi_g:\bbC^{\calD_g}\rightarrow\bbC$ is an isometry then $\calD_g$ is either empty or a singleton set, and in the latter case we assume $\calD_g=\set{1}$ without loss of generality, implying that $\bfPhi_g$ is the $1\times1$ matrix whose sole entry is $1$.
In this setting, $\calD=\cup_{g\in\calG}(\set{g}\times\calD_g)$ is naturally identified with the subset \smash{$\widetilde{\calD}:=\set{g\in\calG: \calD_g\neq\emptyset}$} of $\calG$ that essentially has $\set{\bfPsi_g}_{g\in\calG}$ as its characteristic function.
Such isometries $\set{\bfPsi_g}_{g\in\calG}$ automatically form a TFF for $\bbC$,
satisfying
$\sum_{g\in\calG}\bfPsi_g^{}\bfPsi_g^*
=\sum_{\set{g\in\calG: \calD_g\neq\emptyset}}1
=D$.
For example, when $\calG=\bbZ_7$ and we let $\calD_g=\set{1}$ for $g\in\set{1,2,4}$ and
$\calD_g=\emptyset$ for $g\in\set{0,3,5,6}$,
the set $\calD=\set{(1,1),(2,1),(4,1)}$ is naturally identified with the subset $\widetilde{\calD}=\set{1,2,4}$ of $\bbZ_7$ whose characteristic function $\bfone_{\widetilde{\calD}}$ is essentially given by the sequence of isometries $\set{\bfPsi_n}_{n\in\bbZ_7}$ of the corresponding $\TFF(1,7,\set{0,1,1,0,1,0,0})$, namely
\begin{equation}
\label{eq.EITFF(3,7,1) 1}
\bfPsi
=\left[\begin{array}{c|c|c|c|c|c|c}
\ \bfPsi_0\ &\ \bfPsi_1\ &\ \bfPsi_2\ &\ \bfPsi_3\ &\ \bfPsi_4\ &\ \bfPsi_5\ &\ \bfPsi_6\
\end{array}\right]
=\left[\begin{array}{c|c|c|c|c|c|c}
\ \cdot\ &\ 1\ &\ 1\ &\ \cdot\ &\ 1\ &\ \cdot\ &\ \cdot\
\end{array}\right].
\end{equation}
The harmonic matrix ensemble \smash{$\set{\bfPhi_\gamma}_{\gamma\in\widehat{\calG}}$} generated by $\set{\bfPsi_g}_{g\in\calG}$ forms a $\TFF(D,G,1)$ for $\bbC^\calD$:
for any \smash{$\gamma\in\widehat{\calG}$}, $\bfPhi_\gamma$ is a $\calD\times\set{1}$ matrix which, for any $g\in\calG$ such that $\calD_g\neq\emptyset$, has $((g,1),1)$th entry
\begin{equation*}
\bfPhi_\gamma((g,1),1)
=\tfrac1{\sqrt{D}}\gamma(g)\overline{\bfPsi_{g}(1,1)}
=\tfrac1{\sqrt{D}}\gamma(g).
\end{equation*}
Thus, \smash{$\set{\bfPhi_\gamma}_{\gamma\in\widehat{\calG}}$} is naturally identified with the traditional harmonic frame arising from $\widetilde{\calD}$,
namely the UNTF \smash{$\set{\bfphi_\gamma}_{\gamma\in\widehat{\calG}}$} for \smash{$\bbC^{\widetilde{\calD}}$} whose synthesis operator is obtained by extracting the rows of the $\calG\times\widehat{\calG}$ character table of $\calG$ that are indexed by $\widetilde{\calD}$~\cite{GoyalKK01},
and then scaling them by~\smash{$\frac1{\sqrt{D}}$}.
Here, \eqref{eq.cross Gram of big frame} reduces to a well-known expression for the entries of such a frame's $\widehat{\calG}$-circulant Gram matrix in terms of the DFT of $\bfone_{\widetilde{\calD}}$:
\begin{equation}
\label{eq.EITFF(3,7,1) 1.5}
(\bfPhi_{\gamma_1}^*\bfPhi_{\gamma_2}^{})(1,1)
=\tfrac1D\sum_{g\in\calG}(\gamma_1^{-1}\gamma_2^{})(g)\bfone_{\widetilde{\calD}}(g)
=\tfrac1D\ip{\gamma_1^{}\gamma_2^{-1}}{\bfone_{\widetilde{\calD}}}
=\tfrac1D(\bfGamma^*\bfone_{\widetilde{\calD}})(\gamma_1^{}\gamma_2^{-1}).
\end{equation}
It is straightforward to show (and has long been known~\cite{Turyn65}) that \smash{$\abs{(\bfGamma^*\bfone_{\widetilde{\calD}})(\gamma)}$} is constant over all $\gamma\neq1$ if and only if $\widetilde{\calD}$ is a \textit{difference set} for $\calG$,
having \smash{$\#\set{(d_1,d_2)\in\widetilde{\calD}\times\widetilde{\calD}: d_1-d_2=g}$} be constant over all $g\neq0$.
Combining this with~\eqref{eq.EITFF(3,7,1) 1.5} gives the well-known fact~\cite{Konig99,StrohmerH03,XiaZG05,DingF07} that this occurs if and only if \smash{$\set{\bfphi_\gamma}_{\gamma\in\widehat{\calG}}$} is an ETF for \smash{$\bbC^{\widetilde{\calD}}$}.
In Theorem~\ref{thm.EITFF}, we show more generally that a harmonic matrix ensemble forms an EITFF if and only if the generating TFF's isometries $\set{\bfPsi_g}_{g\in\calG}$ have a difference-set-like property.

For example,
\smash{$\widetilde{\calD}=\set{1,2,4}$} is a difference set for $\calG=\bbZ_7$ since every $n\neq0$ appears in its \textit{difference table} the same number of times, namely exactly once:
\begin{equation*}
\begin{array}{r|rrr}
-&1&2&4\\
\hline
1&0&6&4\\
2&1&0&5\\
4&3&1&0
\end{array}.
\end{equation*}
As such,
identifying $\calG=\bbZ_7$ with its Pontryagin dual by identifying $m\in\bbZ_7$ with $\gamma_m:\bbZ_7\rightarrow\bbT$, $\gamma_m(n)=\omega^{mn}$ where $\omega=\exp(\frac{2\pi\rmi}{7})$,
we have that the harmonic matrix ensemble $\set{\bfPhi_m}_{m\in\bbZ_7}$ generated by $\set{\bfPsi_n}_{n\in\bbZ_7}$ of~\eqref{eq.EITFF(3,7,1) 1} forms an $\EITFF(3,7,1)$ for $\bbC^\calD$,
namely
\begin{equation*}
\bfPhi_m=\tfrac1{\sqrt{3}}\left[\begin{array}{ccc}
\omega^m&0&0\\
0&\omega^{2m}&0\\
0&0&\omega^{4m}
\end{array}\right]
\left[\begin{array}{c}1\\1\\1\end{array}\right],
\ \forall m\in\bbZ_7,\quad\text{i.e.,}\quad
\bfPhi
=\tfrac1{\sqrt{3}}\left[\begin{array}{ccccccc}
1&\omega^{1}&\omega^{2}&\omega^{3}&\omega^{4}&\omega^{5}&\omega^{6}\\
1&\omega^{2}&\omega^{4}&\omega^{6}&\omega^{1}&\omega^{3}&\omega^{5}\\
1&\omega^{4}&\omega^{1}&\omega^{5}&\omega^{2}&\omega^{6}&\omega^{4}
\end{array}\right].
\end{equation*}
Here,
$\bfPhi_{m_1}^*\bfPhi_{m_2}^{}$ is a $1\times1$ matrix whose sole entry of
$\frac13[\omega^{m_2-m_1}+\omega^{2(m_2-m_1)}+\omega^{4(m_2-m_1)}]$ has modulus \smash{$\frac{\sqrt{2}}{3}$} whenever $m_1\neq m_2$.

\end{example}

We are now ready to state the main result of this section.

\begin{theorem}
\label{thm.small TFF yields harmonic}
In the setting of Definition~\ref{def.harmonic matrix ensemble},
\smash{$\set{\bfPhi_\gamma}_{\gamma\in\widehat{\calG}}$} are isometries that form a $\TFF(D,G,R)$ for $\bbC^\calD$ if and only if \smash{$\set{\bfPsi_g}_{g\in\calG}$}
are isometries that form a
\smash{$\TFF(R,G,\set{D_g}_{g\in\calG})$} for $\bbC^\calR$,
in which case each of the following holds:
\begin{enumerate}
\renewcommand{\labelenumi}{(\alph{enumi})}
\item
the harmonic matrix ensemble generated by any spatial (resp.\ Naimark) complement of $\set{\bfPsi_g}_{g\in\calG}$ is a Naimark (resp.\ spatial) complement of \smash{$\set{\bfPhi_\gamma}_{\gamma\in\widehat{\calG}}$};
\item
given any appropriately sized unitary matrices $\bfU$ and $\set{\bfV_g}_{g\in\calG}$, there exists a unitary matrix $\bfV$ such that the harmonic matrix ensemble generated by \smash{$\set{\bfU\bfPsi_g\bfV_g}_{g\in\calG}$} is \smash{$\set{\bfV^*\bfPhi_\gamma \bfU^*}_{\gamma\in\widehat{\calG}}$};
\item
ensembles of isometries that directly sum to \smash{$\set{\bfPsi_g}_{g\in\calG}$} generate harmonic matrix ensembles that directly sum to~\smash{$\set{\bfPhi_\gamma}_{\gamma\in\widehat{\calG}}$}.
\end{enumerate}
Furthermore, the $\calR\times\calR$ cross-Gram matrices of $\set{\bfPhi_\gamma}_{\gamma\in\widehat{\calG}}$ are given by
\begin{equation}
\label{eq.cross Gram of big frame}
\bfPhi_{\gamma_1}^*\bfPhi_{\gamma_2}^{}
=\tfrac{R}{D}\sum_{g\in\calG}(\gamma_1^{-1}\gamma_2^{})(g)\bfPsi_g^{}\bfPsi_g^*,
\quad\forall\, \gamma_1,\gamma_2\in\widehat{\calG}.
\end{equation}

\end{theorem}

We alert the reader to a subtlety in Theorem~\ref{thm.small TFF yields harmonic}(a): a \textit{spatial} complement of $\set{\bfPsi_g}_{g\in\calG}$ corresponds with a \textit{Naimark} complement of \smash{$\set{\bfPhi_\gamma}_{\gamma\in\widehat{\calG}}$}, and vice versa.

\begin{proof}[Proof of Theorem~\ref{thm.small TFF yields harmonic}]
For any $\gamma_1,\gamma_2\in\widehat{\calG}$, and any $r_1,r_2\in\calR$,
\eqref{eq.def of big frame} implies
\begin{align*}
(\bfPhi_{\gamma_1}^*\bfPhi_{\gamma_2}^{})(r_1,r_2)
&=\sum_{(g,d)\in\calD}\overline{\bfPhi_{\gamma_1}((g,d),r_1)}\bfPhi_{\gamma_1}((g,d),r_2)\\
&=\tfrac{R}{D}\sum_{g\in\calG}\sum_{d\in\calD_g}
\overline{\gamma_1(g)}\bfPsi_g(r_1,d)
\gamma_2(g)\overline{\bfPsi_g(r_2,d)}\\
&=\tfrac{R}{D}\sum_{g\in\calG}(\gamma_1^{-1}\gamma_2^{})(g)
\sum_{d\in\calD_g}\bfPsi_g^{}(r_1,d)\bfPsi_g^*(d,r_2)\\
&=\biggparen{\tfrac{R}{D}\sum_{g\in\calG}(\gamma_1^{-1}\gamma_2^{})(g)\bfPsi_g^{}\bfPsi_g^*}(r_1,r_2),
\end{align*}
namely~\eqref{eq.cross Gram of big frame}.
When $\gamma_1=\gamma_2$, \eqref{eq.cross Gram of big frame} reduces to
$\bfPhi_{\gamma}^*\bfPhi_{\gamma}^{}
=\tfrac{R}{D}\sum_{g\in\calG}\bfPsi_g^{}\bfPsi_g^*$ for all $\gamma\in\widehat{\calG}$.  In particular,
\begin{equation}
\label{eq.pf of small TFF yields harmonic 1}
\bfPhi_{\gamma}^*\bfPhi_{\gamma}^{}=\bfI,\ \forall\,\gamma\in\widehat{\calG}
\quad\Longleftrightarrow\quad
\sum_{g\in\calG}\bfPsi_g^{}\bfPsi_g^*=\tfrac{D}{R}\bfI.
\end{equation}
Similarly, for any $(g_1,d_1),(g_2,d_2)\in\calD$, \eqref{eq.def of big frame} implies
\begin{align*}
\biggparen{\sum_{\gamma\in\widehat{\calG}}\bfPhi_\gamma^{}\bfPhi_\gamma^*}((g_1,d_1),(g_2,d_2))
&=\sum_{\gamma\in\widehat{\calG}}\sum_{r\in\calR}
\bfPhi_\gamma((g_1,d_1),r)\overline{\bfPhi_\gamma((g_2,d_2),r)}\\
&=\tfrac{R}{D}\sum_{\gamma\in\widehat{\calG}}\sum_{r\in\calR}
\gamma(g_1)\overline{\bfPsi_{g_1}(r,d_1)}
\overline{\gamma(g_2)}\bfPsi_{g_2}(r,d_2)\\
&=\tfrac{R}{D}\sum_{r\in\calR}\overline{\bfPsi_{g_1}(r,d_1)}\bfPsi_{g_2}(r,d_2)
\sum_{\gamma\in\widehat{\calG}}\gamma(g_1-g_2).
\end{align*}
Since the character table of $\calG$ is a complex Hadamard matrix, $\sum_{\gamma\in\widehat{\calG}}\gamma(g_{1}-g_{2})=G$ if $g_1=g_2$,
and otherwise,
it vanishes.
As such, the above expression simplifies to
\begin{align*}
\biggparen{\sum_{\gamma\in\widehat{\calG}}\bfPhi_\gamma^{}\bfPhi_\gamma^*}((g_1,d_1),(g_2,d_2))
&=\left\{\begin{array}{cl}
\displaystyle\tfrac{GR}{D}\sum_{r\in\calR}\overline{\bfPsi_{g_1}(r,d_1)}\bfPsi_{g_2}(r,d_2),&
\ g_1=g_2,\\
0,&\ g_1\neq g_2,\end{array}\right.\\
&=\left\{\begin{array}{cl}
\displaystyle\tfrac{GR}{D}(\bfPsi_{g_1}^*\bfPsi_{g_2}^{})(d_1,d_2),&
\ g_1=g_2,\\
0,&\ g_1\neq g_2.\end{array}\right.
\end{align*}
As such,
\begin{equation}
\label{eq.pf of small TFF yields harmonic 2}
\sum_{\gamma\in\widehat{\calG}}\bfPhi_\gamma^{}\bfPhi_\gamma^*=\tfrac{GR}{D}\bfI
\quad\Longleftrightarrow\quad
\bfPsi_g^*\bfPsi_g^{}=\bfI,\ \forall\,g\in\calG.
\end{equation}
Overall, \smash{$\set{\bfPhi_\gamma}_{\gamma\in\widehat{\calG}}$} are the isometries of some $\TFF(D,G,R)$ for $\bbC^\calD$ if and only if the left-hand conditions of both
\eqref{eq.pf of small TFF yields harmonic 1} and~\eqref{eq.pf of small TFF yields harmonic 2} hold,
which occurs if and only if the right-hand conditions of both~\eqref{eq.pf of small TFF yields harmonic 1} and~\eqref{eq.pf of small TFF yields harmonic 2} hold,
namely if and only if \smash{$\set{\bfPsi_g}_{g\in\calG}$}
are the isometries of some
\smash{$\TFF(R,G,\set{D_g}_{g\in\calD_g})$} for $\bbC^\calR$.
For the remainder of this proof, assume that this is indeed the case,
and so necessarily $D_g\leq R$ for all $g$ and moreover
$R\leq D=\sum_{g\in\calG}D_g\leq GR$.

Turning to the proof of (a), note that if $D<GR$ then $D_g<R$ for at least one choice of $g$.
In this case, let $\set{\bfPsi_g^\sharp}_{g\in\calG}$ be any $\set{\calR\times\calE_g}_{g\in\calG}$ isometries, respectively, that are spatial-complementary to $\set{\bfPsi_g}_{g\in\calG}$,
and so yield a \smash{$\TFF(R,G,\set{R-D_g}_{g\in\calD_g})$} for $\bbC^\calR$.
Now let \smash{$\set{\bfPhi_\gamma^\flat}_{\gamma\in\widehat{\calG}}$} be the harmonic matrix ensemble generated by \smash{$\set{\bfPsi_g^\sharp}_{g\in\calG}$}:
for any $\gamma\in\widehat{\calG}$, $(g,e)\in\calE:=\calG\times\set{\calE_g}$ and $r\in\calR$,
let
\begin{equation*}
\bfPhi_\gamma^\flat((g,e),r)
:=(\tfrac R{GR-D})^{\frac12}\gamma(g)\overline{\bfPsi_g^\sharp(r,e)}.
\end{equation*}
From what we have already shown, \smash{$\set{\bfPhi_\gamma^\flat}_{\gamma\in\widehat{\calG}}$} are the isometries of a $\TFF(GR-D,G,R)$ for $\bbC^\calE$.
Here, since $\set{\bfPsi_g}_{g\in\calG}$ and $\set{\bfPsi_g^\sharp}_{g\in\calG}$ are spatial-complementary, $\bfPsi_g^{}\bfPsi_g^*+\bfPsi_g^\sharp(\bfPsi_g^\sharp)^*=\bfI$ for all $g$.
Thus,
for any \smash{$\gamma_1,\gamma_2\in\widehat{\calG}$}, $\gamma_1\neq\gamma_2$,
\eqref{eq.cross Gram of big frame}
along with the fact that $\sum_{\gamma\in\widehat{\calG}}(\gamma_1^{-1}\gamma_2^{})(g)=0$ gives
\begin{multline*}
(\bfPhi_{\gamma_1}^\flat)^*\bfPhi_{\gamma_2}^{\flat}
=\tfrac{R}{GR-D}\sum_{\gamma\in\widehat{\calG}}(\gamma_1^{-1}\gamma_2^{})(g)\bfPsi_g^{\sharp}(\bfPsi_g^\sharp)^*
=\tfrac{R}{GR-D}\sum_{\gamma\in\widehat{\calG}}(\gamma_1^{-1}\gamma_2^{})(g)(\bfI-\bfPsi_g^{}\bfPsi_g^*)\\
=-\tfrac{D}{GR-D}\tfrac{R}{D}\sum_{\gamma\in\widehat{\calG}}(\gamma_1^{-1}\gamma_2^{})(g)\bfPsi_g^{}\bfPsi_g^*
=-\tfrac{D}{GR-D}\bfPhi_{\gamma_1}^*\bfPhi_{\gamma_2}^{},
\end{multline*}
and so these $\TFF(D,G,R)$ and $\TFF(GR-D,G,R)$ are in fact Naimark-complementary, as claimed.

Similarly, if $R<D=\sum_{g\in\calG}D_g$,
then let $\set{\bfPsi_g^\flat}_{g\in\calG}$ be any $\set{\calS\times\calD_g}_{g\in\calG}$ isometries, respectively, that are Naimark-complementary to $\set{\bfPsi_g}_{g\in\calG}$,
and so yield a $\TFF(D-R,G,\set{D_g}_{g\in\calG})$ for $\bbC^\calS$ where $\calS$ is some (arbitrary) set with $\#(\calS)=D-R>0$.
Now let \smash{$\set{\bfPhi_\gamma^\sharp}_{\gamma\in\widehat{\calG}}$} be the harmonic matrix ensemble generated by \smash{$\set{\bfPsi_g^\flat}_{g\in\calG}$}:
for any $\gamma\in\widehat{\calG}$, $(g,d)\in\calD=\calG\times\set{\calD_g}$ and $s\in\calS$,
let
\begin{equation*}
\bfPhi_\gamma^\sharp((g,d),s)
:=(\tfrac{D-R}D)^{\frac12}\gamma(g)\overline{\bfPsi_g^\flat(s,d)}.
\end{equation*}
From what we have already shown, \smash{$\set{\bfPhi_\gamma^\sharp}_{\gamma\in\widehat{\calG}}$} are the isometries of a $\TFF(D,G,D-R)$ for $\bbC^\calD$.
To show that they are spatial-complementary to~\smash{$\set{\bfPhi_\gamma}_{\gamma\in\widehat{\calG}}$} it suffices to show that
\smash{$\bfPhi_\gamma^{}\bfPhi_\gamma^*+\bfPhi_\gamma^\sharp(\bfPhi_\gamma^\sharp)^*=\bfI$} for all $\gamma\in\widehat{\calG}$.
To see this, note that since $\set{\bfPsi_g}_{g\in\calG}$ and $\set{\bfPsi_g^\flat}_{g\in\calG}$ are Naimark complementary,
\begin{equation*}
\tfrac RD\bfPsi_{g_1}^*\bfPsi_{g_2}^{}
+\tfrac{D-R}D(\bfPsi_{g_1}^\flat)^*\bfPsi_{g_2}^{\flat}
=\left\{\begin{array}{cl}
\bfI,&\ g_1=g_2,\\
\bfzero,&\ g_1\neq g_2,
\end{array}\right.
\end{equation*}
and so for any $(g_1,d_1),(g_2,d_2)\in\calD=\calG\times\set{\calD_g}$,
\begin{align*}
&[\bfPhi_\gamma^{}\bfPhi_\gamma^*+\bfPhi_\gamma^\sharp(\bfPhi_\gamma^\sharp)^*]
((g_1,d_1),(g_2,d_2))\\
&\qquad=\sum_{r\in\calR}\bfPhi_\gamma((g_1,d_1),r)\overline{\bfPhi_\gamma((g_2,d_2),r)}
+\sum_{s\in\calS}\bfPhi_\gamma^\sharp((g_1,d_1),s)\overline{\bfPhi_\gamma^\sharp((g_2,d_2),s)}\\
&\qquad=\gamma(g_1-g_2)
\biggbracket{
\tfrac RD\sum_{r\in\calR}\overline{\bfPsi_{g_1}(r,d_1)}\bfPsi_{g_2}(r,d_2)
+\tfrac{D-R}D\sum_{s\in\calS}\overline{\bfPsi_{g_1}^\flat(s,d_1)}\bfPsi_{g_2}^\flat(s,d_2)}\\
&\qquad=\gamma(g_1-g_2)
\Bigbracket{
\tfrac RD\bfPsi_{g_1}^*\bfPsi_{g_2}^{}
+\tfrac{D-R}D(\bfPsi_{g_1}^\flat)^*\bfPsi_{g_2}^{\flat}}(d_1,d_2)\\
&\qquad=\gamma(g_1-g_2)
\left\{\begin{array}{cl}
\bfI(d_1,d_2),&\ g_1=g_2,\\
0,&\ g_1\neq g_2,
\end{array}\right.\\
&\qquad=\left\{\begin{array}{cl}
1,&\ (g_1,d_1)=(g_2,d_2),\\
0,&\ (g_1,d_1)\neq(g_2,d_2).
\end{array}\right.
\end{align*}

For (b), note that if $\set{\bfPsi_g}_{g\in\calG}$ are the isometries of a  \smash{$\TFF(R,G,\set{D_g}_{g\in\calD_g})$} for $\bbC^\calR$,
then for any \smash{$\widetilde{\calR}\times\calR$} unitary matrix $\bfU$,
$\set{\bfU\bfPsi_g}_{g\in\calG}$ are the isometries of a  \smash{$\TFF(R,G,\set{D_g}_{g\in\calD_g})$} for \smash{$\bbC^{\widetilde{\calR}}$},
and the harmonic matrix ensemble they generate can be obtained by multiplying $\set{\bfPhi_\gamma}_{\gamma\in\widehat{\calG}}$ on the right by $\bfU^*$:
for any \smash{$\gamma\in\widehat{\calG}$}, $(g,d)\in\calD=\calG\times\set{D_g}$, \smash{$\tilde{r}\in\widetilde{\calR}$},
\begin{multline*}
(\tfrac RD)^{\frac12}\gamma(g)\overline{(\bfU\bfPsi_g)(\tilde{r},d)}
=(\tfrac RD)^{\frac12}
\gamma(g)\overline{\sum_{r\in\calR}\bfU(\tilde{r},r)\bfPsi_g(r,d)}\\
=\sum_{r'\in\calR}(\tfrac RD)^{\frac12}
\gamma(g)\overline{\bfPsi_g(r,d)}\overline{\bfU(\tilde{r},r)}
=\sum_{r'\in\calR}\bfPhi_\gamma((g,d),r)\bfU^*(r,\tilde{r})
=(\bfPhi_\gamma\bfU^*)((g,d),\tilde{r}).
\end{multline*}
This thus only affects the bases of the subspaces of our $\TFF(D,G,R)$,
not the subspaces themselves.

As a dual relation, note that if $\set{\bfV_g}_{g\in\calG}$ are any appropriately sized unitaries, then $\set{\bfPsi_g\bfV_g}_{g\in\calG}$ are alternative isometries for our \smash{$\TFF(R,G,\set{D_g}_{g\in\calD_g})$} for $\bbC^\calR$, which in turn generates a harmonic $\TFF(D,G,R)$ whose isometries are left-unitarily equivalent to the original $\TFF(D,G,R)$:
letting $\bfV$ be the $\calD\times\calD$ block-diagonal unitary whose $g$th diagonal block is $\bfV_g$, i.e.,
\begin{equation*}
\bfV((g_1,d_1),(g_2,d_2))=\left\{\begin{array}{cl}
\bfV_{g_1}(d_1,d_2),&\ g_1=g_2,\\
0,&\ g_1\neq g_2,\end{array}\right.
\end{equation*}
note that for any $\gamma\in\widehat{\calG}$, $(d,r)\in\calD=\calG\times\set{\calD_g}$, $r\in\calR$,
\begin{multline*}
(\tfrac RD)^{\frac12}\gamma(g)\overline{(\bfPsi_g\bfV_g)(r,d)}
=(\tfrac RD)^{\frac12}\gamma(g)
\overline{\sum_{d'\in\calD_g}\bfPsi_g(r,d')\bfV_g(d',d)}
=\sum_{d'\in\calD_g}\bfV_g^*(d,d')\bfPhi_\gamma((g,d'),r)\\
=\sum_{(g',d')\in\calD}\bfV^*((g,d),(g',d'))\bfPhi_\gamma((g',d'),r)
=(\bfV^*\bfPhi_\gamma)((g,d),r).
\end{multline*}

For (c), assume that $\calR$ partitions into nonempty sets $\set{\calR^{(j)}}_{j\in\calJ}$,
that each $\calD_g$ partitions into \smash{$\set{\calD_g^{(j)}}_{j\in\calJ}$} and that each $\calR\times\calD_g$ isometry $\bfPsi_g$ is a block-diagonal matrix whose $j$th diagonal block is some \smash{$\calR^{(j)}\times\calD_g^{(j)}$} matrix \smash{$\bfPsi_g^{(j)}$}.
As outlined in Section~2, it then follows that for each $j\in\calJ$,
\smash{$\set{\bfPsi_g^{(j)}}_{g\in\calG}$} are the isometries of a \smash{$\TFF(R^{(j)},G,\set{D_g^{(j)}})$} for \smash{$\bbC^{\calR^{(j)}}$},
and satisfies the analogue of \eqref{eq.direct sum consistency} in this context, namely that
\begin{equation*}
\tfrac{D^{(j)}}{R^{(j)}}
=\tfrac1{R^{(j)}}\sum_{g\in\calG}D_g^{(j)}
=\tfrac1{R}\sum_{g\in\calG}D_g
=\tfrac{D}{R},
\quad\forall\, j\in\calJ,
\end{equation*}
where, for each $j\in\calJ$, $D^{(j)}$ is defined to be the cardinality of
\smash{$\calD^{(j)}:=\bigcup_{g\in\calG}(\set{g}\times\calD_g^{(j)})$}.
For each $j\in\calJ$,
the harmonic matrix ensemble \smash{$\set{\bfPhi_\gamma^{(j)}}_{\gamma\in\widehat{\calG}}$} generated by $\set{\bfPsi_g^{(j)}}_{g\in\calG}$ forms a
$\TFF(D^{(j)},G,R^{(j)})$ for \smash{$\bbC^{\calD^{(j)}}$}:
for any $\gamma\in\widehat{\calG}$, $(g,d)\in\calD^{(j)}$ and $r\in\calR^{(j)}$,
\begin{equation*}
\bfPhi_\gamma^{(j)}((g,d),r)
=(\tfrac{R^{(j)}}{D^{(j)}})^{\frac12}\gamma(g)\overline{\bfPsi_g^{(j)}(r,d)}
=(\tfrac{R}{D})^{\frac12}\gamma(g)\overline{\bfPsi_g^{(j)}(r,d)}.
\end{equation*}
These sets $\set{\calD^{(j)}}$ form a partition of $\calD$:
they satisfy
\begin{equation*}
\bigcup_{j\in\calJ}\calD^{(j)}
=\bigcup_{g\in\calG}\biggbracket{\set{g}\times\Bigparen{\,\bigcup_{j\in\calJ}\calD_g^{(j)}}}
=\bigcup_{g\in\calG}(\set{g}\times\calD_g)
=\calD,
\end{equation*}
and are disjoint since, for each $g$, the sets $\set{\calD_g^{(j)}}_{j\in\calJ}$ are disjoint.
When combined, these facts give that $\set{\bfPhi_\gamma}_{\gamma\in\widehat{\calG}}$ is the direct sum of \smash{$\set{\bfPhi_\gamma^{(j)}}_{\gamma\in\widehat{\calG}}$} over all $j\in\calJ$:
for any $\gamma\in\widehat{\calG}$, $(g,d)\in\calD$ and $r\in\calR$,
we have that \smash{$(g,d)\in\calD^{(j)}$} for any particular $j$ if and only if \smash{$d\in\calD_g^{(j)}$} for that same $j$, and so
\begin{align*}
\bfPhi_\gamma((g,d),r)
&=(\tfrac RD)^{\frac12}\gamma(g)\overline{\bfPsi_g(r,d)}\\
&=\left\{\begin{array}{cl}
(\tfrac RD)^{\frac12}\gamma(g)\overline{\bfPsi_{g}^{(j)}(r,d)},
&\text{ if }\exists\,j\in\calJ\text{ such that }r\in\calR^{(j)}\text{ and }d\in\calD_g^{(j)},\\
0,&\ \text{else},\end{array}\right.\\
&=\left\{\begin{array}{cl}
\bfPhi_\gamma^{(j)}((g,d),r),&\text{ if }\exists\,j\in\calJ\text{ such that }(g,d)\in\calD^{(j)}\text{ and }r\in\calR^{(j)},\\
0,&\ \text{else},\end{array}\right.
\end{align*}
meaning $\bfPhi_\gamma$ is indeed the block-diagonal matrix whose $j$th diagonal block is $\bfPhi_\gamma^{(j)}$.
\end{proof}

A word of caution regarding (b):
though right-multiplying the isometries $\set{\bfPsi_g}_{g\in\calG}$ of the generating $\TFF(R,G,\set{D_g}_{g\in\calG})$ by unitaries $\set{\bfV_g}_{g\in\calG}$ does equate to left-multiplying the isometries $\set{\bfPhi_\gamma}_{\gamma\in\widehat{\calG}}$ of the harmonic $\TFF(D,G,R)$ by the block-diagonal unitary $\bfV^*$ whose $g$th diagonal block is $\bfV_g^*$,
not every $\calD\times\calD$ unitary is of this form.
In general, left-multiplying such isometries $\set{\bfPhi_\gamma}_{\gamma\in\widehat{\calG}}$ by a unitary will,
of course, yield the isometries of some $\TFF(D,G,R)$.
But,
in general,
there is no guarantee that these isometries will form a harmonic matrix ensemble.
That said,
such a transformation has no effect on the fusion Gram matrix $\bfPhi^*\bfPhi$ of such isometries.
In particular,
in light of~\eqref{eq.cross Gram of big frame},
their fusion Gram matrix is $\widehat{\calG}$-block circulant.
In the next result, we prove a converse to this observation, and use it to characterize TFFs that are ``harmonic" in some basis-independent sense.

\begin{theorem}
\label{thm.harmonic TFF characterization}
Letting $\widehat{\calG}$ be the Pontryagin dual of a finite abelian group $\calG$,
and letting the subspaces $\set{\calU_\gamma}_{\gamma\in\widehat{\calG}}$ form a $\TFF(D,G,R)$ for some complex Hilbert space $\calV$, the following are equivalent:
\begin{enumerate}
\renewcommand{\labelenumi}{(\roman{enumi})}
\item
$\set{\calU_\gamma}_{\gamma\in\widehat{\calG}}$ is unitarily equivalent to a $\TFF(D,G,R)$ that is generated by a $\TFF(R,G,\set{D_g}_{g\in\calG})$ as in Definition~\ref{def.harmonic matrix ensemble}.\smallskip
\item
There exist isometries $\set{\bfPhi_\gamma}_{\gamma\in\widehat{\calG}}$ onto $\set{\calU_\gamma}_{\gamma\in\widehat{\calG}}$ whose corresponding fusion Gram matrix is $\widehat{\calG}$-block-circulant,
i.e., $\bfPhi_{\gamma_1}^*\bfPhi_{\gamma_2}^{}=\bfPhi_{\gamma_3}^*\bfPhi_{\gamma_4}^{}$
whenever $\gamma_1^{-1}\gamma_2^{}=\gamma_3^{-1}\gamma_4^{}$.\smallskip
\item
There exists a subspace $\calU$ of $\calV$ and a $\calV$-unitary representation $\pi$ of $\widehat{\calG}$ (i.e., a homomorphism $\pi$ from $\widehat{\calG}$ into the group of unitary operators on $\calV$) such that $\calU_\gamma=\pi(\gamma)\calU$ for all $\gamma\in\widehat{\calG}$.
\end{enumerate}
\end{theorem}

\begin{proof}
%
(i $\Rightarrow$ ii)
Assume that \smash{$\set{\calU_\gamma}_{\gamma\in\widehat{\calG}}$} is unitarily equivalent to a $\TFF(D,G,R)$ for $\bbC^\calD$ that arises from Definition~\ref{def.harmonic matrix ensemble}.
Specifically, assume there exist isometries $\set{\bfPsi_g}_{g\in\calG}$ of a $\TFF(R,G,\set{D_g}_{g\in\calG})$ for $\bbC^\calR$ that generate a harmonic matrix ensemble \smash{$\set{\bfPhi_\gamma}_{\gamma\in\widehat{\calG}}$}, which in turn form a $\TFF(D,G,R)$ for $\bbC^\calD$, and that there exists a unitary operator $\bfV:\bbC^\calD\rightarrow\calV$ such that for each $\gamma\in\widehat{\calG}$, $\calU_\gamma$ is the image of the isometry \smash{$\widetilde{\bfPhi}_\gamma:=\bfV\bfPhi_\gamma$}.
Here, \eqref{eq.cross Gram of big frame} gives that the cross-Gram matrix
\begin{equation*}
\widetilde{\bfPhi}_{\gamma_1}^*\widetilde{\bfPhi}_{\gamma_2}^{}
=\bfPhi_{\gamma_1}^*\bfV^*\bfV\bfPhi_{\gamma_2}^{}
=\bfPhi_{\gamma_1}^*\bfPhi_{\gamma_2}^{}
=\tfrac{R}{D}\sum_{g\in\calG}(\gamma_1^{-1}\gamma_2^{})(g)\bfPsi_g^{}\bfPsi_g^*
\end{equation*}
depends only on the quotient $\gamma_1^{-1}\gamma_2$ of $\gamma_1$ and $\gamma_2$,
as opposed to their individual values.
Thus, the fusion Gram matrix of \smash{$\set{\widetilde{\bfPhi}_\gamma}_{\gamma\in\widehat{\calG}}$} is $\widehat{\calG}$-block-circulant, as needed.

(ii $\Rightarrow$ iii)
Assume there exist isometries \smash{$\set{\bfPhi_\gamma}_{\gamma\in\widehat{\calG}}$},
\smash{$\bfPhi_\gamma:\bbC^\calR\rightarrow\calV$}, such that \smash{$\bfPhi_\gamma(\bbC^{\calR})=\calU_\gamma$} for each \smash{$\gamma\in\widehat{\calG}$} and such that \smash{$\bfPhi_{\gamma_1}^*\bfPhi_{\gamma_2}^{}=\bfPhi_{\gamma_3}^*\bfPhi_{\gamma_4}^{}$}
whenever \smash{$\gamma_1^{-1}\gamma_2^{}=\gamma_3^{-1}\gamma_4^{}$}.
Defining $\calU$ to be $\calU_1=\bfPhi_1(\bbC^\calR)$,
it thus suffices to construct a $\calV$-unitary representation of $\widehat{\calG}$ with the property that
\begin{equation}
\label{eq.pf of harmonic TFF characterization 4}
\pi(\gamma_1)\bfPhi_{\gamma_2}=\bfPhi_{\gamma_1\gamma_2},
\quad\forall\,\gamma_1,\gamma_2\in\widehat{\calG},
\end{equation}
since when $\gamma_1=\gamma$ and $\gamma_2=1$ this implies
$\calU_\gamma
=\bfPhi_\gamma(\bbC^\calR)
=\pi(\gamma)\calU$,
as desired.
There is at most one choice of such a mapping $\pi$ from $\widehat{\calG}$ into the space of linear operators from $\calV$ into itself:
by~\eqref{eq.pf of harmonic TFF characterization 4},
we necessarily have
$\pi(\gamma_1)\bfPhi_{\gamma_2}^{}\bfPhi_{\gamma_2}^*
=\bfPhi_{\gamma_1\gamma_2}^{}\bfPhi_{\gamma_2}^*$,
and since $\sum_{\gamma}\bfPhi_\gamma^{}\bfPhi_\gamma^*=\tfrac{GR}{D}\bfI$,
summing this equation over all $\gamma_2$ implies that necessarily
\begin{equation}
\label{eq.pf of harmonic TFF characterization 5}
\pi(\gamma_1)
=\tfrac{D}{GR}\sum_{\gamma_2\in\widehat{\calG}}\bfPhi_{\gamma_1\gamma_2}^{}\bfPhi_{\gamma_2}^*,
\quad\forall\,\gamma_1\in\widehat{\calG}.
\end{equation}
Taking~\eqref{eq.pf of harmonic TFF characterization 5} as our definition of $\pi$,
note that for any $\gamma_1,\gamma_2\in\widehat{\calG}$,
our block-circulant assumption gives that
\smash{$\bfPhi_{\gamma_3}^*\bfPhi_{\gamma_2\gamma_4}^{}
=\bfPhi_{\smash{\gamma_2^{-1}\gamma_3^{}}}^*\bfPhi_{\gamma_4}^{}$}
for all $\gamma_3,\gamma_4\in\widehat{\calG}$ and so letting $\gamma_5=\gamma_2^{-1}\gamma_3^{}$ gives
\begin{multline*}
\pi(\gamma_1)\pi(\gamma_2)
=\biggparen{\tfrac{D}{GR}\sum_{\gamma_3\in\widehat{\calG}}\bfPhi_{\gamma_1\gamma_3}^{}\bfPhi_{\gamma_3}^*}
\biggparen{\tfrac{D}{GR}\sum_{\gamma_4\in\widehat{\calG}}\bfPhi_{\gamma_2\gamma_4}^{}\bfPhi_{\gamma_4}^*}
=\tfrac{D^2}{G^2R^2}\sum_{\gamma_3\in\widehat{\calG}}\sum_{\gamma_4\in\widehat{\calG}}
\bfPhi_{\gamma_1\gamma_3}^{}\bfPhi_{\gamma_3}^*\bfPhi_{\gamma_2\gamma_4}^{}\bfPhi_{\gamma_4}^*\\
=\tfrac{D^2}{G^2R^2}\sum_{\gamma_3\in\widehat{\calG}}
\bfPhi_{\gamma_1\gamma_3}^{}\bfPhi_{\smash{\gamma_2^{-1}\gamma_3^{}}}^*
\sum_{\gamma_4\in\widehat{\calG}}\bfPhi_{\gamma_4}^{}\bfPhi_{\gamma_4}^*
=\tfrac{D}{GR}\sum_{\gamma_5\in\widehat{\calG}}
\bfPhi_{\gamma_1\gamma_2\gamma_5}^{}\bfPhi_{\gamma_5}^*
=\pi(\gamma_1\gamma_2).
\end{multline*}
When combined with the fact that
$\pi(1)
=\tfrac{D}{GR}\sum_{\gamma\in\widehat{\calG}}\bfPhi_{\gamma}^{}\bfPhi_{\gamma}^*
=\bfI$ and moreover that
\begin{equation*}
[\pi(\gamma_1)]^*
=\tfrac{D}{GR}\sum_{\gamma_2\in\widehat{\calG}}\bfPhi_{\gamma_2}^{}\bfPhi_{\gamma_1\gamma_2}^*
=\tfrac{D}{GR}\sum_{\gamma_3\in\widehat{\calG}}\bfPhi_{\smash{\gamma_1^{-1}\gamma_3^{}}}^{}\bfPhi_{\gamma_3}^*
=\pi(\gamma_1^{-1}),
\end{equation*}
we have that $\pi$ is indeed a homomorphism for $\widehat{\calG}$ into the group of unitary operators on $\calV$.
To show that it satisfies~\eqref{eq.pf of harmonic TFF characterization 4},
note that for any $\gamma_1,\gamma_2\in\widehat{\calG}$,
our block-circulant assumption gives that
$\bfPhi_{\gamma_3}^*\bfPhi_{\gamma_2}^{}=\bfPhi_{\gamma_1\gamma_3}^*\bfPhi_{\gamma_1\gamma_2}^{}$ for all $\gamma_3\in\widehat{\calG}$, and so
\begin{equation*}
\pi(\gamma_1)\bfPhi_{\gamma_2}
=\tfrac{D}{GR}\sum_{\gamma_3\in\widehat{\calG}}\bfPhi_{\gamma_1\gamma_3}^{}\bfPhi_{\gamma_3}^*\bfPhi_{\gamma_2}^{}
=\biggparen{\tfrac{D}{GR}\sum_{\gamma_3\in\widehat{\calG}}\bfPhi_{\gamma_1\gamma_3}^{}\bfPhi_{\gamma_1\gamma_3}^*}\bfPhi_{\gamma_1\gamma_2}^{}
=\bfPhi_{\gamma_1\gamma_2}.
\end{equation*}

(iii $\Rightarrow$ i)
Assume that there exists some subspace $\calU$ of $\calV$ and some homomorphism $\pi$ from $\widehat{\calG}$ into the group of unitary operators on $\calV$ such that $\calU_\gamma=\pi(\gamma)\calU$ for all \smash{$\gamma\in\widehat{\calG}$}.
We now factor $\pi$ into irreducible (one-dimensional) subrepresentations.
To explain,
since $\widehat{\calG}$ is abelian,
the unitaries \smash{$\set{\pi(\gamma)}_{\gamma\in\widehat{\calG}}$} commute,
implying there exists an orthonormal basis $\set{\bfv_d}_{d=1}^D$ for $\calV$ such that for each $d\in[D]=\set{1,\dotsc,D}$ and $\gamma\in\widehat{\calG}$ we have that $\bfv_d$ is an eigenvector for $\pi(\gamma)$.
The fact that
\smash{$\set{\pi(\gamma)}_{\gamma\in\widehat{\calG}}$} commute moreover implies that for each $d\in[D]$ the function $\gamma\mapsto\ip{\bfv_d}{\pi(\gamma)\bfv_d}$ that maps $\widehat{\calG}$ to the $d$th (unimodular) eigenvalue of $\pi(\gamma)$ is itself a homomorphism, and so is a character of $\widehat{\calG}$.
Identifying $\calG$ with the Pontryagin dual of $\widehat{\calG}$ in the standard way---via the isomorphism $g\mapsto(\gamma\mapsto\gamma(g))$---gives that for any $d\in[D]$ there exists a unique $g_d\in\calG$ such that $\pi(\gamma)\bfv_d =\gamma(g_d)\bfv_d$ for all $\gamma\in\widehat{\calG}$.
As such, defining $\calD_g:=\set{d\in[D]: g_d=g}$ for each $g\in\calG$ yields a partition $\set{\calD_g}_{g\in\calG}$ of $[D]$.
Now define $\calD:=\cup_{g\in\calG}(\set{g}\times\calD_g)$, as in Definition~\ref{def.harmonic matrix ensemble};
this has the effect of pairing each eigenvector index $d$ with its corresponding group element $g$.
In particular, reindexing the vectors $\set{\bfv_d}_{d=1}^D$ of our orthonormal eigenbasis as $\set{\bfv_{(g,d)}}_{(g,d)\in\calD}$, we have
\begin{equation}
\label{eq.pf of harmonic TFF characterization 1}
\pi(\gamma)\bfv_{(g,d)} =\gamma(g)\bfv_{(g,d)},\quad\forall\,\gamma\in\widehat{\calG},\,
(g,d)\in\calD.
\end{equation}
Now let $\bfV:\bbC^\calD\rightarrow\calV$,
$\bfV\bfy:=\sum_{(g,d)\in\calD}\bfy(g,d)\bfv_{(g,d)}$ be the (unitary) synthesis operator of $\set{\bfv_{(g,d)}}_{(g,d)\in\calD}$,
and let $\bfU:\bbC^\calR\rightarrow\calV$,
$\bfU\bfx:=\sum_{r\in\calR}\bfx(r)\bfu_r$ be the synthesis operator of some arbitrarily chosen orthonormal basis $\set{\bfu_r}_{r\in\calR}$ for $\calU$.
Here, since $\bfU$ is an isometry onto $\calU$,
$\set{\pi(\gamma)\bfU}_{\gamma\in\widehat{\calG}}$ are isometries onto the subspaces
\smash{$\set{\calU_\gamma}_{\gamma\in\widehat{\calG}}
=\set{\pi(\gamma)\,\calU}_{\gamma\in\widehat{\calG}}$} of our given $\TFF(D,G,R)$ for $\calV$.
Applying $\bfV^*$ to them yields unitarily equivalent isometries for a $\TFF(D,G,R)$ for $\bbC^\calD$, namely
\begin{equation}
\label{eq.pf of harmonic TFF characterization 2}
\set{\bfPhi_\gamma}_{\gamma\in\widehat{\calG}},
\quad\bfPhi_\gamma:=\bfV^*\pi(\gamma)\bfU.
\end{equation}
For any $g\in\calG$, now define a $\calR\times\calD_g$ matrix with entries
\begin{equation}
\label{eq.pf of harmonic TFF characterization 3}
\bfPsi_g(r,d):=(\tfrac{D}{R})^{\frac12}\ip{\bfu_r}{\bfv_{(g,d)}},
\quad\forall\, r\in\calR,\, d\in\calD_g.
\end{equation}
Now note that
\smash{$\set{\bfPsi_g}_{g\in\calG}$} and \smash{$\set{\bfPhi_\gamma}_{\gamma\in\widehat{\calG}}$} are related via~\eqref{eq.def of big frame}:
for any $\gamma\in\widehat{\calG}$, $g\in\calG$, $d\in\calD_g$ and $r\in\calR$,
combining~\eqref{eq.pf of harmonic TFF characterization 2},
\eqref{eq.pf of harmonic TFF characterization 1} and \eqref{eq.pf of harmonic TFF characterization 3} gives
\begin{multline*}
\bfPhi_\gamma((g,d),r)
=(\bfV^*\pi(\gamma)\bfU)((g,d),r)
=\ip{\bfdelta_{(g,d)}}{\bfV^*\pi(\gamma)\bfU\bfdelta_r}\\
=\ip{\bfv_{(g,d)}}{\pi(\gamma)\bfu_r}
=\ip{[\pi(\gamma)]^*\bfv_{(g,d)}}{\bfu_r}
=\ip{\overline{\gamma(g)}\bfv_{(g,d)}}{\bfu_r}\\
=\gamma(g)\ip{\bfv_{(g,d)}}{\bfu_r}
=(\tfrac{R}{D})^{\frac12}\gamma(g)\overline{(\tfrac{D}{R})^{\frac12}\ip{\bfu_r}{\bfv_{(g,d)}}}
=(\tfrac{R}{D})^{\frac12}\gamma(g)\overline{\bfPsi_g(r,d)}.
\end{multline*}
In summary, $\set{\calU_\gamma}_{\gamma\in\widehat{\calG}}
=\set{\pi(\gamma)\,\calU}_{\gamma\in\widehat{\calG}}$
is unitarily equivalent to a $\TFF(D,G,R)$ for $\bbC^\calD$ that has $\set{\bfPhi_\gamma}_{\gamma\in\widehat{\calG}}$ as isometries and moreover arises as the harmonic matrix ensemble generated by $\set{\bfPsi_g}_{g\in\calG}$.
Here, Theorem~\ref{thm.small TFF yields harmonic} moreover gives that \smash{$\set{\bfPsi_g}_{g\in\calG}$} are themselves the isometries of a $\TFF(R,G,\set{D_g}_{g\in\calG})$ for $\bbC^\calR$ where $D_g=\#(\calD_g)$ for all $g$.
\end{proof}

\begin{definition}
\label{def.harmonic}
Letting $\widehat{\calG}$ be the Pontryagin dual of a finite abelian group $\calG$,
we say that a $\TFF(D,G,R)$ \smash{$\set{\calU_\gamma}_{\gamma\in\widehat{\calG}}$} for a complex Hilbert space $\calV$ is \textit{harmonic} if it satisfies any one of the three equivalent properties given in Theorem~\ref{thm.harmonic TFF characterization}.
In this case, we say a TFF $\set{\calW_g}_{g\in\calG}$ for $\bbC^\calR$ is a corresponding \textit{generating TFF} if it is formed by isometries $\set{\bfPsi_g}_{g\in\calG}$ that generate a harmonic matrix ensemble \smash{$\set{\bfPhi_\gamma}_{\gamma\in\widehat{\calG}}$}, and furthermore, there exists a unitary operator $\bfV:\bbC^\calD\rightarrow\calV$ such that the isometries \smash{$\set{\bfV\bfPhi_\gamma}_{\gamma\in\widehat{\calG}}$} form the original TFF $\set{\calU_\gamma}_{\gamma\in\widehat{\calG}}$.
\end{definition}

We now give some facts that follow easily by combining Theorems~\ref{thm.small TFF yields harmonic}, \ref{thm.harmonic TFF characterization} and Definition~\ref{def.harmonic}:

\begin{corollary}
\label{cor.complements}
Any harmonic $\TFF(D,G,R)$ with $D<GR$ has a Naimark complement which is a harmonic $\TFF(GR-D,G,R)$.
Any harmonic $\TFF(D,G,R)$ with $R<D$ has a spatial complement which is a harmonic $\TFF(D-R,G,R)$.
\end{corollary}

It was recently shown that taking iterated alternating Naimark and spatial complements of any $\TFF(D,N,R)$ with either $N>4$ or $N=4$ and $D\neq 2R$ yields an infinite number of TFFs with distinct parameters~\cite{FickusMW21}.
From Corollary~\ref{cor.complements}, we see that if one TFF in such a \textit{Naimark-spatial orbit} is harmonic then all TFFs in this orbit are as well.
Recalling that any Naimark or spatial complement of an ECTFF is itself an ECTFF,
this provides a simple way to generate an infinite number of distinct harmonic ECTFFs provided the order of the underlying group is at least $5$.
We caution however that the only members of such an orbit that might be EITFFs necessarily have ``minimal" $R$ parameters.
For example, taking iterated alternating Naimark and spatial complements of the harmonic $\EITFF(4,5,2)$ discussed in Example~\ref{ex.EITFF(4,5,2)} yields harmonic ECTFFs with $(D,N,R)$ parameters $(4,5,2)$, $(6,5,2)$, $(6,5,4)$, $(14,5,4)$, $(14,5,10)$, etc., but of these, only those with parameters $(4,5,2)$ and $(6,5,2)$ are EITFFs.
These ideas from~\cite{FickusMW21} were in part motivated by a desire to certify when a certain generalization of an ECTFF construction method of~\cite{King16} yielded genuinely new ECTFFs.
As we now discuss, it turns out that both this construction method~\cite{King16} and its generalization~\cite{FickusMW21} are essentially special cases of harmonic matrix ensembles.

\begin{remark}
Here regard the group $\calG$ of Theorem~\ref{thm.small TFF yields harmonic} to be a subgroup of some larger $K$-element finite abelian group $\calK$ with Pontryagin dual $\widehat{\calK}$.
The corresponding \textit{annihilator} of $\calG$ with respect to $\calK$ is  $\calG^\perp:=\set{\kappa\in\widehat{\calK}: \kappa(g)=1,\ \forall\, g\in\calG}$.
It is well known that $\calG^\perp$ is a subgroup of $\widehat{\calK}$ of order $R:=K/G$ and, in fact, is isomorphic to the Pontryagin dual of $\calK/\calG$.
Let $\calD$ be a nonempty $D$-element subset of $\calK$ and let \smash{$\set{\bfphi_\kappa}_{\kappa\in\widehat{\calK}}$} be the traditional harmonic frame for $\bbC^\calD$, having \smash{$\bfphi_\kappa(d):=\frac1{\sqrt{D}}\kappa(d)$} for all $\kappa\in\widehat{\calK}$ and $d\in\calD$.
We now partition the vectors of this UNTF for $\bbC^\calD$ according to the cosets of $\calG^\perp$ in $\widehat{\calK}$.
To be precise,
for any $\kappa\in\widehat{\calK}$, let $\bfPhi_\kappa$ be the $\calD\times\calG^\perp$ synthesis operator of $\set{\bfphi_{\kappa\eta}}_{\eta\in\calG^\perp}$,
and note that for any $\kappa_1,\kappa_2\in\widehat{\calK}$ and $\eta_1,\eta_2\in\calG^\perp$, the $(\eta_1,\eta_2)$th entry of the corresponding $(\kappa_1,\kappa_2)$th cross-Gram matrix is
\begin{multline}
\label{eq.other harmonic 1}
(\bfPhi_{\kappa_1}^*\bfPhi_{\kappa_2}^{})(\eta_1,\eta_2)
=\ip{\bfphi_{\kappa_1\eta_1}}{\bfphi_{\kappa_2\eta_2}}
=\tfrac1D\sum_{d\in\calD}\overline{(\kappa_1\eta_1)(d)}(\kappa_2\eta_2)(d)\\
=\tfrac1D\sum_{k\in\calK}\overline{(\kappa_1^{}\kappa_2^{-1}\eta_1\eta_2^{-1})(k)}\bfone_{\calD}(k)
=\tfrac1D(\bfGamma_\calK^*\bfone_\calD)(\kappa_1^{}\kappa_2^{-1}\eta_1\eta_2^{-1}),
\end{multline}
where $\bfGamma_\calK$ is the $\calK\times\widehat{\calK}$ character table of $\calK$ and $\bfone_\calD$ is the characteristic function of $\calD$.
In particular, each $\bfPhi_\kappa$ is an isometry if and only if $(\bfGamma_\calK^*\bfone_\calD)(\eta)=0$ for all $\eta\in\calG^\perp$, $\eta\neq 1$;
by the Poisson summation formula, this equates to the cardinality of $\calD^{(k)}:=\calG\cap(\calD-k)$ being independent of $k$~\cite{FickusMW21}.
In this case, the fact that \smash{$\set{\bfphi_\kappa}_{\kappa\in\widehat{\calK}}$} is a UNTF for $\bbC^\calD$ implies that
\smash{$\set{\bfPhi_\kappa}_{\kappa\calG^\perp\in\widehat{\calK}/\calG^\perp}$} are the isometries of a $\TFF(D,G,R)$ for $\bbC^\calD$;
here, \smash{$\set{\bfPhi_\kappa}_{\kappa\calG^\perp\in\widehat{\calK}/\calG^\perp}$} is (abused) notation for a subsequence of \smash{$\set{\bfPhi_\kappa}_{\kappa\in\widehat{\calK}}$} that is indexed by some transversal (choice of coset representatives) of $\calG^\perp$ in $\widehat{\calK}$.

In~\cite{GodsilR09} it was shown that such matrices \smash{$\set{\bfPhi_\kappa}_{\kappa\calG^\perp\in\widehat{\calK}/\calG^\perp}$} are the (unitary) synthesis operators of $D$ \textit{mutually unbiased bases} for $\bbC^\calD$ if and only if $\calD$ is a \textit{semiregular relative difference set} for $\calK$.
In~\cite{King16} it was more generally shown that if $\calD$ is a \textit{semiregular divisible difference set} for $\calK$ then
\smash{$\set{\bfPhi_\kappa}_{\kappa\calG^\perp\in\widehat{\calK}/\calG^\perp}$} are the isometries of an $\ECTFF(D,G,R)$ for $\bbC^\calD$ with the remarkable property that $\abs{\ip{\bfphi_{\kappa_1\eta_1}}{\bfphi_{\kappa_2\eta_2}}}$ is constant over all $\kappa_1,\kappa_2\in\widehat{\calK}$ that lie in distinct cosets of $\calG^\perp$ and all \smash{$\eta_1,\eta_2\in\calG^\perp$}.
In~\cite{FickusMW21} it was more generally shown that
\smash{$\set{\bfPhi_\kappa}_{\kappa\calG^\perp\in\widehat{\calK}/\calG^\perp}$} are the isometries of an $\ECTFF(D,G,R)$ for $\bbC^\calD$ if and only if the subsets $\set{\calD^{(k)}}_{k+\calG\in\calK/\calG}$ of $\calG$ all have the same cardinality and form a \textit{difference family} for $\calG$, having
\smash{$\sum_{k+\calG\in\calK/\calG}\abs{(\bfGamma^*\bfone_{\calD^{(k)}})(\gamma)}^2$} be constant over all \smash{$\gamma\in\widehat{\calG}$}, $\gamma\neq1$;
they are moreover the isometries of an $\EITFF(D,G,R)$ for $\bbC^\calD$ if and only if each $\calD^{(k)}$ is a difference set for $\calG$.
(A similar idea has also been used to partition certain harmonic ETFs into regular simplices~\cite{FickusJKM18,FickusS20} or, more generally, into \textit{mutually unbiased ETFs}~\cite{FickusM21}.)

Such TFFs are ``harmonic" in the sense that their isometries are constructed by extracting rows from the character table of a finite abelian group.
In light of~\eqref{eq.other harmonic 1}, each cross-Gram matrix of the isometries \smash{$\set{\bfPhi_\kappa}_{\kappa\calG^\perp\in\widehat{\calK}/\calG^\perp}$} is $\calG^\perp$-circulant.
Here, \eqref{eq.other harmonic 1} further implies that the fusion Gram matrix of these isometries is block-circulant if the representatives of the cosets of $\calG^\perp$ in $\widehat{\calK}$ can themselves be chosen to form some subgroup of $\widehat{\calK}$ of order $G$.
In particular, by Theorem~\ref{thm.harmonic TFF characterization}, the ``harmonic" TFFs of~\cite{King16,FickusMW21} are also harmonic in the sense of Definition~\ref{def.harmonic} if $\widehat{\calK}$ is the (internal) direct product of $\calG^\perp$ and some other subgroup of $\widehat{\calK}$.
This is often possible, but not always so.
That said, one can show that whenever an $\ECTFF(D,G,R)$ or $\EITFF(D,G,R)$ arises from the method of~\cite{King16,FickusMW21}, there is also a harmonic ECTFF or EITFF (in the sense of Definition~\ref{def.harmonic}) with the same parameters.
Furthermore, the $\EITFF(D,N,R)$ construction in~\cite{King16,FickusMW21} is necessarily ``tensor sized'', in the sense that $R$ divides $D$ and an $\ETF(\tfrac{D}{R},N)$ exists.
This defect does not appear in the more general approach of Definition~\ref{def.harmonic}, as illustrated by the results of Section~4 below.
\end{remark}

We now characterize when a TFF that is harmonic in the sense of Definition~\ref{def.harmonic} is an EITFF or an ECTFF in terms of a matrix-valued DFT of the projections of the corresponding generating TFF.

\begin{theorem}
\label{thm.EITFF}
Let $\widehat{\calG}$ be the Pontryagin dual of a finite abelian group $\calG$,
and let $\set{\bfM_\gamma}_{\gamma\in\widehat{\calG}}$ be the entrywise DFT of any $\calG$-indexed sequence $\set{\bfP_g}_{g\in\calG}$ of $\calR\times\calR$ orthogonal projection matrices:
\begin{equation}
\label{eq.DFT of projections}
\bfM_\gamma:=\sum_{g\in\calG}\overline{\gamma(g)}\bfP_g,
\quad\forall\,\gamma\in\widehat{\calG}.
\end{equation}
Then the images of $\set{\bfP_g}_{g\in\calG}$ form a TFF for $\bbC^\calR$ if and only if $\bfM_1=A\bfI$ for some $A>0$.
Assuming this to be the case, and defining
\begin{equation*}
D:=\sum_{g\in\calG}D_g
\quad\text{where}\quad
D_g:=\rank(\bfP_g)\quad\forall\,g\in\calG,
\qquad
G:=\#(\calG),
\qquad
R:=\#(\calR),
\end{equation*}
we have that $A=\frac DR$, that any isometries $\set{\bfPsi_g}_{g\in\calG}$ of this $\TFF(R,G,\set{\calD_g}_{g\in\calG})$ generate a harmonic matrix ensemble $\set{\bfPhi_\gamma}_{\gamma\in\widehat{\calG}}$ that forms a harmonic $\TFF(D,G,R)$, and that
\begin{equation}
\label{eq.pf of harmonic EITFF characterization 1}
\bfPhi_{\gamma_1}^*\bfPhi_{\gamma_2}^{}
=\tfrac{R}{D}\bfM_{\gamma_1^{}\gamma_2^{-1}},
\quad\forall\,\gamma_1,\gamma_2\in\widehat{\calG}.
\end{equation}
Furthermore:
\begin{enumerate}
\renewcommand{\labelenumi}{(\alph{enumi})}
\item
The following are equivalent:
\begin{enumerate}
\renewcommand{\labelenumii}{(\roman{enumii})}
\item
this harmonic $\TFF(D,G,R)$ is an $\EITFF(D,G,R)$;\smallskip
\item
there exists a scalar $B$ such that $\bfM_\gamma^*\bfM_\gamma^{}=B\bfI$ for all $\gamma\in\widehat{\calG}$, $\gamma\neq 1$;\smallskip
\item
there exists a scalar $C$ such that $\sum_{g'\in\calG}\bfP_{g'}\bfP_{g+g'}=C\bfI$ for all $g\in\calG$, $g\neq0$.
\end{enumerate}
Moreover, if (i)--(iii) hold, then $B=\tfrac{D(GR-D)}{R^2(G-1)}$ and $C=\tfrac{D(D-R)}{R^2(G-1)}$.
\item
The following are equivalent:
\begin{enumerate}
\setcounter{enumii}{3}
\renewcommand{\labelenumii}{(\roman{enumii})}
\item
this harmonic $\TFF(D,G,R)$ is an $\ECTFF(D,G,R)$;\smallskip
\item
there exists a scalar $B$ such that $\tfrac1R\norm{\bfM_\gamma}_\Fro^2=\tfrac1R\Tr(\bfM_\gamma^*\bfM_\gamma^{})=B$ for all $\gamma\in\widehat{\calG}$, $\gamma\neq 1$;\smallskip
\item
there exists a scalar $C$ such that $\sum_{g'\in\calG}\tfrac1R\Tr(\bfP_{g'}\bfP_{g+g'})=C$ for all $g\in\calG$, $g\neq0$.
\end{enumerate}
Moreover, if (iv)--(vi) hold, then $B=\tfrac{D(GR-D)}{R^2(G-1)}$ and $C=\tfrac{D(D-R)}{R^2(G-1)}$.

In particular, if $\set{\bfP_g}_{g\in\calG}$ are the projections of an $\ECTFF(R,G,\frac DG)$, then they generate a harmonic $\ECTFF(D,G,R)$ and $\ip{\bfM_{\gamma_1}}{\bfM_{\gamma_2}}_\Fro=0$ for $\gamma_1,\gamma_2\in\widehat{\calG}$ with $\gamma_1\neq\gamma_2$.
\item
This harmonic $\TFF(D,G,R)$ is real if
$\bfP_{g}(r_1,r_2)=\overline{\bfP_{-g}(r_1,r_2)}$ for all $g\in\calG$ and $r_1,r_2\in\calR$.
\item
For any $\gamma\in\widehat{\calG}$,
$\displaystyle\Tr(\bfM_\gamma)
=\sum_{g\in\calG}\overline{\gamma(g)}D_g$.
\end{enumerate}
\end{theorem}

\begin{proof}
Let $\calR$ be any nonempty finite set, let $\set{\bfP_g}_{g\in\calG}$ be any $\calG$-indexed sequence of $\calR\times\calR$ orthogonal projection matrices and define $D$, $D_g$, $G$ and $R$ as above.
For each $g\in\calG$ let \smash{$\bfPsi_g$} be any \smash{$\calR\times\calD_g$} isometry such that \smash{$\bfP_g=\bfPsi_g^{}\bfPsi_g^*$}.
By definition, the subspaces
\smash{$\set{\bfPsi_g(\bbC^{\calD_g})}_{g\in\calG}
=\set{\bfP_g(\bbC^{\calR})}_{g\in\calG}$} form a TFF for $\bbC^\calR$ if and only if
$\bfM_1=\sum_{g\in\calG}\bfP_g=\sum_{g\in\calG}\bfPsi_g^{}\bfPsi_g^*$ equals $A\bfI$ for some $A>0$.
Assume this holds for the remainder of this proof.
Here, \smash{$AR=\Tr(A\bfI)=\sum_{n\in\calN}\Tr(\bfP_n)=\sum_{n\in\calN}D_n=D$} and so $A=\frac DR$.
By Theorem~\ref{thm.small TFF yields harmonic} and Definition~\ref{def.harmonic},
$\set{\bfPsi_g}_{g\in\calG}$ generate the isometries \smash{$\set{\bfPhi_\gamma}_{\gamma\in\widehat{\calG}}$} of a harmonic $\TFF(D,G,R)$ for $\bbC^\calD$ where $\calD:=\cup_{g\in\calG}\set{g}\times\calD_g$ and,
by~\eqref{eq.cross Gram of big frame} and \eqref{eq.DFT of projections},
\begin{equation*}
\bfPhi_{\gamma_1}^*\bfPhi_{\gamma_2}^{}
=\tfrac{R}{D}\sum_{g\in\calG}(\gamma_1^{-1}\gamma_2^{})(g)\bfPsi_g^{}\bfPsi_g^*
=\tfrac{R}{D}\sum_{g\in\calG}\overline{(\gamma_1^{}\gamma_2^{-1})(g)}\bfP_g
=\tfrac{R}{D}\bfM_{\gamma_1^{}\gamma_2^{-1}}
\end{equation*}
for any $\gamma_1,\gamma_2\in\widehat{\calG}$.
This immediately implies that this harmonic $\TFF(D,G,R)$ is an EITFF if and only if there exists $B\geq0$ such that $\bfM_\gamma$ is $\sqrt{B}$ times a unitary for every $\gamma\neq1$, or equivalently, such that $\bfM_\gamma^*\bfM_\gamma^{}=B\bfI$ for all $\gamma\neq1$.
Thus, (i) equates to (ii).
To show that (ii) equates to (iii),
let $\set{\bfA_g}_{g\in\calG}$ be the \textit{autocorrelation} of $\set{\bfP_g}_{g\in\calG}$,
namely the sequence $\set{\bfA_g}_{g\in\calG}$ of $\calR\times\calR$ matrices defined by \smash{$\bfA_g:=\sum_{g'\in\calG}\bfP_{g'}\bfP_{g'+g}$} for each $g\in\calG$,
and note that for any \smash{$\gamma\in\widehat{\calG}$},
\begin{multline}
\label{eq.pf of harmonic EITFF characterization 2}
\bfM_\gamma^*\bfM_\gamma^{}
=\biggbracket{\,\sum_{g_1\in\calG}\overline{\gamma(g_1)}\bfP_{g_1}}^*
\biggbracket{\,\sum_{g_2\in\calG}\overline{\gamma(g_2)}\bfP_{g_2}}
=\sum_{g_1\in\calG}\sum_{g_2\in\calG}\gamma(g_1-g_2)\bfP_{g_1}\bfP_{g_2}\\
=\sum_{g_1\in\calG}\sum_{g\in\calG}\gamma(-g)\bfP_{g_1}\bfP_{g_1+g}
=\sum_{g\in\calG}\overline{\gamma(g)}\bfA_g.
\end{multline}
Since $\bfM_1=\sum_{g\in\calG}\bfP_g=\frac DR\bfI$, (ii) thus holds for a given $B$ if and only if
\begin{equation}
\label{eq.pf of harmonic EITFF characterization 3}
\sum_{g\in\calG}\overline{\gamma(g)}\bfA_g
=\bfM_\gamma^*\bfM_\gamma^{}
=\left\{\begin{array}{rl}
(\frac DR)^2\,\bfI,&\ \gamma=1\\
B\,\bfI,&\ \gamma\neq1\end{array}\right\}
=\Bigset{\bigbracket{(\tfrac DR)^2-B}\bfdelta_1(\gamma)+B}\bfI,
\quad
\forall\,\gamma\in\widehat{\calG}.
\end{equation}
Here, $\sum_{g\in\calG}\overline{\gamma(g)}\bfA_g$ is the (entrywise) DFT of $\set{\bfA_g}$, and so taking the (entrywise) inverse DFT of this equation gives,
for any given scalar $B$, that (ii) equates to having
\begin{equation}
\label{eq.pf of harmonic EITFF characterization 4}
\bfA_g
=\tfrac1G\sum_{\gamma\in\widehat{\calG}}\gamma(g)
\Bigset{\bigbracket{(\tfrac DR)^2-B}\bfdelta_1(\gamma)+B}\bfI
=\Bigset{\tfrac1G\bigbracket{(\tfrac DR)^2-B}+B\delta_0(g)}\bfI,
\quad\forall\,g\in\calG.
\end{equation}
In particular, if (ii) holds then $\bfA_g=C\bfI$ for all $g\neq0$
where
\begin{equation}
\label{eq.pf of harmonic EITFF characterization 5}
C=\tfrac1G\bigbracket{(\tfrac DR)^2-B},
\end{equation}
implying (iii).
Conversely, if (iii) holds,
namely if there exists a scalar $C$ such that $\bfA_g=C\bfI$ for all $g\neq0$,
then (ii) holds for the unique choice of $B$ that satisfies~\eqref{eq.pf of harmonic EITFF characterization 5} since, in this case, \eqref{eq.pf of harmonic EITFF characterization 4} holds for all $g\neq0$, while for $g=0$,
\begin{multline}
\label{eq.pf of harmonic EITFF characterization 6}
\bfA_0
=\sum_{g\in\calG}\bfA_g-\sum_{g\neq 0}\bfA_g
=\sum_{g\in\calG}\sum_{g'\in\calG}\bfP_{g'}\bfP_{g'+g}-\sum_{g\neq 0}C\bfI
=\biggparen{\,\sum_{g\in\calG}\bfP_g}^2-(G-1)C\bfI\\
=(\tfrac DR)^2\bfI-\tfrac{G-1}G\bigbracket{(\tfrac DR)^2-B}\bfI
=\Bigset{\tfrac1G\bigbracket{(\tfrac DR)^2-B}+B}\bfI.
\end{multline}
Altogether, (i), (ii) and (iii) are equivalent,
and the requisite scalars $B$ and $C$ satisfy~\eqref{eq.pf of harmonic EITFF characterization 5}.

The equivalence of (iv), (v) and (vi) is proven similarly.
In brief,
by~\eqref{eq.pf of harmonic EITFF characterization 1},
the harmonic $\TFF(D,G,R)$ under consideration here is an ECTFF if and only if there exists a scalar $B$ such that \smash{$\tfrac1R\norm{\bfM_\gamma}_\Fro^2=\tfrac1R(\bfM_\gamma^*\bfM_\gamma^{})=B$} for all $\gamma\neq1$,
and so (iv) equates to (v).
Moreover, taking the trace of~\eqref{eq.pf of harmonic EITFF characterization 2} gives
\smash{$\tfrac1R\Tr(\bfM_\gamma^*\bfM_\gamma^{})
=\sum_{g\in\calG}\overline{\gamma(g)}\tfrac1R\Tr(\bfA_g)$} for any $\gamma\in\widehat{\calG}$,
at which point the trace-based analogs of \eqref{eq.pf of harmonic EITFF characterization 3}, \eqref{eq.pf of harmonic EITFF characterization 4},
and~\eqref{eq.pf of harmonic EITFF characterization 6} imply that (v) holds for a given $B$ if and only if (vi) holds for a given $C$ where $B$ and $C$ are again related by~\eqref{eq.pf of harmonic EITFF characterization 5}.

When the harmonic $\TFF(D,G,R)$ here is an ECTFF (which includes the case where it is an EITFF) there is but one choice of $B$:
from~\eqref{eq.chordal Welch} and~\eqref{eq.pf of harmonic EITFF characterization 1},
recall that necessarily
\begin{equation*}
\tfrac{GR-D}{D(G-1)}
=\tfrac1R\norm{\bfPhi_{\gamma_1}^*\bfPhi_{\gamma_2}^{}}_\Fro^2
=\tfrac1R\norm{\tfrac{R}{D}\bfM_{\gamma_1^{}\gamma_2^{-1}}}_\Fro^2
=\tfrac{R^2}{D^2}\tfrac1R\norm{\bfM_{\gamma_1^{}\gamma_2^{-1}}}_\Fro^2
=\tfrac{R^2}{D^2}B,
\end{equation*}
whenever $\gamma_1\neq\gamma_2$, and so necessarily $B=\tfrac{D(GR-D)}{R^2(G-1)}$.
Substituting this into~\eqref{eq.pf of harmonic EITFF characterization 5} gives
\begin{equation*}
C
=\tfrac1G\bigbracket{(\tfrac DR)^2-B}
=\tfrac1G\bigbracket{(\tfrac DR)^2-\tfrac{D(GR-D)}{R^2(G-1)}}
=\tfrac{D[D(G-1)-(GR-D)]}{R^2G(G-1)}
=\tfrac{D(D-R)}{R^2(G-1)}.
\end{equation*}
This completes the proof of (a).
For the remaining conclusions of (b), note that if $\set{\bfP_g}_{g\in\calG}$ are the projections of an $\ECTFF(R,G,\frac DG)$ then $\Tr(\bfP_{g_1}\bfP_{g_2})$ is constant over all $g_1\neq g_2$, implying that (vi) holds, and so the harmonic $\TFF(D,G,R)$ here is an ECTFF.
In this case,
(vi) in fact gives that \smash{$\Tr(\bfP_{g_1}\bfP_{g_2})=\frac{RC}{G}=\frac{D(D-R)}{GR(G-1)}$} whenever $g_1\neq g_2$, and so more generally for any $g_1,g_2\in\calG$,
\begin{equation*}
\Tr(\bfP_{g_1}\bfP_{g_2})
=\tfrac{D}{G}\left\{\begin{array}{cl}
1,&\ g_1=g_2\\
\tfrac{D-R}{R(G-1)},&\ g_1\neq g_2
\end{array}\right\}
=\tfrac{D}{GR(G-1)}[(GR-D)\bfdelta_{g_1}(g_2)+(D-R)].
\end{equation*}
As such, for any $\gamma_1,\gamma_2\in\widehat{\calG}$,
\begin{align*}
\ip{\bfM_{\gamma_1}}{\bfM_{\gamma_2}}_\Fro
&=\Tr\biggbracket{
\biggparen{\,\sum_{g_1\in\calG}\overline{\gamma_1(g_1)}\bfP_{g_1}}^*
\sum_{g_2\in\calG}\overline{\gamma_2(g_2)}\bfP_{g_2}}\\
&=\tfrac{D}{GR(G-1)}\sum_{g_1\in\calG}\sum_{g_2\in\calG}\gamma_1(g_1)\overline{\gamma_2(g_2)}
[(GR-D)\bfdelta_{g_1}(g_2)+(D-R)]\\
&=\tfrac{D}{GR(G-1)}\sum_{g_1\in\calG}\gamma_1(g_1)
\bigbracket{(GR-D)\overline{\gamma_2(g_1)}+G(D-R)\bfdelta_1(\gamma_2)}\\
&=\tfrac{D}{GR(G-1)}\bigbracket{(GR-D)\ip{\gamma_2}{\gamma_1}+G^2(D-R)\bfdelta_1(\gamma_1)\bfdelta_1(\gamma_2)}\\
&=\left\{\begin{array}{cl}
\tfrac{D^2}{R}&\ \gamma_1=\gamma_2=1,\smallskip\\
\tfrac{D(GR-D)}{R(G-1)},&\ \gamma_1=\gamma_2\neq1,\smallskip\\
0,&\ \gamma_1\neq\gamma_2.
\end{array}\right.
\end{align*}
In particular, in this case, the matrices $\set{\bfM_\gamma}_{\gamma\in\widehat{\calG}}$ are pairwise orthogonal; this observation has previously been made in the special case where $R=1$ as a method of proving \textit{Gerzon's bound}~\cite{BalanBCE09}.

For (c), note that if $\bfP_g$ is the (entrywise) conjugate of $\bfP_{-g}$ for all $g$ then each matrix $\bfM_\gamma$ is real:
for any $r_1,r_2\in\calR$,
\begin{multline*}
\overline{\bfM_{\gamma}(r_1,r_2)}
=\overline{\sum_{g\in\calG}\overline{\gamma(g)}\bfP_g(r_1,r_2)}
=\sum_{g\in\calG}\gamma(g)\overline{\bfP_g(r_1,r_2)}\\
=\sum_{g\in\calG}\overline{\gamma(-g)}\bfP_{-g}(r_1,r_2)
=\sum_{g\in\calG}\overline{\gamma(g)}\bfP_{g}(r_1,r_2)
=\bfM_\gamma(r_1,r_2).
\end{multline*}
In light of~\eqref{eq.pf of harmonic EITFF characterization 1},
this means our harmonic $\TFF(D,G,R)$ is real, having a real-valued fusion Gram matrix.
Finally, for any $\gamma\in\widehat{\calG}$,
(d) follows by taking the trace of~\eqref{eq.DFT of projections}:
\begin{equation*}
\Tr(\bfM_\gamma)
=\sum_{g\in\calG}\overline{\gamma(g)}\Tr(\bfP_g)
=\sum_{g\in\calG}\overline{\gamma(g)}\rank(\bfP_g)
=\sum_{g\in\calG}\overline{\gamma(g)}D_g.
\qedhere
\end{equation*}
\end{proof}

As discussed in Example~\ref{ex.EITFF(3,7,1)},
the harmonic TFFs of Definition~\ref{def.harmonic} equate to traditional harmonic frames in the special case where $R=1$.
In this case, each projection $\bfP_g$ is a $1\times1$ matrix whose sole entry is either $1$ or $0$, implying there exists some subset $\calD$ of $\calG$ such that $\bfP_g(1,1)=\bfone_\calD(g)$ for all $g\in\calG$.
By Theorem~\ref{thm.EITFF}(iii), a harmonic $\TFF(D,N,1)$ is thus an EITFF if and only if there exists some scalar $C$ such that for all $g\neq0$,
\begin{equation*}
C
=\sum_{g'\in\calG}\bfone_\calD(g')\bfone_\calD(g+g')
=\#[\calD\cap(\calD-g)]
=\#\set{(d_1,d_2)\in\calD\times\calD: g=d_1-d_2},
\end{equation*}
where the final equality follows from the fact that $d\mapsto(d+g,d)$ is a bijection from $\calD\cap(\calD-g)$ onto $\set{(d_1,d_2)\in\calD\times\calD: g=d_1-d_2}$.
In particular, in the special case where $R=1$, the equivalence of (i) and (iii) of Theorem~\ref{thm.EITFF} reduces to well known equivalence~\cite{Konig99,StrohmerH03,XiaZG05,DingF07} between harmonic ETFs and difference sets.
In this case, Theorem~\ref{thm.EITFF}(a) gives that $C$ is necessarily \smash{$\tfrac{D(D-1)}{G-1}$},
a fact which is often proven in the context of difference sets using a simple counting argument.
Moreover, since here
$\bfM_\gamma(1,1)
=\sum_{g\in\calG}\overline{\gamma(g)}\bfP_g(1,1)
=\sum_{g\in\calG}\overline{\gamma(g)}(\bfone_\calD)(g)
=(\bfGamma^*\bfone_\calD)(\gamma)$
for all \smash{$\gamma\in\widehat{\calG}$},
Theorem~\ref{thm.EITFF}(ii) gives that this equates to \smash{$\abs{(\bfGamma^*\bfone_\calD)(\gamma)}^2$} being constant over all $\gamma\neq1$.
Thus, when $R=1$, the equivalence of (ii) and (iii) reduces to the classical character-based characterization of difference sets~\cite{Turyn65}.
From this perspective, projections $\set{\bfP_g}_{g\in\calG}$ that satisfy (i)--(iii) form a type of operator-theoretic generalization of a difference set:

\begin{definition}
\label{def.difference projections}
We say $\calR\times\calR$ orthogonal projection matrices $\set{\bfP_g}_{g\in\calG}$ are \textit{difference projections} of the finite abelian group $\calG$ if their images form a TFF for $\bbC^\calR$ and they satisfy any one of the three equivalent properties (i)--(iii) of Theorem~\ref{thm.EITFF}.
\end{definition}

In the next section, we construct some nontrivial difference projections for the additive groups of finite fields so as to obtain some new (harmonic) EITFFs.
We note the standard methods for converting one difference set for $\calG$ into another such set---translations, automorphisms and complementation---generalize to this setting.
In particular, if $\set{\bfP_g}_{g\in\calG}$ sum to a positive multiple of the identity and satisfy the condition of Theorem~\ref{thm.EITFF}(ii) then the same is true for $\set{\bfP_{g+g'}}_{g\in\calG}$ for any $g'\in\calG$ as well as for $\set{\bfP_{\sigma(g)}}_{g\in\calG}$ for any automorphism $\sigma$ of $\calG$ since,
for any $\gamma\in\widehat{\calG}$,
\begin{equation*}
\sum_{g\in\calG}\overline{\gamma(g)}\bfP_{g+g'}
=\sum_{g\in\calG}\overline{\gamma(g-g')}\bfP_{g}
=\gamma(g')\bfM_\gamma,
\qquad
\sum_{g\in\calG}\overline{\gamma(g)}\bfP_{\sigma(g)}
=\sum_{g\in\calG}\overline{(\gamma\circ\sigma^{-1})(g)}\bfP_{g}
=\bfM_{\gamma\circ\sigma^{-1}}.
\end{equation*}
Moreover, by Theorem~\ref{thm.small TFF yields harmonic}(a) and Corollary~\ref{cor.complements},
any harmonic $\EITFF(D,G,R)$ with $D<GR$ has a Naimark complement which is a harmonic $\EITFF(GR-D,G,R)$,
and the corresponding generating TFFs are spatial-complementary.
Thus, if $\set{\bfP_g}_{g\in\calG}$ are difference projections with $\bfP_g\neq\bfI$ for at least one choice of $g$ then so are $\set{\bfI-\bfP_g}_{g\in\calG}$.

\begin{remark}
Rather than regarding a harmonic ETF as arising from a generating TFF for a space of dimension one,
we may instead regard it as a generating ECTFF that, by Theorem~\ref{thm.EITFF}(b), yields a harmonic ECTFF with larger parameters.
Remarkably, one can show that the resulting ECTFFs are a special case of those constructed by Zauner~\cite{Zauner99} from \textit{balanced incomplete block designs} (BIBDs).
\end{remark}

\section{Constructions of harmonic equichordal and equi-isoclinic tight fusion frames}

In this section we use Theorem~\ref{thm.EITFF} to construct new infinite families of (harmonic) EITFFs.
In particular,
we construct difference projections (Definition~\ref{def.difference projections}) by exploiting the classical theory of Gauss sums over finite fields.
We start by reviewing some preliminaries~\cite{LidlN97}.

\subsection{Preliminaries}

Recall that the Pontryagin dual $\widehat{\calG}$ of a finite abelian group $\calG$ is the set of all characters on $\calG$,
namely the set of all homomorphisms $\gamma$ from $\calG$ into $\bbT:=\set{z\in\bbC: \abs{z}=1}$.
Let $Q$ be any prime power, and let $\chi$ be any \textit{multiplicative character} of the $Q$-element field $\bbF_Q$, that is, any character of its (cyclic) multiplicative group \smash{$\bbF_Q^\times$}.
Let $\gamma$ be an \textit{additive character} of $\bbF_Q$,
that is, a character of the additive group of $\bbF_Q$.
Now consider the inner product of the restriction of $\gamma$ to $\bbF_Q^\times$ with $\chi$:
\begin{equation*}
\ip{\gamma}{\chi}_{\times}
=\sum_{x\in\bbF_Q^\times}\overline{\gamma(x)}\chi(x).
\end{equation*}
(Throughout, we place a ``$\times$" subscript on an inner product when its arguments should be regarded as vectors whose entries are indexed by $\bbF_Q^\times$.)
Complex numbers of this form are types of Gauss sums,
and we now review some facts about them that we will need.

We denote the Pontryagin duals of the additive and multiplicative groups of $\bbF_Q$ as $\widehat{\bbF}_Q$ and $\widehat{\bbF}_Q^\times$, respectively.
Recall that if $\calG$ is any finite abelian group then \smash{$\set{\gamma}_{\gamma\in\widehat{\calG}}$} is an equal norm orthogonal basis for $\bbC^\calG$ that satisfies $\abs{\gamma(g)}=1$ for all $g\in\calG$, $\gamma\in\widehat{\calG}$.
Since restricting the identity additive character $1\in\widehat{\bbF}_Q$ to $\bbF_Q^\times$ yields $1\in\widehat{\bbF}_Q^\times$ we have that $\ip{1}{\chi}_\times=0$ whenever $\chi\neq1$.
When $\gamma\neq1$ and $\chi=1$, we instead have
\smash{$\ip{\gamma}{1}_\times
=\sum_{x\in\bbF_Q^\times}\overline{\gamma(x)}
=-1+\sum_{x\in\bbF_Q^\times}\overline{\gamma(x)}
=-1$}.
That is,
\begin{equation}
\label{eq.Gauss sum easy cases}
\ip{\gamma}{\chi}_\times
=\left\{\begin{array}{rl}
Q-1,&\ \gamma=1,\,\chi=1,\\
  0,&\ \gamma=1,\,\chi\neq 1,\\
 -1,&\ \gamma\neq1,\,\chi=1.
\end{array}\right.
\end{equation}
The remaining case has a well-known modulus:
\begin{equation}
\label{eq.Gauss sum magnitude}
\abs{\ip{\gamma}{\chi}_\times}
=\sqrt{Q},
\quad\forall\, \gamma\neq1,\,\chi\neq 1.
\end{equation}
We will also need a certain relationship between a Gauss sum of $\chi$ and a Gauss sum of its inverse (entrywise conjugate) $\overline{\chi}$:
\begin{equation}
\label{eq.Gauss sum of conjugate}
\ip{\gamma}{\overline{\chi}}_\times=\chi(-1)\overline{\ip{\gamma}{\chi}_\times},
\quad\forall\, \gamma\in\widehat{\bbF}_Q^{}, \, \chi \in \widehat{\bbF}^\times_Q.
\end{equation}

In general,
since the multiplicative group of $\bbF_Q$ is cyclic,
we can enumerate its characters by letting $x_0$ be any fixed generator of $\bbF_Q$ and defining \smash{$\chi_m(x_0^n):=\exp(\frac{2\pi\rmi mn}{Q-1})$} for all $m,n\in\bbZ_{Q-1}$.
In particular,
$\bbF_Q$ has no nontrivial real-valued multiplicative character when $Q$ is even,
and has a sole such character when $Q$ is odd.
Extending this multiplicative character to $\bbF_Q$ yields the \textit{Legendre symbol},
namely the function
\begin{equation}
\label{eq.Legendre symbol}
(\tfrac{\cdot}{Q}):\bbF_Q\rightarrow\bbC,
\quad
(\tfrac{x}{Q}):=\left\{\begin{array}{rl}
 0,&\ x=0,\\
 1,&\ \exists\,y\in\bbF_Q^\times\text{ such that }x=y^2,\\
-1,&\ \text{else}.
\end{array}\right.
\end{equation}
We abuse terminology, sometimes also using ``Legendre symbol" to refer to the multiplicative character obtained by restricting~\eqref{eq.Legendre symbol} to $\bbF_Q^\times$.

Any multiplicative character satisfies $\chi(-x)=\chi(-1)\chi(x)$ for all $x\in\bbF_Q$ where $\chi(-1)\in\set{1,-1}$,
and so is \textit{even} (i.e., has even symmetry) when $\chi(-1)=1$ and is \textit{odd} when $\chi(-1)=-1$.
When $Q$ is even, every multiplicative character is even.
When $Q$ is odd, exactly half of the $Q-1$ multiplicative characters are even.
Furthermore, when $Q$ is odd, the Legendre symbol is even if and only if $Q\equiv 1\bmod 4$.
Finally, we will use the standard result that any nontrivial multiplicative character $\chi$ satisfies
\begin{equation}
\label{eq.conjugate reversal relation}
\chi(x)=\overline{\chi(-x)}\quad \forall\, x\in\bbF_Q^\times
\qquad
\Longleftrightarrow
\qquad
Q\equiv 1\bmod 4\quad \text{and}\quad \chi(x)=(\tfrac{x}{Q})\quad \forall\, x\in\bbF_Q^\times.
\end{equation}

\subsection{Constructions}

With these facts in hand, we are ready to construct an infinite family of harmonic EITFFs that generalizes the one given in Example~\ref{ex.EITFF(4,5,2)}:

\begin{theorem}
\label{thm.EITFF(Q-1,Q,2)}
A harmonic $\EITFF(Q-1,Q,2)$ exists for any odd prime power $Q$.
Also, a harmonic $\ECTFF(Q-1,Q,R)$ exists for any prime power $Q$ and positive integer $R\leq Q-1$, and it can be chosen to be real when both $Q\equiv 1\bmod 4$ and $R=2$.
\end{theorem}

(To clarify, Theorem~\ref{thm.EITFF(Q-1,Q,2)} does not announce the existence of a \textit{real} $\EITFF(Q-1,Q,2)$ for any prime power $Q\equiv 1\bmod 4$.)

For context, a complex $\ETF(\frac12(N-1),N)$ is known to exist whenever $N=Q$ is a prime power with $Q\equiv3\bmod4$~\cite{XiaZG05,Renes07} or, more generally, whenever there exists a skew-symmetric (real) conference matrix of size $N+1$~\cite{Strohmer08},
and directly summing two copies of such an ETF yields an $\EITFF(N-1,N,2)$.
However, whenever $Q$ is a prime power with $Q\equiv1\bmod4$, no complex $\ETF(\frac12(Q-1),Q)$ is known to exist~\cite{FickusM16}.
For $Q$ of this latter type, one may alternatively attempt to obtain a complex $\EITFF(Q-1,Q,2)$ by applying Hoggar's $\bbH$-to-$\bbC$ method to a quaternionic $\ETF(\frac12(Q-1),Q)$, and such ETFs are known to exist for $Q\in\set{5,7,9,11,13}$~\cite{CohnKM16,EtTaoui20,Waldron20}, but not more generally.
The EITFFs produced by Theorem~\ref{thm.EITFF(Q-1,Q,2)} above are thus apparently new for all of the infinite number of prime powers $Q\geq 19$ with $Q\equiv1\bmod 4$.
This itself suggests a question that we leave for future work: does a quaternionic $\ETF(\frac12(N-1),N)$ exist for every odd $N>1$?
In Theorem~\ref{thm.conference},
we generalize the construction of Theorem~\ref{thm.EITFF(Q-1,Q,2)},
finding that a complex $\EITFF(N-1,N,2)$ exists whenever a skew-symmetric complex conference matrix of size $N+1$ does.

\begin{proof}[Proof of Theorem~\ref{thm.EITFF(Q-1,Q,2)}]
We apply Theorem~\ref{thm.EITFF}, taking $\calG$ to be the additive group of $\bbF_Q$.
Let $R$ be any integer such that $1\leq R\leq Q-1$,
and let $\set{\chi_r}_{r=1}^R$ be any sequence of distinct multiplicative characters of \smash{$\bbF_Q$} with $\chi_1$ being the identity character.
To construct a harmonic $\ECTFF(Q-1,Q,R)$,
let $\bfPsi_0$ be the (degenerate, rank-$0$) $R\times\emptyset$ isometry,
and for any nonzero $x\in\bbF_Q$ let $\bfPsi_x$ be the $R\times 1$ matrix (column vector) with \smash{$\bfPsi_x(r)=\frac1{\sqrt{R}}\chi_r(x)$} for all $x\neq0$ and $r\in[R]$.
Defining $\bfP_x:=\bfPsi_x^{}\bfPsi_x^*$ for all $x\in\bbF_Q$,
we thus have that $\bfP_0:=\bfPsi_0^{}\bfPsi_0^*=\bfzero$,
and that for any $x\neq0$, $\bfP_x=\bfPsi_x^{}\bfPsi_x^*$ is the $R\times R$ rank-$1$ projection whose $(r_1,r_2)$th entry is
\begin{equation}
\label{eq.pf of EITFF(Q-1,Q,2) 1}
\bfP_x(r_1,r_2)
=(\bfPsi_x^{}\bfPsi_x^*)(r_1,r_2)
=\tfrac1R\chi_{r_1}(x)\overline{\chi_{r_2}}(x)
=\tfrac1R(\chi_{r_1}^{}\chi_{r_2}^{-1})(x).
\end{equation}
For any additive character $\gamma$ of $\bbF_Q$, the $(r_1,r_2)$th entry of the matrix $\bfM_\gamma$ of~\eqref{eq.DFT of projections} is thus
\begin{equation}
\label{eq.pf of EITFF(Q-1,Q,2) 2}
\bfM_\gamma(r_1,r_2)
=\sum_{x\in\bbF_Q}\overline{\gamma(x)}\bfP_x(r_1,r_2)
=0+\sum_{x\in\bbF_Q^\times}\overline{\gamma(x)}\tfrac1R(\chi_{r_1}^{}\chi_{r_2}^{-1})(x)
=\tfrac1R\ip{\gamma}{\chi_{r_1}^{}\chi_{r_2}^{-1}}_\times.
\end{equation}
When $\gamma=1$, this reduces via~\eqref{eq.Gauss sum easy cases} to \smash{$\bfM_1=\sum_{x\in\bbF_Q}\bfP_x=\tfrac{Q-1}{R}\bfI$},
implying that $\set{\bfPsi_x}_{x\in\bbF_Q}$ are the isometries of a generating $\TFF(R,Q,\set{D_x}_{x\in\bbF_Q})$ where $D_0=0$ and $D_x=1$ for all $x\neq0$,
and so generate a harmonic $\TFF(Q-1,Q,R)$ by Theorem~\ref{thm.small TFF yields harmonic}.
To show that this harmonic $\TFF(Q-1,Q,R)$ is an ECTFF,
note that for any $\gamma\neq1$,
\eqref{eq.Gauss sum easy cases} and~\eqref{eq.Gauss sum magnitude} give
\begin{equation*}
\abs{\bfM_\gamma(r_1,r_2)}^2
=\tfrac1{R^2}\abs{\ip{\gamma}{\chi_{r_1}^{}\chi_{r_2}^{-1}}_\times}^2
=\tfrac1{R^2}\left\{\begin{array}{cl}
1,&\ r_1=r_2,\\
Q,&\ r_1\neq r_2,
\end{array}\right.
\end{equation*}
implying that
$\norm{\bfM_\gamma}_\Fro^2
=\tfrac1{R^2}[R+R(R-1)Q]$
is independent of $\gamma\neq1$,
meaning the conditions of Theorem~\ref{thm.EITFF}(b) are met.

In light of~\eqref{eq.pf of EITFF(Q-1,Q,2) 1} and the fact that $\bfP_0=\bfzero$,
this harmonic $\ECTFF(Q-1,Q,R)$ meets the sufficient condition for being real that is given in Theorem~\ref{thm.EITFF}(c) if and only if
\begin{equation*}
\tfrac1R(\chi_{r_1}^{}\chi_{r_2}^{-1})(x)
=\bfP_x(r_1,r_2)
=\overline{\bfP_{-x}(r_1,r_2)}
=\overline{\tfrac1R(\chi_{r_1}^{}\chi_{r_2}^{-1})(-x)}
\end{equation*}
for all $x\neq0$ and all distinct $r_1,r_2\in[R]$.
In light of~\eqref{eq.conjugate reversal relation} and the fact that $\chi_1$ is the identity multiplicative character,
this holds if and only if either $R=1$ or $R=2$, $\chi_2$ is the Legendre symbol, and $Q\equiv1\bmod 4$.

We conclude by characterizing when this harmonic $\ECTFF(Q-1,Q,R)$ is an EITFF in the special case where $R=2$.
Here, letting $\chi_2=\chi$ be any nonidentity multiplicative character,
and letting $\gamma$ be any nonidentity additive character,
\eqref{eq.pf of EITFF(Q-1,Q,2) 2}, \eqref{eq.Gauss sum easy cases} and~\eqref{eq.Gauss sum of conjugate} give
\begin{equation}
\label{eq.pf of EITFF(Q-1,Q,2) 3}
\bfM_\gamma
=\tfrac12\left[\begin{array}{cc}
\ip{\gamma}{1}_{\times}&\ip{\gamma}{\chi^{-1}}_{\times}\\
\ip{\gamma}{\chi}_{\times}&\ip{\gamma}{1}_{\times}
\end{array}\right]
=\tfrac12\left[\begin{array}{cc}
-1&\ip{\gamma}{\overline{\chi}}_{\times}\\
\ip{\gamma}{\chi}_{\times}&-1
\end{array}\right]
=\tfrac12\left[\begin{array}{cc}
-1&\chi(-1)\overline{\ip{\gamma}{\chi}_{\times}}\\
\ip{\gamma}{\chi}_{\times}&-1
\end{array}\right].
\end{equation}
For any $\gamma\neq1$,
\eqref{eq.Gauss sum magnitude} gives $\abs{\ip{\gamma}{\chi}_\times}^2=Q$ and so this in turn implies that
\begin{equation*}
\bfM_\gamma^*\bfM_\gamma^{}
=\tfrac14\left[\begin{array}{cc}
Q+1&-[1+\chi(-1)]\overline{\ip{\gamma}{\chi}_\times}\\
-[1+\chi(-1)]\ip{\gamma}{\chi}_\times&Q+1
\end{array}\right].
\end{equation*}
By Theorem~\ref{thm.EITFF}(a), such a harmonic $\ECTFF(Q-1,Q,2)$ is an EITFF if and only if $\chi(-1)=-1$,
namely if and only if $\chi$ has odd symmetry.
Such $\chi$ exist whenever $Q$ is odd.
\end{proof}

In the next result, we modify the construction of Theorem~\ref{thm.EITFF(Q-1,Q,2)} so as to obtain new (harmonic) $\EITFF(Q,Q,2)$.

\begin{theorem}
\label{thm.EITFF(Q,Q,2)}
A harmonic $\EITFF(Q,Q,2)$ exists for any prime power $Q\geq 4$,
and it can be chosen to be real whenever $Q\equiv 1\bmod 4$.
Also, a harmonic $\ECTFF(Q,Q,R)$ exists for any prime power $Q$ and positive integer  $R\leq Q-1$.
\end{theorem}

\begin{proof}
Let $R$ be any integer such that $1\leq R\leq Q-1$,
and let $\set{\chi_r}_{r=1}^R$ be any sequence of distinct multiplicative characters of \smash{$\bbF_Q$} with $\chi_1$ being the identity character.
To construct a harmonic $\ECTFF(Q,Q,R)$,
let \smash{$\set{\bfPsi_x}_{x\in\bbF_Q}$} be the $R\times 1$ matrices (column vectors) obtained by letting $\bfPsi_0$ be the first standard basis vector and letting
\begin{equation*}
\bfPsi_x(r)=\left\{\begin{array}{ll}
\bigbracket{\tfrac{Q-R}{R(Q-1)}}^{\frac12},&\ r=1,\smallskip\\
\bigbracket{\tfrac{Q}{R(Q-1)}}^{\frac12}\chi_r(x),&\ r=2,\dotsc,R,
\end{array}\right.
\end{equation*}
for any $x\neq0$ and $r\in[R]$.
These are rank-$1$ isometries (unit vectors):
for any $x\neq0$,
\begin{equation*}
\sum_{r=1}^R\abs{\bfPsi_x(r)}^2
=\tfrac{Q-R}{R(Q-1)}+(R-1)\tfrac{Q}{R(Q-1)}
=1.
\end{equation*}
The corresponding rank-$1$ projections $\set{\bfP_x}_{x\in\bbF_Q}$,
$\bfP_x=\bfPsi_x^{}\bfPsi_x^*$ have $(r_1,r_2)$th entries
\begin{equation}
\label{eq.pf of EITFF(Q,Q,2) 1}
\bfP_0(r_1,r_2)
=\left\{\begin{array}{cl}
1,&\ r_1=1=r_2,\\
0,&\ \text{else},
\end{array}\right.
\quad
\bfP_x(r_1,r_2)
=\left\{\begin{array}{ll}
\tfrac{Q-R}{R(Q-1)},&\ x\neq0,\, r_1=1=r_2,\smallskip\\
\tfrac{[Q(Q-R)]^{\frac12}}{R(Q-1)}\chi_{r_2}^{-1}(x),&\ x\neq0,\, r_1=1<r_2,\smallskip\\
\tfrac{[Q(Q-R)]^{\frac12}}{R(Q-1)}\chi_{r_1}^{}(x),&\ x\neq0,\, r_1>1=r_2,\smallskip\\
\tfrac{Q}{R(Q-1)}(\chi_{r_1}^{}\chi_{r_2}^{-1})(x),&\ x\neq0,\, r_1,r_2>1.
\end{array}\right.
\end{equation}
For any additive character $\gamma$,
the $(r_1,r_2)$th entry of the matrix $\bfM_\gamma$ of~\eqref{eq.DFT of projections} is thus
\begin{equation}
\label{eq.pf of EITFF(Q,Q,2) 2}
\bfM_\gamma(r_1,r_2)
=\left\{\begin{array}{ll}
1+\tfrac{Q-R}{R(Q-1)}\ip{\gamma}{1}_\times,&\ r_1=1=r_2,\smallskip\\
\tfrac{[Q(Q-R)]^{\frac12}}{R(Q-1)}\ip{\gamma}{\chi_{r_2}^{-1}}_\times,&\ r_1=1<r_2,\smallskip\\
\tfrac{[Q(Q-R)]^{\frac12}}{R(Q-1)}\ip{\gamma}{\chi_{r_1}}_\times,&\ r_1>1=r_2,\smallskip\\
\tfrac{Q}{R(Q-1)}\ip{\gamma}{\chi_{r_1}^{}\chi_{r_2}^{-1}}_\times,&\ r_1,r_2>1.
\end{array}\right.
\end{equation}
Since the multiplicative characters $\set{\chi_r}_{r=1}^R$ are distinct with $\chi_1=1$,
this reduces via~\eqref{eq.Gauss sum easy cases} to the following in the special case where $\gamma=1$:
\begin{equation*}
\bfM_1(r_1,r_2)
=\left\{\begin{array}{cl}
1+\tfrac{Q-R}{R(Q-1)}(Q-1),&\ r_1=r_2=1\smallskip\\
\tfrac{Q}{R(Q-1)}(Q-1),&\ r_1=r_2>1\smallskip\\
0,&\ \text{else}
\end{array}\right\}
=\tfrac{Q}{R}\bfI(r_1,r_2).
\end{equation*}
Thus, \smash{$\bfM_1=\sum_{x\in\bbF_Q}\bfP_x=\tfrac{Q}{R}\bfI$},
implying that $\set{\bfPsi_x}_{x\in\bbF_Q}$ are the isometries of a generating $\TFF(R,Q,1)$,
and so yield a harmonic $\TFF(Q,Q,R)$ by Theorem~\ref{thm.small TFF yields harmonic}.
To show that this harmonic $\TFF(Q,Q,R)$ is an ECTFF,
note that for any $\gamma\neq1$,
combining~\eqref{eq.Gauss sum easy cases} with~\eqref{eq.pf of EITFF(Q,Q,2) 2} gives the following for any $\gamma\neq1$:
\begin{equation}
\label{eq.pf of EITFF(Q,Q,2) 3}
\bfM_\gamma(r_1,r_2)
=\left\{\begin{array}{ll}
\tfrac{(R-1)Q}{R(Q-1)},&\ \gamma\neq 1,\, r_1=r_2=1,\smallskip\\
-\tfrac{Q}{R(Q-1)},&\ \gamma\neq 1,\, r_1=r_2>1,\smallskip\\
\tfrac{[Q(Q-R)]^{\frac12}}{R(Q-1)}\ip{\gamma}{\chi_{r_2}^{-1}}_\times,&\ \gamma\neq 1,\, r_1=1<r_2,\smallskip\\
\tfrac{[Q(Q-R)]^{\frac12}}{R(Q-1)}\ip{\gamma}{\chi_{r_1}}_\times,&\ \gamma\neq 1,\, r_1>1=r_2,\smallskip\\
\tfrac{Q}{R(Q-1)}\ip{\gamma}{\chi_{r_1}^{}\chi_{r_2}^{-1}}_\times,&\ \gamma\neq 1,\, r_1\neq r_2,\, r_1,r_2>1.
\end{array}\right.
\end{equation}
Moreover, in the last three cases above, \eqref{eq.Gauss sum magnitude} gives that $\ip{\gamma}{\chi_{r_2}^{-1}}_\times$, $\ip{\gamma}{\chi_{r_1}}_\times$ and $\ip{\gamma}{\chi_{r_1}^{}\chi_{r_2}^{-1}}_\times$ have magnitude $\sqrt{Q}$, respectively.
As such, for any $\gamma\neq1$,
\begin{equation*}
\norm{\bfM_\gamma}_\Fro^2
=\bigbracket{\tfrac{(R-1)Q}{R(Q-1)}}^2
+(R-1)\bigbracket{\tfrac{Q}{R(Q-1)}}^2
+2(R-1)\tfrac{Q(Q-R)}{[R(Q-1)]^2}Q
+(R-1)(R-2)\bigbracket{\tfrac{Q}{R(Q-1)}}^2Q.
\end{equation*}
Thus, our harmonic $\TFF(Q,Q,R)$ satisfies the condition of Theorem~\ref{thm.EITFF}(b), and so is indeed an ECTFF.
(As a check on this work, the interested reader can verify that the quantity above simplifies to
\smash{$
\tfrac{(R-1)Q^2}{R(Q-1)}$},
which is consistent with the necessary value of
$B$ given in Theorem~\ref{thm.EITFF}(b).)

Next,
in light of~\eqref{eq.pf of EITFF(Q,Q,2) 1},
the fact that $\chi_1=1$, and the fact that $\bfP_0$ is real,
this harmonic $\ECTFF(Q,Q,R)$ satisfies the sufficient condition for realness given in Theorem~\ref{thm.EITFF}(c) if and only if
$(\chi_{r_1}^{}\chi_{r_2}^{-1})(x)
=\overline{(\chi_{r_1}^{}\chi_{r_2}^{-1})(-x)}$
for all $x\neq0$ and all distinct $r_1,r_2\in[R]$.
As in the proof of Theorem~\ref{thm.EITFF(Q-1,Q,2)},
\eqref{eq.conjugate reversal relation} implies that this occurs if and only if either $R=1$ or $R=2$, $\chi_2$ is the Legendre symbol, and $Q\equiv1\bmod 4$.
We do not explicitly mention this fact in the statement of the result since,
as we now explain, it is a subset of the cases in which this harmonic ECTFF is an EITFF.

When $R=2$,
for any nontrivial multiplicative character $\chi=\chi_2$,
\eqref{eq.Gauss sum of conjugate} and~\eqref{eq.pf of EITFF(Q,Q,2) 3} give
\begin{equation}
\label{eq.pf of EITFF(Q,Q,2) 4}
\bfM_\gamma
=\left[\begin{array}{cc}
\frac{Q}{2(Q-1)}&
\frac{[Q(Q-2)]^{\frac12}}{2(Q-1)}\chi(-1)\overline{\ip{\gamma}{\chi}_\times}\\
\frac{{[Q(Q-2)]^{\frac12}}}{2(Q-1)}\ip{\gamma}{\chi}_\times&
-\frac{Q}{2(Q-1)}
\end{array}\right]
\end{equation}
for any $\gamma\neq 1$.
In light of~\eqref{eq.Gauss sum magnitude},
we thus have
\begin{equation*}
\bfM_\gamma^*\bfM_\gamma^{}
=\left[\begin{array}{cc}
\frac{Q^2}{4(Q-1)}&
\frac{Q^{\frac32}(Q-2)^{\frac12}}{4(Q-1)^2}[\chi(-1)-1]\overline{\ip{\gamma}{\chi}_\times}\smallskip\\
\frac{Q^{\frac32}(Q-2)^{\frac12}}{4(Q-1)^2}[\chi(-1)-1]\ip{\gamma}{\chi}_\times&
\frac{Q^2}{4(Q-1)}
\end{array}\right],
\end{equation*}
for any $\gamma\neq 1$.
By Theorem~\ref{thm.EITFF}(a), our harmonic $\ECTFF(Q,Q,2)$ is thus an EITFF if and only if $\chi(-1)=1$, that is, if and only if $\chi$ has even symmetry.
A nontrivial multiplicative character $\chi$ for $\bbF_Q$ with even symmetry exists if and only if $Q\geq 4$.
In fact, when $Q\equiv 1\bmod 4$,
we can choose $\chi$ to be the Legendre symbol,
and as we have already explained, doing so yields a real TFF, and so a real harmonic $\EITFF(Q,Q,2)$.
\end{proof}

When $N$ is even, the existence of an $\EITFF(N,N,2)$ is not too surprising since $\ETF(\frac{N}{2},N)$ are relatively common~\cite{FickusM16}.
But, when $N$ is odd, no such EITFF can arise by either direct sums or Hoggar's methods.
That said, applying a clever modification of Hoggar's $\bbC$-to-$\bbR$ method to the appropriate matrix yields an $\EITFF(N,N,2)$ whenever $N=Q$ is a prime power with $Q\equiv1\bmod4$~\cite{EtTaoui18} or, more generally, whenever there exists a complex symmetric conference matrix of size $N$~\cite{BlokhuisBE18}.
Nevertheless, the EITFFs of Theorem~\ref{thm.EITFF(Q,Q,2)} seem to be new for all (infinitely many) prime powers $Q\geq 7$ with $Q\equiv3\bmod4$.
For example, when $Q=7$,
we have $\bbF_7^\times=\set{1,2,3,4,5,6}=\set{3^0,3^2,3^1,3^4,3^5,3^3}$,
and so a new (complex, harmonic) $\EITFF(7,7,2)$ is generated by the following isometries $\set{\bfPsi_x}_{x\in\bbF_7}$ of a $\TFF(2,7,1)$:
\begin{equation*}
\bfPsi
=\left[\begin{array}{c|c|c|c|c|c|c}
\bfPsi_0&\bfPsi_1&\bfPsi_2&\bfPsi_3&\bfPsi_4&\bfPsi_5&\bfPsi_6
\end{array}\right]
=\left[\begin{array}{c|c|c|c|c|c|c}
1&
a&
a&
a&
a&
a&
a\\
\ 0\ &
\ b\ &
\ b\omega^2\ &
\ b\omega\ &
\ b\omega\ &
\ b\omega^2\ &
\ b\
\end{array}\right]
\end{equation*}
where $a=(\tfrac{5}{12})^{\frac12}$, $b=(\tfrac{7}{12})^{\frac12}$ and $\omega=\exp(\frac{2\pi\rmi}{3})$.
In Theorem~\ref{thm.combinatorial},
we generalize Theorem~\ref{thm.EITFF(Q,Q,2)},
finding that an $\EITFF(N,N,2)$ exists whenever a complex symmetric conference matrix of size $N+1$ exists.
As detailed there, this is notably distinct from the method of~\cite{EtTaoui18,BlokhuisBE18} that produces such an EITFF from such a matrix of size $N$,
but nevertheless equates to it in an important special case.

The two previous results show that the generalization of difference sets to difference projections given in Definition~\ref{def.difference projections} bears fruit.
We conclude this section with a construction of a real harmonic $\EITFF(11,11,3)$ from a generating $\TFF(3,11,1)$, specifically from $11$ rank-$1$ difference projections for the additive group of $\bbF_{11}$.
Up to Naimark complements, this is only the second (real or complex) $\EITFF(D,N,R)$ to have ever been discovered in which $D$ is relatively prime to $R>2$, the first being the real $\EITFF(8,6,3)$ of~\cite{EtTaouiR09,IversonKM21}.

\begin{theorem}
\label{thm.EITFF(11,11,3)}
There exists a real harmonic $\EITFF(11,11,3)$.
Explicitly, let $\chi$ be any generator of $\widehat{\bbF}_{11}^\times$, let $\gamma \in \widehat{\bbF}_{11}$ be given by $\gamma(x) := \exp( \tfrac{ 2 \pi \mathrm{i} x}{11} )$, put \smash{$\omega := \frac{ \langle \gamma, \chi^6 \rangle_{\times} }{ \langle \gamma, \chi \rangle_{\times} }$}, $c := \cos( \tfrac{2\pi}{5} )$, and $s := \sin( \tfrac{2\pi}{5} )$, and define
\begin{equation*}
\bfU:=\frac1{\sqrt{2}}
\left[\begin{array}{ccr}
\sqrt{2}&0&0\\
0&1&1\\
0&\rmi&-\rmi
\end{array}\right],
\quad
\bfPsi_0:=\left[\begin{array}{c}1\\0\\0\end{array}\right],
\quad
\bfPsi_x:=\tfrac{\sqrt{11}}{\sqrt{30}}\left[\begin{array}{c}
\frac{\sqrt{8}}{\sqrt{11}}\smallskip\\
\chi(x)[c(\tfrac{x}{11})+\rmi s\omega]\smallskip\\
\overline{\chi(x)}[c(\tfrac{x}{11})+\rmi s\overline{\omega}]
\end{array}\right],\ \forall\, x\in\bbF_{11}^\times,
\end{equation*}
where $(\tfrac{x}{11})$ denotes the Legendre symbol of $x$.
Then the harmonic matrix ensemble generated by $\set{\bfU \bfPsi_x}_{x\in\bbF_{11}}$ forms a real harmonic $\EITFF(11,11,3)$.
\end{theorem}

It is easy to verify Theorem~\ref{thm.EITFF(11,11,3)} with the aid of a computer algebra system, and the arXiv version of this paper includes an ancillary file with code that implements such a verification in GAP~\cite{GAP}.
For the sake of completeness, we also include a human-readable (albeit technical) proof of Theorem~\ref{thm.EITFF(11,11,3)} as Appendix~A.

\section{Combinatorial generalizations of some harmonic EITFF constructions}

In Theorem~\ref{thm.EITFF(Q-1,Q,2)},
we used Theorem~\ref{thm.EITFF} to construct a harmonic $\EITFF(Q-1,Q,2)$ for any odd prime power.
Examining~\eqref{eq.pf of EITFF(Q-1,Q,2) 3} from the perspective of Theorems~\ref{thm.small TFF yields harmonic} and~\ref{thm.EITFF} reveals that for any distinct $n_1,n_2\in[Q]$, the $(n_1,n_2)$th block of this EITFF's signature matrix $\bfS$ is of the form
\begin{equation}
\label{eq.sig of EITFF(N-1,N,2)}
\bfS_{n_1,n_2}
=\tfrac1{\sqrt{N+1}}\left[\begin{array}{cc}
-1&-\sqrt{N}\,\overline{\bfZ(n_1,n_2)}\\
\sqrt{N}\,\bfZ(n_1,n_2)&-1
\end{array}\right],
\end{equation}
where $\bfZ$ is some $\bbF_Q$-circulant complex matrix whose off-diagonal entries are unimodular and whose diagonal entries are $0$.
Applying a similar logic to~\eqref{eq.pf of EITFF(Q,Q,2) 4} reveals that the $(n_1,n_2)$th block of the signature matrix of the $\EITFF(Q,Q,2)$ we constructed in Theorem~\ref{thm.EITFF(Q,Q,2)} is of the form
\begin{equation}
\label{eq.sig of EITFF(N,N,2)}
\bfS_{n_1,n_2}=\tfrac1{\sqrt{N-1}}\left[\begin{array}{cc}
1&\sqrt{N-2}\,\overline{\bfZ(n_1,n_2)}\\
\sqrt{N-2}\,\bfZ(n_1,n_2)&-1\end{array}\right],
\end{equation}
for some such matrix $\bfZ$.
In this section, we relax the requirement that $\bfZ$ be $\bbF_Q$-circulant,
and instead characterize those $N\times N$ matrices $\bfZ$ for
which~\eqref{eq.sig of EITFF(N-1,N,2)} or~\eqref{eq.sig of EITFF(N,N,2)} define a signature matrix of an $\EITFF(N-1,N,2)$ or of an $\EITFF(N,N,2)$, respectively.
As we shall see, it turns out that~\eqref{eq.sig of EITFF(N,N,2)} is distinct from, but also related to, another recently discovered construction of $\EITFF(N,N,2)$~\cite{EtTaoui18,BlokhuisBE18},
and moreover that all three constructions involve \textit{conference matrices}.

To elaborate,
for any positive integer $M$,
a \textit{complex conference matrix} of size $M$
is any $M\times M$ complex matrix $\bfC$ whose diagonal entries are zero,
whose off-diagonal entries are unimodular,
and which is a scalar multiple of a unitary matrix, satisfying $\bfC^*\bfC=(M-1)\bfI$.
A \textit{(real) conference matrix} is a complex conference matrix $\bfC$ with
$\bfC(m_1,m_2)\in\set{1,-1}$ for all $m_1\neq m_2$.
We say that a complex matrix $\bfC$ is \textit{$\varepsilon$-symmetric} for some $\varepsilon\in\set{1,-1}$ if $\bfC^\rmT=\varepsilon\bfC$.
That is, $\bfC$ is $1$-symmetric or $-1$-symmetric if it is symmetric or skew-symmetric, respectively.
A complex conference matrix $\bfC$ is \textit{$\varepsilon$-normalized} if $\bfC(m,1)=1=\varepsilon\bfC(1,m)$ for all $m>1$, namely if it is of the following form:
\begin{equation}
\label{eq.def of core}
\bfC
=\left[\begin{array}{cc}
0&\varepsilon\bfone^*\\
\bfone&\bfC_\varepsilon
\end{array}\right],
\end{equation}
where $\bfone$ is an $(M-1)\times 1$ all-ones vector.
When $\bfC$ is $\varepsilon$-normalized, we say $\bfC_\varepsilon$ is its $\varepsilon$-\textit{core} which is real and/or $\varepsilon$-symmetric if and only if $\bfC$ is as well.
(We also use the terms \textit{core} and \textit{skew-core} to refer to $\bfC_\varepsilon$ when $\varepsilon$ is $1$ or $-1$, respectively.)
If $\bfC$ is any complex conference matrix, then $\bfD_1\bfC\bfD_{2,\varepsilon}$ is an $\varepsilon$-normalized complex conference matrix where $\bfD_1$ and $\bfD_{2,\varepsilon}$ are $M\times M$ diagonal matrices with $\bfD_1(1,1)=\bfD_{2,\varepsilon}(1,1)=1$ and $\bfD_1(m,m)=\overline{\bfC(m,1)}$, $\bfD_{2,\varepsilon}(m,m)=\varepsilon\overline{\bfC(1,m)}$ for all $m>1$.
If $\bfC$ is real and/or $\varepsilon$-symmetric then $\bfD_1\bfC\bfD_{2,\varepsilon}$ is as well; for the latter, note that
$\bfD_1=\bfD_{2,\varepsilon}$ and so
$(\bfD_1\bfC\bfD_{2,\varepsilon})^\rmT
=\bfD_{2,\varepsilon}^\rmT\bfC^\rmT\bfD_1^\rmT
=\bfD_{2,\varepsilon}\varepsilon\bfC\bfD_1
=\varepsilon\bfD_1\bfC\bfD_{2,\varepsilon}$.
For any $M\times M$ complex matrix $\bfC$ of the form~\eqref{eq.def of core} for some $\varepsilon\in\set{1,-1}$ we have
\begin{equation*}
\bfC^*\bfC
=\left[\begin{array}{cc}
0&\bfone^*\\
\varepsilon\bfone&\bfC_\varepsilon^*
\end{array}\right]\left[\begin{array}{cc}
0&\varepsilon\bfone^*\\
\bfone&\bfC_\varepsilon^{}
\end{array}\right]
=\left[\begin{array}{cc}
M-1&\bfone^*\bfC_\varepsilon\\
\bfC_\varepsilon^*\bfone&\bfone\bfone^*+\bfC_\varepsilon^*\bfC_\varepsilon^{}
\end{array}\right].
\end{equation*}
As such, an $(M-1)\times(M-1)$ complex matrix $\bfC_\varepsilon$ whose diagonal entries are zero and whose off-diagonal entries are unimodular is the $\varepsilon$-core of an $\varepsilon$-normalized complex conference matrix
if and only if both $\bfC_\varepsilon^*\bfone=\bfzero$
and $\bfC_\varepsilon^*\bfC_\varepsilon^{}=(M-1)\bfI-\bfJ$ where $\bfJ=\bfone\bfone^*$ is an $(M-1)\times(M-1)$ all-ones matrix.
Here, having $\bfC_\varepsilon^*\bfC_\varepsilon^{}=(M-1)\bfI-\bfJ$
implies
$\bfone\in\ker((M-1)\bfI-\bfJ)
=\ker(\bfC_\varepsilon^*\bfC_\varepsilon^{})
=\ker(\bfC_\varepsilon^{})$.
(In fact, these two conditions equate to the columns of $\bfC_\varepsilon$ summing to zero and forming a regular simplex in the orthogonal complement of the all-ones vector.)
As such, when $\bfC_\varepsilon$ is $\varepsilon$-symmetric and $\bfC_\varepsilon^*\bfC_\varepsilon^{}=(M-1)\bfI-\bfJ$,
we automatically have that
$\bfC_\varepsilon^*\bfone
=\overline{\bfC_\varepsilon^\rmT}\bfone
=\varepsilon\overline{\bfC_\varepsilon\bfone}
=\bfzero$.
Altogether, for $\varepsilon\in\set{1,-1}$, we see that an $(M-1)\times(M-1)$ complex matrix $\bfC_\varepsilon$ is the $\varepsilon$-core of an $\varepsilon$-normalized complex $\varepsilon$-symmetric conference matrix $\bfC$ of size $M$ if and only if
\begin{equation}
\label{eq.core of symmetric conference}
\bfC_{\varepsilon}^*\bfC_{\varepsilon}^{}=(M-1)\bfI-\bfJ,
\qquad
\bfC_{\varepsilon}^\rmT=\varepsilon\bfC_{\varepsilon}^{},
\qquad
\abs{\bfC_{\varepsilon}(m_1,m_2)}
=\left\{\begin{array}{cl}
1,&\ m_1=m_2,\\
0,&\ m_1\neq m_2,
\end{array}\right.
\end{equation}
and moreover, in this case we necessarily have
\begin{equation}
\label{eq.core of symmetric conference sums to zero}
\bfC_\varepsilon^{}\bfone=\bfzero=\bfC_\varepsilon^\rmT\bfone.
\end{equation}
In particular, for any positive integer $M$ and choice of $\varepsilon\in\set{1,-1}$,
a complex $\varepsilon$-symmetric conference matrix of size $M$ exists if and only if such a matrix $\bfC_\varepsilon$ exists.

Real conference matrices are a subject of classical interest,
see~\cite{BlokhuisBE18,FickusM16} for recent overviews.
In brief, it is known that if a real symmetric conference matrix of size $M$ exists then $M\equiv2\bmod4$ and $M-1$ is a sum of two squares, and that if a real skew-symmetric conference matrix of size $M$ exists then either $M=2$ or $4\mid M$.
Notably, if any conference matrix $M$ exists then either a symmetric or skew-symmetric conference matrix of size $M$ exists, depending on whether $M\equiv 2$ or $M\equiv 0\bmod 4$~\cite{DelsarteGS71}.
The existence of a real symmetric conference matrix of size $M$ equates to that of a real $\ETF(\tfrac{M}{2},M)$;
such conference matrices are known to exist, for example, when $M=Q+1$ where $Q\equiv 1\bmod 4$ is a prime power, and when $M=5(9^J)+1$ where $J$ is an odd positive integer.
Meanwhile, the existence of a real skew-symmetric conference matrix of size $M$ equates to the existence of a complex $\ETF(\tfrac{M}{2},M)$ whose Gram matrix's off-diagonal entries are purely imaginary, and when $M>2$ moreover implies the existence of a complex $\ETF(\tfrac{M}{2}-1,M-1)$~\cite{Strohmer08};
such conference matrices are known to exist, for example, when $M=Q+1$ where $Q\equiv 3\bmod 4$.

Less is understood about complex conference matrices in general.
That said, a recent series of papers~\cite{EtTaoui18,BlokhuisBE18} gives that any real symmetric conference matrix of size $M+1$ yields a complex symmetric conference matrix of size $M$, which in turn yields an $\EITFF(M,M,2)$.
In particular, a real $\EITFF(Q,Q,2)$ exists for any prime power $Q\equiv1\bmod4$~\cite{EtTaoui18}.
\textit{Dihedral} complex symmetric conference matrices are also known to exist for any
$M\in\set{2,5,6,8,9,10,14,17,18}$,
but of these, only the $M=8$ case is particularly notable since (real) symmetric conference matrices of these sizes are known to exist whenever $M\in\set{2,6,10,14,18}$,
and complex symmetric conference matrices of these sizes arises from (real) symmetric conference matrices of size $M+1$ when $M\in\set{5,9,17}$.
We also note that~\cite{FickusS20} provides a circulant complex conference matrix of size $M$ whenever $M-1$ is a prime power, but these are not necessarily symmetric or skew-symmetric.

In the following result, we generalize Theorems~\ref{thm.EITFF(Q-1,Q,2)} and~\ref{thm.EITFF(Q,Q,2)} to construct an $\EITFF(N-1,N,2)$ from any complex skew-symmetric conference matrix of size $N+1$, and also to construct an $\EITFF(N,N,2)$ from any complex symmetric conference matrix of size $N+1$;
notably, applying the construction of~\cite{EtTaoui18,BlokhuisBE18} to the latter would produce an EITFF with distinct parameters, namely an $\EITFF(N+1,N+1,2)$.

\begin{theorem}
\label{thm.combinatorial}
Let $\bfZ$ be an $N\times N$ complex matrix whose diagonal entries are zero,
and whose off-diagonal entries are unimodular.
Let $\bfS$ be a $2N\times2N$ matrix with entries partitioned into $2\times 2$ blocks such that the diagonal blocks are zero matrices.
\begin{enumerate}
\renewcommand{\labelenumi}{(\alph{enumi})}
\item
If the off-diagonal blocks of $\bfS$ are given in terms of $\bfZ$ by \eqref{eq.sig of EITFF(N-1,N,2)},
then $\bfS$ is the signature matrix of an $\EITFF(N-1,N,2)$ if and only if $\bfZ$ is the skew-core of a skew-normalized complex skew-symmetric conference matrix.
\item
If the off-diagonal blocks of $\bfS$ are given in terms of $\bfZ$ by \eqref{eq.sig of EITFF(N,N,2)},
then $\bfS$ is the signature matrix of an $\EITFF(N,N,2)$ if and only if $\bfZ$ is the core of a normalized complex symmetric conference matrix.
\end{enumerate}
In both (a) and (b), the resulting EITFF is real if the requisite complex conference matrix is real.
As a consequence, no $5\times 5$ complex skew-symmetric conference matrix or $4\times 4$ complex symmetric conference matrix exists.
\end{theorem}

For context, no $3\times 3$ complex conference matrix exists (symmetric, skew-symmetric or otherwise) since no $3\times 3$ matrix with $0$ diagonal entries and unimodular off-diagonal entries has orthogonal columns.
Meanwhile, Proposition~3.1 of~\cite{EtTaouiM21} gives that any $4\times 4$ complex conference matrix is equivalent---via left and/or right multiplication by phased permutation matrices---to a (real) skew-symmetric conference matrix.
Furthermore, Theorem~3.4 of~\cite{EtTaouiM21} gives that every $5\times 5$ complex conference matrix is equivalent in this same sense to a certain complex symmetric conference matrix.
The final statement of the theorem above seems related to these results from~\cite{EtTaouiM21}, but is also not immediately implied by them.
In fact, this suggests a problem, which we leave for future research:
can a complex symmetric conference matrix of size $N>2$ be equivalent to a skew-symmetric one?

\begin{proof}[Proof of Theorem~\ref{thm.combinatorial}]
Our proof of (b) is slightly less technical than that of (a),
and so we discuss it first.
For any $n_1\neq n_2$,
the fact that $\bfZ(n_1,n_2)$ is unimodular implies that the matrix $\bfS_{n_1,n_2}$ defined by~\eqref{eq.sig of EITFF(N,N,2)} is unitary.
As such, \eqref{eq.signature matrix conditions} gives that $\bfS$ is the signature matrix of an $\EITFF(N,N,2)$ if and only if
\begin{equation*}
\bfS_{n_1,n_2}=\bfS_{n_2,n_1}^*,
\qquad
\sum_{n_3=1}^N\bfS_{n_1,n_3}\bfS_{n_3,n_2}=\bfzero,
\qquad\forall\,n_1\neq n_2.
\end{equation*}
Here, \eqref{eq.sig of EITFF(N,N,2)} moreover implies that $\bfS_{n_1,n_2}=\bfS_{n_2,n_1}^*$ for all $n_1\neq n_2$ if and only if $\bfZ(n_1,n_2)=\bfZ(n_2,n_1)$ for all $n_1\neq n_2$,
that is, if and only if $\bfZ$ is symmetric.
For any pairwise distinct $n_1,n_2,n_3\in[N]$,
\eqref{eq.sig of EITFF(N,N,2)} also gives that $\bfS_{n_1,n_3}\bfS_{n_3,n_2}$ is
\begin{multline*}
\tfrac1{N-1}\left[\begin{array}{cc}
1&\sqrt{N-2}\,\overline{\bfZ(n_1,n_3)}\\
\sqrt{N-2}\,\bfZ(n_1,n_3)&-1\end{array}\right]
\left[\begin{array}{cc}
1&\sqrt{N-2}\,\overline{\bfZ(n_3,n_2)}\\
\sqrt{N-2}\,\bfZ(n_3,n_2)&-1\end{array}\right]\\
=\tfrac1{N-1}\left[\begin{array}{cc}
1+(N-2)\overline{\bfZ(n_1,n_3)}\bfZ(n_3,n_2)&
\sqrt{N-2}\,\overline{[\bfZ(n_3,n_2)-\bfZ(n_1,n_3)]}\smallskip\\
\sqrt{N-2}\,[\bfZ(n_1,n_3)-\bfZ(n_3,n_2)]&
1+(N-2)\bfZ(n_1,n_3)\overline{\bfZ(n_3,n_2)}
\end{array}\right].
\end{multline*}
Summing this over all $n_3\in[N]$ that are distinct from $n_1$ and $n_2$,
and recalling that $\bfS_{n,n}=\bfzero$ and $\bfZ(n,n)=0$ for all $n\in[N]$,
we see that $\bfS$ is the signature matrix of an $\EITFF(N,N,2)$ if and only if $\bfZ$ is symmetric and the following two equations are satisfied for all distinct $n_1,n_2\in[N]$:
\begin{align}
\label{eq.pf of combinatorial 1}
0&=\sum_{\substack{n_3\in[N],\\n_3\notin\set{n_1,n_2}}}
[1+(N-2)\overline{\bfZ(n_1,n_3)}\bfZ(n_3,n_2)]
=(N-2)[1+(\bfZ^*\bfZ)(n_1,n_2)],\\
\label{eq.pf of combinatorial 2}
0&=\sum_{\substack{n_3\in[N],\\n_3\notin\set{n_1,n_2}}}
[\bfZ(n_1,n_3)-\bfZ(n_3,n_2)]
=\sum_{n_3=1}^N\bfZ(n_1,n_3)-\sum_{n_3=1}^N\bfZ(n_3,n_2).
\end{align}
Here, since every off-diagonal entry of $\bfZ$ is unimodular,
$(\bfZ^*\bfZ)(n,n)=N-1$ for all $n$,
and so having~\eqref{eq.pf of combinatorial 1} for all $n_1\neq n_2$ equates to having $\bfZ^*\bfZ=N\bfI-\bfJ$.
In particular, if $\bfS$ is the signature matrix of an $\EITFF(N,N,2)$ then $\bfC_\varepsilon=\bfZ$ satisfies~\eqref{eq.core of symmetric conference} when $\varepsilon=1$ and $M=N+1$, that is, when $\bfZ$ is the core of a normalized complex symmetric conference matrix.
Conversely, if $\bfC_\varepsilon=\bfZ$ satisfies~\eqref{eq.core of symmetric conference} when $\varepsilon=1$ and $M=N+1$,
then~\eqref{eq.core of symmetric conference sums to zero} also holds,
implying that $\bfZ$ is symmetric and satisfies both~\eqref{eq.pf of combinatorial 1} and~\eqref{eq.pf of combinatorial 2} for all $n_1,n_2$,
namely that $\bfS$ is the signature matrix of an $\EITFF(N,N,2)$.

For (a),
we instead use~\eqref{eq.sig of EITFF(N-1,N,2)} to define $\bfS_{n_1,n_2}$ for any $n_1\neq n_2$,
which is a unitary matrix since $\bfZ(n_1,n_2)$ is unimodular.
Since \smash{$(NR-2D)\bigbracket{\tfrac{N-1}{D(NR-D)}}^{\frac12}
=\frac2{\sqrt{N+1}}$}
when $(D,N,R)=(N-1,N,2)$,
we thus have by~\eqref{eq.signature matrix conditions} that $\bfS$ is the signature matrix of an $\EITFF(N-1,N,2)$ if and only if
\begin{equation*}
\bfS_{n_1,n_2}=\bfS_{n_2,n_1}^*,
\qquad
\sum_{n_3=1}^N\bfS_{n_1,n_3}\bfS_{n_3,n_2}
=\tfrac2{\sqrt{N+1}}\bfS_{n_1,n_2},
\qquad\forall\,n_1\neq n_2.
\end{equation*}
By~\eqref{eq.sig of EITFF(N-1,N,2)},
$\bfS_{n_1,n_2}=\bfS_{n_2,n_1}^*$ if and only if $\bfZ(n_1,n_2)=-\bfZ(n_2,n_1)$.
So, $\bfS_{n_1,n_2}=\bfS_{n_2,n_1}^*$ for all $n_1\neq n_2$ if and only if $\bfZ^\rmT=-\bfZ$.
Next, note that for any pairwise distinct $n_1,n_3,n_3\in[N]$,
\eqref{eq.sig of EITFF(N-1,N,2)} implies
\begin{align*}
\bfS_{n_1,n_3}\bfS_{n_3,n_2}
&=\tfrac1{N+1}\left[\begin{array}{cc}
-1&-\sqrt{N}\,\overline{\bfZ(n_1,n_3)}\\
\sqrt{N}\,\bfZ(n_1,n_3)&-1
\end{array}\right]\left[\begin{array}{cc}
-1&-\sqrt{N}\,\overline{\bfZ(n_3,n_2)}\\
\sqrt{N}\,\bfZ(n_3,n_2)&-1
\end{array}\right]\\
&=\tfrac1{N+1}\left[\begin{array}{cc}
1-N\overline{\bfZ(n_1,n_3)}\bfZ(n_3,n_2)&
\sqrt{N}\,\overline{[\bfZ(n_3,n_2)+\bfZ(n_1,n_3)]}\\
-\sqrt{N}\,[\bfZ(n_1,n_3)+\bfZ(n_3,n_2)]&
1-N\bfZ(n_1,n_3)\overline{\bfZ(n_3,n_2)}
\end{array}\right].
\end{align*}
Summing this over all $n_3\in[N]$ that are distinct from $n_1$ and $n_2$,
we see that $\bfS$ is thus the signature matrix of an $\EITFF(N-1,N,2)$ if and only if $\bfZ^\rmT=-\bfZ$ and, for all $n_1\neq n_2$, we have that
\begin{align*}
-\tfrac{2}{N+1}
=\sum_{\substack{n_3\in[N],\\n_3\notin\set{n_1,n_2}}}
\tfrac1{N+1}[1-N\overline{\bfZ(n_1,n_3)}\bfZ(n_3,n_2)]
&=\tfrac{N-2}{N+1}+\tfrac{N}{N+1}\sum_{n_3=1}^{N}\overline{\bfZ(n_3,n_1)}\bfZ(n_3,n_2)\\
&=\tfrac{N-2}{N+1}+\tfrac{N}{N+1}(\bfZ^*\bfZ)(n_1,n_2),
\end{align*}
i.e., $(\bfZ^*\bfZ)(n_1,n_2)=-1$,
and also that
\begin{align*}
\tfrac{2\sqrt{N}}{N+1}\bfZ(n_1,n_2)
&=\sum_{\substack{n_3\in[N],\\n_3\notin\set{n_1,n_2}}}
\bigparen{-\tfrac{\sqrt{N}}{N+1}}[\bfZ(n_1,n_3)+\bfZ(n_3,n_2)]\\
&=\bigparen{-\tfrac{\sqrt{N}}{N+1}}
\biggparen{-2\bfZ(n_1,n_2)+\sum_{n_3=1}^N\bfZ(n_1,n_3)+\sum_{n_3=1}^N\bfZ(n_3,n_2)},
\end{align*}
namely that $\sum_{n_3=1}^N\bfZ(n_1,n_3)+\sum_{n_3=1}^N\bfZ(n_3,n_2)=0$.
Paralleling the argument we used at this point of our proof of (b) then gives that $\bfS$ is the signature matrix of an $\EITFF(N-1,N,2)$ if and only if $\bfC_\varepsilon=\bfZ$ satisfies~\eqref{eq.core of symmetric conference} when $\varepsilon=-1$ and $M=N+1$, that is, when $\bfZ$ is the skew-core of a skew-normalized complex skew-symmetric conference matrix.

In either (a) or (b), note that by~\eqref{eq.sig of EITFF(N-1,N,2)} and~\eqref{eq.sig of EITFF(N,N,2)}, respectively, we have that $\bfS$ is real if $\bfZ$ is real,
and that in this case, the (skew-)normalized complex (skew-)symmetric conference matrix that $\bfZ$ is the (skew-)core of is also real.

For the final conclusions, note that if a $5\times 5$ complex skew-symmetric conference matrix exists then taking $\bfZ$ in (a) to be its $4\times 4$ core yields an $\EITFF(3,4,2)$, contradicting the fact that an $\EITFF(D,N,R)$ with $R<D$ can only exist if $R\leq\frac D2$.
Similarly, if a $4\times 4$ complex symmetric conference matrix exists then taking $\bfZ$ in (b) to be its $3\times 3$ core yields an $\EITFF(3,3,2)$, a contradiction.
\end{proof}

Applying Theorem~\ref{thm.combinatorial} to the EITFFs of Theorems~\ref{thm.EITFF(Q-1,Q,2)} and~\ref{thm.EITFF(Q,Q,2)} implies that a complex skew-symmetric conference matrix of size $Q+1$ exists for any odd prime power $Q$,
and that a complex symmetric conference matrix of size $Q+1$ exists for any prime power $Q\geq 4$, respectively.
These families are not mentioned in~\cite{EtTaoui18,BlokhuisBE18,EtTaouiM21} and so are apparently new.
As such, we also provide a shorter, direct proof of their existence:

\begin{theorem}
\label{thm.conference}
Let $Q$ be a prime power.
If $Q\geq 4$ then a complex symmetric conference matrix of size $Q+1$ exists.
If $Q$ is odd then a complex skew-symmetric conference matrix of size $Q+1$ exists.
\end{theorem}

\begin{proof}
We generalize the classical construction of a Paley conference matrix to the complex setting.
Let $\chi$ be any nontrivial multiplicative character of $\bbF_Q$,
and extend it to a function $\chi_0:\bbF_Q\rightarrow\bbC$ by letting $\chi_0(0):=0$.
Let $\bfC_0$ be the $\bbF_Q$-circulant matrix with $\chi_0$ as its $0$th column,
having $\bfC_0(x_1,x_2)=\chi_0(x_1-x_2)$ for all $x_1,x_2\in\bbF_Q$.
Note that $\bfC_0(x,x)=\chi_0(0)=0$ for all $x\in\bbF_Q$,
and moreover that $\abs{\bfC_0(x_1,x_2)}=\abs{\chi(x_1-x_2)}=1$ for any $x_1\neq x_2$.
Further note that $\bfC_0$ is either symmetric or skew-symmetric depending on whether $\chi$ has even or odd symmetry, that is, whether $\varepsilon:=\chi(-1)$ is either $1$ or $-1$;
for any prime power $Q\geq 4$, we can choose $\chi$ to be even,
and for any odd prime power $Q$ we can choose $\chi$ to be odd.
Being $\bbF_Q$-circulant, $\bfC_0$ is moreover diagonalized by the additive characters of $\bbF_Q$:
for any $\gamma\in\widehat{\bbF}_Q$, and $x_1\in\bbF_Q$,
letting $x_3=x_1-x_2$ below gives
\begin{equation*}
(\bfC_0\gamma)(x_1)
=\sum_{x_2\in\bbF_Q}\chi_0(x_1-x_2)\gamma(x_2)
=\sum_{x_3\in\bbF_Q^\times}\overline{\gamma(x_3)}\chi_0(x_3)\gamma(x_1)
=\bigbracket{\ip{\gamma}{\chi}_\times}\gamma(x_1),
\end{equation*}
for all $x\in\bbF_Q$.
Letting $\bfGamma$ be the corresponding $\bbF_Q\times\widehat{\bbF}_Q$ character table, defined by $\bfGamma(x,\gamma):=\gamma(x)$, we thus have that $\bfC_0=\frac1Q\bfGamma\bfD_0\bfGamma^*$ where $\bfD_0$ is the \smash{$\widehat{\bbF}_Q\times\widehat{\bbF}_Q$} diagonal matrix whose $\gamma$th diagonal entry is $\ip{\gamma}{\chi}_\times$.
As such,
\smash{$\bfC_0^*\bfC_0
=(\frac1Q\bfGamma\bfD_0\bfGamma^*)(\frac1Q\bfGamma\bfD_0^{}\bfGamma^*)
=\bfGamma\frac1Q\bfD_0^*\bfD_0^{}\bfGamma^*$}.
Here, by~\eqref{eq.Gauss sum easy cases} and~\eqref{eq.Gauss sum magnitude},
$\frac1Q\bfD_0^*\bfD_0^{}$ is a diagonal matrix whose $\gamma$th diagonal entry is
\begin{equation*}
\tfrac1Q(\bfD_0^*\bfD_0^{})(\gamma,\gamma)
=\tfrac1Q\abs{\ip{\gamma}{\chi}_\times}^2
=\left\{\begin{array}{cl}
0,&\ \gamma=1\\
1,&\ \gamma\neq 1
\end{array}\right\}
=(\bfI-\bfdelta_1^{}\bfdelta_1^*)(\gamma,\gamma).
\end{equation*}
Since $\bfGamma\bfdelta_1=\bfone$ is the all-ones vector,
we thus have that
\begin{equation*}
\bfC_0^*\bfC_0^{}
=\bfGamma\tfrac1Q\bfD_0^*\bfD_0^{}\bfGamma^*
=\bfGamma(\bfI-\bfdelta_1^{}\bfdelta_1^*)\bfGamma^*
=Q\bfI-(\bfGamma\bfdelta_1)(\bfGamma\bfdelta_1)^*
=Q\bfI-\bfone\bfone^*
=Q\bfI-\bfJ.
\end{equation*}
Altogether, we see that $\bfC_0$ satisfies~\eqref{eq.core of symmetric conference} for $\varepsilon=\chi(-1)$ and $M=Q+1$,
meaning it is the core of an $\varepsilon$-normalized complex $\varepsilon$-symmetric conference matrix of size $Q+1$.
\end{proof}

The distinction between the EITFFs of Theorem~\ref{thm.combinatorial}(b) and those of~\cite{EtTaoui18,BlokhuisBE18} is subtle.
To explain, the signature matrix $\bfS$ of an $\EITFF(N,N,2)$ can be obtained by applying either one of the following mappings $\eta,\widetilde{\eta}:\bbC\rightarrow\bbR^{2\times 2}$ to every entry of a suitably chosen $N\times N$ complex matrix $\bfZ$ whose diagonal entries are $0$ and whose off-diagonal entries are unimodular:
\begin{equation*}
\eta(z):=\left[\begin{array}{cr}
\real(z)&-\imag(z)\\
\imag(z)& \real(z)
\end{array}\right],
\quad
\widetilde{\eta}(z)
:=\left[\begin{array}{cr}
\real(z)& \imag(z)\\
\imag(z)&-\real(z)
\end{array}\right]
=\left[\begin{array}{cr}
\real(z)&-\imag(z)\\
\imag(z)& \real(z)
\end{array}\right]\left[\begin{array}{cr}
1&0\\
0&-1
\end{array}\right].
\end{equation*}
Here, $\eta$ is an injective ring homomorphism (a faithful representation of the field $\bbC$) that restricts to the classical canonical isomorphism from $\bbT:=\set{z\in\bbC: \abs{z}=1}$ onto $\operatorname{SO}(2)$.
As essentially noted by Hoggar~\cite{Hoggar77},
applying $\eta$ to each entry of the signature matrix of a complex $\EITFF(D,N,R)$ yields the signature matrix of a real $\EITFF(2D,N,2R)$.
In particular,
$\bfS_{n_1,n_2}:=\eta(\bfZ(n_1,n_2))$ defines the signature matrix of a real $\EITFF(N,N,2)$ if and only if $\bfZ$ is the signature matrix of an $\ETF(\frac N2,N)$,
namely an $N\times N$ complex Hermitian conference matrix;
note here that $N$ is necessarily even.
In contrast, $\widetilde{\eta}$ is not an isomorphism.
In fact, the image of its restriction to $\bbT$ is not even a group,
being the nontrivial coset of $\operatorname{SO}(2)$ in $\operatorname{O}(2)$.
That said, $\widetilde{\eta}$ is related to the isomorphism $\eta$ in a nice way,
and has other useful properties:
\begin{equation*}
\eta(1)=\bfI,
\qquad
[\widetilde{\eta}(z)]^\rmT
=[\widetilde{\eta}(z)]^*
=\widetilde{\eta}(z),
\qquad
\widetilde{\eta}(z_1)\widetilde{\eta}(z_2)
=\eta(z_1\overline{z_2}),
\qquad
\forall\,z,z_1,z_2\in\bbC.
\end{equation*}
As shown in~\cite{EtTaoui18,BlokhuisBE18},
these facts quickly imply that $\bfS_{n_1,n_2}:=\widetilde{\eta}(\bfZ(n_1,n_2))$ defines the signature matrix of a real $\EITFF(N,N,2)$ if and only if $\bfZ$ is an $N\times N$ complex symmetric conference matrix,
and such matrices exist for an infinite number of odd $N$.
Meanwhile, Theorem~\ref{thm.combinatorial}(b) yields an $\EITFF(N,N,2)$ if and only if $\bfZ$ the core of an $(N+1)\times(N+1)$ complex symmetric conference matrix.
That said, the construction of~\cite{EtTaoui18,BlokhuisBE18} does overlap with that of Theorem~\ref{thm.combinatorial}(b) in a special case:
Theorem~1 of~\cite{BlokhuisBE18} constructs a complex symmetric conference matrix $\bfZ$ of size $N$ from a normalized real symmetric conference matrix $\bfC$ of size $N+1$,
and a careful read of its proof reveals that applying $\widetilde{\eta}$ to $\bfZ$ yields the same signature matrix $\bfS$ that one obtains by applying Theorem~\ref{thm.combinatorial}(b) directly to the core of $\bfC$.
From this perspective, Theorems~1 and~3 of~\cite{BlokhuisBE18} and Theorem~\ref{thm.combinatorial}(b) yield two distinct generalizations of Et-Taoui's original breakthrough construction of an $\EITFF(Q,Q,2)$ from the Legendre symbol on $\bbF_Q$ when $Q\equiv1\bmod4$~\cite{EtTaoui18}.

\section*{Acknowledgments}
The views expressed in this article are those of the authors and do not reflect the official policy or position of the United States Air Force, Department of Defense, or the U.S.~Government.
This work was partially supported by NSF DMS 1830066, NSF DMS 1829955, AFOSR FA9550-18-1-0107, and the Air Force Office of Scientific Research Summer Faculty Fellowship Program.

\appendix

\section{Proof of Theorem~\ref{thm.EITFF(11,11,3)}}

\begin{proof}[Proof of Theorem~\ref{thm.EITFF(11,11,3)}]
Let $\chi$ be any generator of $\widehat{\bbF}_{11}^\times$,
implying $\chi(-1)=-1$ and that $\chi^5$ is the restriction of the Legendre symbol~\eqref{eq.Legendre symbol} to $\bbF_{11}^\times$,
satisfying $\chi^5(x)=(\frac{x}{11})\in\set{1,-1}$ for all $x\in\bbF_{11}^\times$.
Also let $c:=\cos(\theta)$, $s:=\sin(\theta)$ where $\theta\in\bbR$,
and let $\omega\in\bbC$ be unimodular.
Consider the eleven $3\times 1$ rank-$1$ matrices $\set{\bfPsi_x}_{x\in\bbF_{11}}$ defined by
\begin{equation*}
\bfPsi_0:=\left[\begin{array}{c}1\\0\\0\end{array}\right],
\quad
\bfPsi_x:=\tfrac{\sqrt{11}}{\sqrt{30}}\left[\begin{array}{c}
\frac{\sqrt{8}}{\sqrt{11}}\smallskip\\
\chi(x)[c(\tfrac{x}{11})+\rmi s\omega]\smallskip\\
\overline{\chi(x)}[c(\tfrac{x}{11})+\rmi s\overline{\omega}]
\end{array}\right],\ \forall\, x\in\bbF_{11}^\times.
\end{equation*}
For any $x\neq0$, the squared modulus of the second and third entries of this matrix are
\begin{align*}
\abs{\bfPsi_x(2,1)}^2
&=\tfrac{11}{30}\abs{c(\tfrac{x}{11})+\rmi s\omega}^2
=\tfrac{11}{30}[c^2+s^2+2cs\real(\rmi\omega)(\tfrac{x}{11})]
=\tfrac{11}{30}[1-2cs\imag(\omega)(\tfrac{x}{11})],\\
\abs{\bfPsi_x(3,1)}^2
&=\tfrac{11}{30}\abs{c(\tfrac{x}{11})+\rmi s\overline{\omega}}^2
=\tfrac{11}{30}[c^2+s^2+2cs\real(\rmi\overline{\omega})(\tfrac{x}{11})]
=\tfrac{11}{30}[1+2cs\imag(\omega)(\tfrac{x}{11})],
\end{align*}
respectively.
In particular, each $\bfPsi_x$ is an isometry (unit-norm column vector):
for any $x\neq0$,
\begin{equation*}
\abs{\bfPsi_x(1,1)}^2+\abs{\bfPsi_x(2,1)}^2+\abs{\bfPsi_x(3,1)}^2
=\tfrac{8}{30}
+\tfrac{11}{30}[1-2cs\imag(\omega)(\tfrac{x}{11})]
+\tfrac{11}{30}[1+2cs\imag(\omega)(\tfrac{x}{11})]
=1.
\end{equation*}
For any $x\neq0$, we moreover have that
\begin{equation*}
(\bfPsi_x^{}\bfPsi_x^*)(2,3)
=\bfPsi(2,1)\overline{\bfPsi(3,1)}
=\tfrac{11}{30}\chi(x)[c(\tfrac{x}{11})+\rmi s\omega]
\chi(x)[c(\tfrac{x}{11})-\rmi s\omega]
=\tfrac{11}{30}(c^2+s^2\omega^2)\chi^2(x).
\end{equation*}
As such, the $3\times 3$ rank-$1$ projections $\set{\bfP_x}_{x\in\bbF_{11}}$,
$\bfP_x:=\bfPsi_x^{}\bfPsi_x^*$ are
$\bfP_0=\left[\begin{smallmatrix}1&0&0\\0&0&0\\0&0&0\end{smallmatrix}\right]$ and,
for any $x\neq0$,
\begin{equation*}
\bfP_x
=\tfrac{11}{30}\left[\begin{array}{ccc}
\tfrac{8}{11}&
\tfrac{\sqrt{8}}{\sqrt{11}}\overline{\chi(x)}[c(\tfrac{x}{11})-\rmi s\overline{\omega}]&
\tfrac{\sqrt{8}}{\sqrt{11}}\chi(x)[c(\tfrac{x}{11})-\rmi s\omega]\smallskip\\
\tfrac{\sqrt{8}}{\sqrt{11}}\chi(x)[c(\tfrac{x}{11})+\rmi s\omega]&
1-2cs\imag(\omega)(\tfrac{x}{11})&
(c^2+s^2\omega^2)\chi^2(x)\smallskip\\
\tfrac{\sqrt{8}}{\sqrt{11}}\overline{\chi(x)}[c(\tfrac{x}{11})+\rmi s\overline{\omega}]&
(c^2+s^2\overline{\omega}^2)\overline{\chi^2}(x)&
1+2cs\imag(\omega)(\tfrac{x}{11})
\end{array}\right].
\end{equation*}
We next compute the matrices $\set{\bfM_\gamma}_{\gamma\in\widehat{\bbF}_{11}}$ of~\eqref{eq.DFT of projections}.
Here, we elect to identify $\bbF_{11}$ with $\widehat{\bbF}_{11}$: for any $y\in\bbF_{11}$,
define the additive character $\gamma_y$ by
\smash{$\gamma_y(x):=\exp(\frac{2\pi\rmi xy}{11})$} for each $x\in\bbF_{11}$.
We also adopt the following shorthand notation for Gauss sums:
\begin{equation}
\label{eq.pf of EITFF(11,11,3) 1}
z_{n,y}:=\ip{\gamma_y}{\chi^n}_\times
=\sum_{x\in\bbF_{11}^\times}\overline{\gamma_y(x)}\chi^n(x)
=\sum_{x\in\bbF_{11}^\times}\exp(-\tfrac{2\pi\rmi xy}{11})\chi^n(x),
\quad\forall\,n\in\bbZ,\, y\in\bbF_{11}.
\end{equation}
In particular, for any $y\in\bbF_{11}$,
the matrix $\bfM_y:=\bfM_{\gamma_y}$ of~\eqref{eq.DFT of projections} has $(1,1)$ entry
\begin{equation*}
\bfM_y(1,1)
=\bfP_0(1,1)+\sum_{x\in\bbF_{11}^\times}\overline{\gamma_y(x)}\bfP_x(1,1)
=1+\tfrac{8}{30}\ip{\gamma_y}{1}_\times
=\tfrac1{30}(30+8z_{0,y}).
\end{equation*}
We compute each remaining entry of $\bfM_y$ by using the facts that $(\frac{x}{11})=\chi^5(x)$ and $\overline{\chi(x)}=\chi^{-1}(x)$ for all $x\in\bbF_{11}^\times$.
For example, by~\eqref{eq.pf of EITFF(11,11,3) 1},
\begin{multline*}
\bfM_y(1,2)
=\sum_{x\in\bbF_{11}^\times}
\overline{\gamma_y(x)}\tfrac{11}{30}\tfrac{\sqrt{8}}{\sqrt{11}}\overline{\chi(x)}[c(\tfrac{x}{11})-\rmi s\overline{\omega}]\\
=\tfrac{11}{30}\tfrac{\sqrt{8}}{\sqrt{11}}\sum_{x\in\bbF_{11}^\times}\overline{\gamma_y(x)}
[c\chi^4(x)-\rmi s\overline{\omega}\chi^{-1}(x)]
=\tfrac{11}{30}\tfrac{\sqrt{8}}{\sqrt{11}}(cz_{4,y}-\rmi s\overline{\omega}z_{-1,y}).
\end{multline*}
Performing this type of analysis for each entry gives that
\begin{equation}
\label{eq.pf of EITFF(11,11,3) 2}
\bfM_y
=\tfrac{11}{30}\left[\begin{array}{ccc}
\tfrac1{11}(30+8z_{0,y})&
\tfrac{\sqrt{8}}{\sqrt{11}}(cz_{4,y}-\rmi s\overline{\omega}z_{-1,y})&
\tfrac{\sqrt{8}}{\sqrt{11}}(cz_{6,y}-\rmi s\omega z_{1,y})\smallskip\\
\tfrac{\sqrt{8}}{\sqrt{11}}(cz_{6,y}+\rmi s\omega z_{1,y})&
z_{0,y}-2cs\imag(\omega)z_{5,y}&
(c^2+s^2\omega^2)z_{2,y}\smallskip\\
\tfrac{\sqrt{8}}{\sqrt{11}}(cz_{4,y}+\rmi s\overline{\omega}z_{-1,y})&
(c^2+s^2\overline{\omega}^2)z_{-2,y}&
z_{0,y}+2cs\imag(\omega)z_{5,y}
\end{array}\right],
\end{equation}
for all $y\in\bbF_{11}$.
Here, combining~\eqref{eq.Gauss sum easy cases} and~\eqref{eq.pf of EITFF(11,11,3) 1} gives $z_{0,0}=10$ and $z_{n,0}=0$ whenever $10\nmid n$,
and so~\eqref{eq.pf of EITFF(11,11,3) 2} reduces to the following when $y=0$:
\begin{equation}
\label{eq.pf of EITFF(11,11,3) 3}
\bfM_0
=\sum_{x\in\bbF_{11}}\bfP_x
=\tfrac{11}{30}\left[\begin{array}{ccc}
\tfrac1{11}[30+8(10)]&0&0\\
0&10&0\\
0&0&10
\end{array}\right]
=\tfrac{11}{3}\bfI.
\end{equation}
Thus, $\set{\bfP_x}_{x\in\bbF_{11}}$ are the projections of a $\TFF(3,11,1)$,
and so yield a harmonic $\TFF(11,11,3)$ via Theorem~\ref{thm.small TFF yields harmonic}.
Moreover, by Theorem~\ref{thm.EITFF},
this harmonic $\TFF(11,11,3)$ is an EITFF if and only if $\bfM_y^*\bfM_y^{}=B\bfI$ for all $y\neq 0$ where \smash{$B=\frac{D(GR-D)}{R^2(G-1)}=\frac{11(22)}{9(10)}=(\tfrac{11}{3})^2\tfrac15$},
that is, if and only if \smash{$\frac{3\sqrt{5}}{11}\bfM_y$} is unitary for all $y\neq0$.
Here, for any $y\neq0$, \eqref{eq.Gauss sum easy cases} implies $z_{0,y}=-1$ and so in this case \eqref{eq.pf of EITFF(11,11,3) 2} becomes
\begin{equation}
\label{eq.pf of EITFF(11,11,3) 4}
\tfrac{3\sqrt{5}}{11}\bfM_y
=\tfrac{1}{2\sqrt{5}}\left[\begin{array}{ccc}
2&
\tfrac{\sqrt{8}}{\sqrt{11}}(cz_{4,y}-\rmi s\overline{\omega}z_{-1,y})&
\tfrac{\sqrt{8}}{\sqrt{11}}(cz_{6,y}-\rmi s\omega z_{1,y})\smallskip\\
\tfrac{\sqrt{8}}{\sqrt{11}}(cz_{6,y}+\rmi s\omega z_{1,y})&
-1-2cs\imag(\omega)z_{5,y}&
(c^2+s^2\omega^2)z_{2,y}\smallskip\\
\tfrac{\sqrt{8}}{\sqrt{11}}(cz_{4,y}+\rmi s\overline{\omega}z_{-1,y})&
(c^2+s^2\overline{\omega}^2)z_{-2,y}&
-1+2cs\imag(\omega)z_{5,y}
\end{array}\right].
\end{equation}
To simplify these matrices,
note that by~\eqref{eq.Gauss sum of conjugate} and~\eqref{eq.pf of EITFF(11,11,3) 1},
\begin{equation*}
z_{-n,y}
=\ip{\gamma_y}{\chi^{-n}}_\times
=\ip{\gamma_y}{\overline{\chi^n}}
=[\chi(-1)]^n\overline{\ip{\gamma_y}{\chi^n}}
=(-1)^n\overline{z_{n,y}},
\quad
\forall\,n\in\bbZ,\,y\in\bbF_{11}.
\end{equation*}
As such, for any $y\neq0$, $z_{5,y}$ is purely imaginary,
and moreover the entries of~\eqref{eq.pf of EITFF(11,11,3) 4} arise in conjugate pairs.
Specifically, its
$(2,2)$ and $(3,3)$ entries are conjugates of each other,
and the same is true for its $(1,2)$ and $(1,3)$ entries,
its $(2,1)$ and $(3,1)$ entries,
and its $(2,3)$ and $(3,2)$ entries.
Thus,
\begin{equation}
\label{eq.pf of EITFF(11,11,3) 5}
\tfrac{3\sqrt{5}}{11}\bfM_y
=\tfrac{1}{2\sqrt{5}}\left[\begin{array}{ccc}
2&
\tfrac{\sqrt{8}}{\sqrt{11}}\overline{(cz_{6,y}-\rmi s\omega z_{1,y})}&
\tfrac{\sqrt{8}}{\sqrt{11}}(cz_{6,y}-\rmi s\omega z_{1,y})\smallskip\\
\tfrac{\sqrt{8}}{\sqrt{11}}(cz_{6,y}+\rmi s\omega z_{1,y})&
-1-2cs\imag(\omega)z_{5,y}&
(c^2+s^2\omega^2)z_{2,y}\smallskip\\
\tfrac{\sqrt{8}}{\sqrt{11}}\overline{(cz_{6,y}+\rmi s\omega z_{1,y})}&
\overline{(c^2+s^2\omega^2)z_{2,y}}&
\overline{-1-2cs\imag(\omega)z_{5,y}}
\end{array}\right],
\end{equation}
for any $y\neq0$.
To proceed,
note that combining~\eqref{eq.Gauss sum magnitude} and~\eqref{eq.pf of EITFF(11,11,3) 1} gives $\abs{z_{n,y}}=\sqrt{11}$ whenever $y\neq0$ and $10\nmid n$.
In particular, $z_{5,y}\in\set{\sqrt{11}\rmi,-\sqrt{11}\rmi}$ for all $y\neq0$.
Moreover, the modulus of the $(2,2)$ entry of~\eqref{eq.pf of EITFF(11,11,3) 1} is independent of $y\neq0$, and the same is true for its $(2,3)$, $(3,2)$ and $(3,3)$ entries.
In order for \eqref{eq.pf of EITFF(11,11,3) 5} to be unitary for all $y\neq0$,
the modulus of its $(2,1)$ entry must thus also be independent of $y\neq0$.
This suggests we choose $\omega:=\frac{z_{6,1}}{z_{1,1}}$.
To elaborate,
note \smash{$\abs{\omega}=\frac{\abs{z_{6,1}}}{\abs{z_{1,1}}}=\frac{\sqrt{11}}{\sqrt{11}}=1$}
as required.
Moreover,
\eqref{eq.pf of EITFF(11,11,3) 1} and a simple reindexing give
\begin{equation}
\label{eq.pf of EITFF(11,11,3) 6}
z_{n,y}
=\ip{\gamma_y}{\chi^n}_\times
=\chi^n(y^{-1})\ip{\gamma_1}{\chi^n}
=\overline{\chi^n(y)}z_{n,1},
\quad\forall\,n\in\bbZ,\,y\in\bbF_{11}^\times.
\end{equation}
As such, for any $y\neq0$,
we have $\tfrac{z_{6,y}}{z_{1,y}}
=\tfrac{\chi^1(y)}{\chi^6(y)}\tfrac{z_{6,1}}{z_{1,1}}
=\omega\chi^5(y)
=\omega(\tfrac{y}{11})$
and so
\begin{equation*}
cz_{6,y}\pm \rmi s\omega z_{1,y}
=c\omega z_{1,y}(\tfrac{y}{11})\pm \rmi s\omega z_{1,y}
=\omega z_{1,y}[c(\tfrac{y}{11})\pm\rmi s].
\end{equation*}
As such, for this particular choice of $\omega$,
\eqref{eq.pf of EITFF(11,11,3) 5} can be reexpressed as
\begin{equation}
\label{eq.pf of EITFF(11,11,3) 7}
\tfrac{3\sqrt{5}}{11}\bfM_y
=\tfrac{1}{2\sqrt{5}}\left[\begin{array}{ccc}
2&
~\tfrac{\sqrt{8}}{\sqrt{11}}\overline{\omega z_{1,y}}[c(\tfrac{y}{11})+\rmi s]~&
\tfrac{\sqrt{8}}{\sqrt{11}}\omega z_{1,y}[c(\tfrac{y}{11})-\rmi s]\smallskip\\
\tfrac{\sqrt{8}}{\sqrt{11}}\omega z_{1,y}[c(\tfrac{y}{11})+\rmi s]&
-1-2cs\imag(\omega)z_{5,y}&
(c^2+s^2\omega^2)z_{2,y}\smallskip\\
\tfrac{\sqrt{8}}{\sqrt{11}}\overline{\omega z_{1,y}}[c(\tfrac{y}{11})-\rmi s]&
\overline{(c^2+s^2\omega^2)z_{2,y}}&
\overline{-1-2cs\imag(\omega)z_{5,y}}
\end{array}\right],
\end{equation}
for any $y\neq0$.
Next, for any $y\neq0$,
we conjugate~\eqref{eq.pf of EITFF(11,11,3) 7} by the $3\times 3$ diagonal unitary matrix $\bfD_y$ whose $(1,1)$, $(2,2)$ and $(3,3)$ entries are $1$, \smash{$\frac1{\sqrt{11}}\overline{\omega z_{1,y}}$} and \smash{$\frac1{\sqrt{11}}\omega z_{1,y}$}, respectively.
In particular, the $(3,2)$ entry of the resulting matrix is
\begin{multline*}
(\tfrac{3\sqrt{5}}{11}\bfD_y^{}\bfM_y^{}\bfD_y^*)(3,2)
=\tfrac{3\sqrt{5}}{11}\bfD_y(3,3)(\bfM_y^{})(3,1)\overline{\bfD_y(2,2)}\\
=\tfrac1{2\sqrt{5}}\tfrac1{\sqrt{11}}\omega z_{1,y}\overline{(c^2+s^2\omega^2)z_{2,y}}\tfrac1{\sqrt{11}}\omega z_{1,y}
=\tfrac1{2\sqrt{5}}(c^2\omega^2+s^2)\tfrac{z_{1,y}^2}{z_{2,y}}.
\end{multline*}
From this, it follows that
\begin{equation}
\label{eq.pf of EITFF(11,11,3) 8}
\tfrac{3\sqrt{5}}{11}\bfD_y^{}\bfM_y^{}\bfD_y^*
=\tfrac{1}{2\sqrt{5}}\left[\begin{array}{ccc}
2&
\sqrt{8}[c(\tfrac{y}{11})+\rmi s]&
\sqrt{8}[c(\tfrac{y}{11})-\rmi s]\smallskip\\
\sqrt{8}[c(\tfrac{y}{11})+\rmi s]&
~-1-2cs\imag(\omega)z_{5,y}~&
\overline{(c^2\omega^2+s^2)\tfrac{z_{1,y}^2}{z_{2,y}}}\smallskip\\
\sqrt{8}[c(\tfrac{y}{11})-\rmi s]&
(c^2\omega^2+s^2)\tfrac{z_{1,y}^2}{z_{2,y}}&
\overline{-1-2cs\imag(\omega)z_{5,y}}
\end{array}\right],
\end{equation}
for any $y\neq0$.
To simplify this even further,
for each $y\neq0$,
we multiply~\eqref{eq.pf of EITFF(11,11,3) 8} on both the left and right by the $3\times 3$ diagonal unitary matrix $\bfDelta_y$ whose $(1,1)$, $(2,2)$ and $(3,3)$ entries are $1$, $c(\tfrac{y}{11})-\rmi s$ and $c(\tfrac{y}{11})+\rmi s$, respectively.
(We caution that $\bfDelta_y$ is not self-adjoint, and so this is not conjugation by $\bfDelta_y$, but nevertheless preserves the property of being unitary.)
In particular, since $[c(\tfrac{y}{11})\pm\rmi s]^{-1}=c(\tfrac{y}{11})\mp\rmi s$ and $[c(\tfrac{y}{11})\pm\rmi s]^2=c^2-s^2\pm2\rmi cs(\tfrac{y}{11})$ for all $y\neq0$,
this yields
\begin{equation}
\label{eq.pf of EITFF(11,11,3) 9}
\tfrac{3\sqrt{5}}{11}\bfDelta_y^{}\bfD_y^{}\bfM_y^{}\bfD_y^*\bfDelta_y^{}
=\tfrac{1}{\sqrt{5}}\left[\begin{array}{ccc}
1&
\sqrt{2}&
\sqrt{2}\smallskip\\
\sqrt{2}&
\frac{-1-2cs\imag(\omega)z_{5,y}}{2[c^2-s^2+2\rmi cs(\tfrac{y}{11})]}&
\overline{\frac{(c^2\omega^2+s^2)}{2}\tfrac{z_{1,y}^2}{z_{2,y}}}\smallskip\\
\sqrt{2}&
\frac{(c^2\omega^2+s^2)}{2}\tfrac{z_{1,y}^2}{z_{2,y}}&
\overline{\frac{-1-2cs\imag(\omega)z_{5,y}}{2[c^2-s^2+2\rmi cs(\tfrac{y}{11})]}}
\end{array}\right],
\end{equation}
for any $y\neq0$.
To summarize,
to show that our harmonic $\TFF(11,11,3)$ here is an EITFF,
it suffices to show that~\eqref{eq.pf of EITFF(11,11,3) 9} is a unitary matrix for all $y\neq0$.
We will prove something stronger,
namely that for the appropriate choices of $\chi$ and $\theta$,
\eqref{eq.pf of EITFF(11,11,3) 9} is some fixed real-symmetric orthogonal matrix regardless of $y\neq0$, namely that
\begin{equation}
\label{eq.pf of EITFF(11,11,3) 10}
\tfrac{3\sqrt{5}}{11}\bfDelta_y^{}\bfD_y^{}\bfM_y^{}\bfD_y^*\bfDelta_y^{}
:=\tfrac{1}{\sqrt{5}}\left[\begin{array}{ccc}
1&
\sqrt{2}&
\sqrt{2}\smallskip\\
\sqrt{2}&
\frac{-1+\sqrt{5}}{2}&
\frac{-1-\sqrt{5}}{2}\smallskip\\
\sqrt{2}&
\frac{-1-\sqrt{5}}{2}&
\frac{-1+\sqrt{5}}{2}
\end{array}\right]
\end{equation}
for all $y\neq0$.
Comparing~\eqref{eq.pf of EITFF(11,11,3) 9} and~\eqref{eq.pf of EITFF(11,11,3) 10},
it suffices to show that for some $\chi$ and $\theta$,
the following two equations hold for all $y\neq0$:
\begin{align}
\label{eq.pf of EITFF(11,11,3) 11}
-1-2cs\imag(\omega)z_{5,y}
&=(-1+\sqrt{5})[c^2-s^2+2\rmi cs(\tfrac{y}{11})],\\
\nonumber
(c^2\omega^2+s^2)z_{1,y}^2
&=(-1-\sqrt{5})z_{2,y}.
\end{align}
(To be clear, \smash{$z_{n,y}:=\ip{\gamma_y}{\chi^n}_{\times}$} depends on one's choice of $\chi$,
and thus so does \smash{$\omega=\frac{z_{6,1}}{z_{1,1}}$}.)
We choose $\theta=\frac{2\pi}{5}$, implying $c=\cos(\frac{2\pi}{5})=\frac14(\sqrt{5}-1)$, \smash{$s=\sin(\frac{2\pi}{5})=\tfrac14\sqrt{10+2\sqrt{5}}$} and also
\begin{equation}
\label{eq.pf of EITFF(11,11,3) 12}
c^2=\tfrac1{8}(3-\sqrt{5}),
\quad
s^2=\tfrac1{8}(5+\sqrt{5}),
\quad
(-1+\sqrt{5})(c^2-s^2)
=(-1+\sqrt{5})\tfrac14(-1-\sqrt{5})
=-1.
\end{equation}
Simplifying~\eqref{eq.pf of EITFF(11,11,3) 11} with~\eqref{eq.pf of EITFF(11,11,3) 12},
we see that it suffices to show that for some $\chi$,
\begin{align*}
-\imag(\omega)z_{5,y}
&=(-1+\sqrt{5})\rmi(\tfrac{y}{11}),\\
\tfrac18[(3-\sqrt{5})\omega^2+(5+\sqrt{5})]z_{1,y}^2
&=(-1-\sqrt{5})z_{2,y},
\end{align*}
for all $y\neq0$.
Now note that for any $y\neq0$, \eqref{eq.pf of EITFF(11,11,3) 6} gives
\begin{equation*}
z_{5,y}=\overline{\chi^5(y)}z_{5,1}=(\tfrac{y}{11})z_{5,1},
\quad
z_{1,y}^2=[\overline{\chi(y)}z_{1,1}]^2=\overline{\chi^2(y)}z_{1,1}^2,
\quad
z_{2,y}=\overline{\chi^2(y)}z_{2,1}.
\end{equation*}
Using this to simplify the previous equations reveals that it suffices to prove that the following two equations hold for some choice of $\chi$;
notably, these equations are independent of $y$:
\begin{align}
\label{eq.pf of EITFF(11,11,3) 13}
-\imag(\omega)z_{5,1}
&=(-1+\sqrt{5})\rmi,\\
\label{eq.pf of EITFF(11,11,3) 14}
\tfrac18[(3-\sqrt{5})\omega^2+(5+\sqrt{5})]z_{1,1}^2
&=(-1-\sqrt{5})z_{2,1}.
\end{align}
It will turn out that~\eqref{eq.pf of EITFF(11,11,3) 13} and~\eqref{eq.pf of EITFF(11,11,3) 14} hold, for example,
for the standard generator $\chi$ of $\widehat{\bbF}_{11}^\times$ that arises from the fact that $2$ is a generator for $\bbF_{11}^\times$,
namely for \smash{$\chi(2^m):=\exp(\frac{2\pi\rmi m}{10})$} for all $m\in\bbZ$.
Assume this for the moment.
Then~\eqref{eq.pf of EITFF(11,11,3) 10} holds for all $y\neq0$,
and so by Theorem~\ref{thm.EITFF}(a) the resulting harmonic $\TFF(11,11,3)$ is an EITFF.
Remarkably, this $\EITFF(11,11,3)$ is real, as we now explain.
Borrowing~\eqref{eq.pf of harmonic EITFF characterization 1} from the proof of Theorem~\ref{thm.EITFF},
we see that the isometries $\set{\bfPhi_y}_{y\in\bbF_{11}}$ of our harmonic $\EITFF(11,11,3)$ satisfy $\bfPhi_{y_1}^*\bfPhi_{y_2}^{}
=\tfrac{3}{11}\bfM_{y_1-y_2}$ for all $y_1,y_2\in\bbF_{11}$.
Moreover, for any $y\in\bbF_{11}$,
\eqref{eq.pf of EITFF(11,11,3) 3} and~\eqref{eq.pf of EITFF(11,11,3) 5} imply that there exists $d\in\bbR$ and $e,f,g,h\in\bbC$ such that
\begin{equation*}
\bfM_y
=\left[\begin{array}{ccc}
d&\overline{e}&e\\
f&g&h\\
\overline{f}&\overline{h}&\overline{g}
\end{array}\right].
\end{equation*}
Conjugating any matrix of this form by
$\bfU:=\frac1{\sqrt{2}}
\left[\begin{array}{ccr}
\sqrt{2}&0&0\\
0&1&1\\
0&\rmi&-\rmi
\end{array}\right]$ yields a real matrix:
\begin{equation*}
\tfrac1{\sqrt{2}}
\left[\begin{array}{ccr}
\sqrt{2}&0&0\\
0&1&1\\
0&\rmi&-\rmi
\end{array}\right]
\left[\begin{array}{ccc}
d&\overline{e}&e\\
f&g&h\\
\overline{f}&\overline{h}&\overline{g}
\end{array}\right]
\tfrac1{\sqrt{2}}
\left[\begin{array}{ccr}
\sqrt{2}&0&0\\
0&1&-\rmi\\
0&1&\rmi
\end{array}\right]
=\left[\begin{array}{ccc}
d&\sqrt{2}\real(e)&-\sqrt{2}\imag(e)\\
\phantom{-}\sqrt{2}\real(f)&\phantom{-}\real(g)+\real(h)&\imag(g)-\imag(h)\\
-\sqrt{2}\imag(f)&-\imag(g)-\imag(h)&\real(g)-\real(h)
\end{array}\right].
\end{equation*}
Thus, $\set{\bfPhi_y\bfU^*}_{y\in\bbF_{11}}$ is an (alternative) sequence of isometries for our harmonic $\EITFF(11,11,3)$ whose cross-Gram matrices are real,
having $(\bfPhi_{y_1}\bfU^*)^*(\bfPhi_{y_2}^{}\bfU^*)
=\bfU\bfPhi_{y_1}^*\bfPhi_{y_2}^{}\bfU^*
=\tfrac{3}{11}\bfU\bfM_{y_1-y_2}\bfU^*$
for any $y_1,y_2\in\bbF_{11}$,
implying this harmonic $\EITFF(11,11,3)$ is itself real.
Alternatively but equivalently, the interested reader can verify that the alternate projections  $\set{\bfU\bfP_x\bfU^*}_{x\in\bbF_{11}}$ satisfy the condition of Theorem~\ref{thm.EITFF}(c),
despite the fact that the original projections $\set{\bfP_x}_{x\in\bbF_{11}}$ do not.

It remains to prove~\eqref{eq.pf of EITFF(11,11,3) 13} and \eqref{eq.pf of EITFF(11,11,3) 14}.
Recall that we have now chosen a particular generator $\chi$ of the Pontryagin dual of $\bbF_{11}^\times$, defined by \smash{$\chi(2^m):=\exp(\frac{2\pi\rmi m}{10})$} for all $m\in\bbZ$.
That is, letting $a:=\exp(\frac{2\pi\rmi}{5})$, $\chi$ maps $\bbF_{11}^\times$ into $\bbT$ as follows:
\begin{equation*}
\set{1,2,3,4,5,6,7,8,9,10}\mapsto\set{1,-a^3,a^4,a,a^2,-a^2,-a,-a^4,a^3,-1}.
\end{equation*}
Further recall that we have associated an additive character $\gamma_y$ of $\bbF_{11}$ to any $y\in\bbF_{11}$, defined by $\gamma_y(x):=\exp(\frac{2\pi\rmi xy}{11})=b^{xy}$ for all $x\in\bbF_{11}$, where $b:=\exp(\frac{2\pi\rmi}{11})$.
For any integer $n$ we thus have
\begin{align*}
z_{n,1}
&:=\ip{\gamma_1}{\chi^n}_\times
=\sum_{x\in\bbF_{11}^\times}\overline{\gamma(x)}\chi^n(x)\\
&=[b^{10}+(-1)^nb]
+a^n[b^7+(-1)^n b^4]
+a^{2n}[b^6+(-1)^n b^5]
+a^{3n}[b^2+(-1)b^9]
+a^{4n}[b^8+(-1)b^3].
\end{align*}
In particular, letting $n$ be $1$, $2$, $5$ and $6$ gives
\begin{align}
\label{eq.z1}
z_{1,1}
&=(b^{10}-b)
+a(b^7-b^4)
+a^2(b^6-b^5)
+a^3(b^2-b^9)
+a^4(b^8-b^3),\\
\label{eq.z2}
z_{2,1}
&=(b^{10}+b)
+a^2(b^7+b^4)
+a^4(b^6+b^5)
+a(b^2+b^9)
+a^3(b^8+b^3),\\
\label{eq.z5}
z_{5,1}
&=(b^{10}-b)
+(b^7-b^4)
+(b^6-b^5)
+(b^2-b^9)
+(b^8-b^3),\\
\label{eq.z6}
z_{6,1}
&=(b^{10}+b)
+a(b^7+b^4)
+a^2(b^6+b^5)
+a^3(b^2+b^9)
+a^4(b^8+b^3),
\end{align}
respectively.
Our goal here is to prove that these four quantities satisfy two particular algebraic relationships, namely~\eqref{eq.pf of EITFF(11,11,3) 13} and~\eqref{eq.pf of EITFF(11,11,3) 14}.
From~\eqref{eq.Gauss sum magnitude}, recall that these four numbers have modulus $\sqrt{11}$.
Since clearly $\overline{z_{5,1}}=-z_{5,1}$, this implies $z_{5,1}\in\set{\sqrt{11}\rmi,-\sqrt{11}\rmi}$.
To proceed, we note that in fact, $z_{5,1}=-\sqrt{11}\rmi$:
since $11\equiv 3\bmod 4$, the classical theory of quadratic Gauss sums gives
\begin{equation}
\label{eq.appendix 1 1}
\sqrt{11}\rmi=\sum_{x\in\bbF_{11}}\exp(\tfrac{2\pi\rmi x^2}{11})
=1+2(b+b^3+b^4+b^5+b^9);
\end{equation}
subtracting $0=\sum_{n=0}^{10}b^n$ from this,
and comparing the result against~\eqref{eq.z5} gives $\sqrt{11}\rmi=\overline{z_{5,1}}$.
As such, \eqref{eq.pf of EITFF(11,11,3) 13} equates to having \smash{$\imag(\omega)=\frac{-1+\sqrt{5}}{\sqrt{11}}$},
where $\omega:=\frac{z_{6,1}}{z_{1,1}}$.
We now claim something which both immediately implies this fact---and so~\eqref{eq.pf of EITFF(11,11,3) 13}---and will help us prove~\eqref{eq.pf of EITFF(11,11,3) 14},
namely that
\begin{equation}
\label{eq.appendix 1 2}
\sqrt{11}\rmi\tfrac{z_{6,1}}{z_{1,1}}
=\sqrt{11}\rmi\omega
=(1-\sqrt{5})+\rmi\sqrt{5+2\sqrt{5}}
=1+2a^2-2a^4.
\end{equation}
The right-hand equality of~\eqref{eq.appendix 1 2} involves trigonometry: since $\exp(\frac{2\pi\rmi}{5})=a$ and $\exp(\frac{2\pi\rmi}{10})=-a^3$,
\begin{gather}
\label{eq.appendix 1 3}
\tfrac12(a+a^4)=\cos(\tfrac{2\pi}{5})=\tfrac{1}{4}(-1+\sqrt{5}),
\qquad
\tfrac1{2\rmi}(a-a^4)=\sin(\tfrac{2\pi}{5})=\tfrac{1}{4}\sqrt{10+2\sqrt{5}},\\
\nonumber
\tfrac12(-a^3-a^2)=\cos(\tfrac{2\pi}{10})=\tfrac1{\sqrt{8}}\sqrt{3+\sqrt{5}},
\quad
\tfrac1{2\rmi}(-a^3+a^2)=\sin(\tfrac{2\pi}{10})=\tfrac14\sqrt{10-2\sqrt{5}},
\end{gather}
implying
\begin{equation*}
\sqrt{5+2\sqrt{5}}
=\sqrt{10+2\sqrt{5}}\,\tfrac1{\sqrt{8}}\sqrt{3+\sqrt{5}}
=4\sin(\tfrac{2\pi}{5})\cos(\tfrac{2\pi}{10})
=\tfrac1{\rmi}(a-a^4)(-a^3-a^2);
\end{equation*}
combining these identities with the fact that $0=\sum_{n=0}^4 a^n$ gives
\begin{equation*}
(1-\sqrt{5})+\rmi\sqrt{5+2\sqrt{5}}
=-2(a+a^4)+(a-a^4)(-a^3-a^2)
=-a+a^2-a^3-3a^4
=1+2a^2-2a^4,
\end{equation*}
as claimed.
In light of this,
\eqref{eq.appendix 1 1} and the fact that $\omega:=\frac{z_{6,1}}{z_{1,1}}$,
\eqref{eq.appendix 1 2} holds if and only if
\begin{equation}
\label{eq.appendix 1 4}
[1+2(b+b^3+b^4+b^5+b^9)]z_{6,1}-(1+2a^2-2a^4)z_{1,1}
=0.
\end{equation}
By~\eqref{eq.z1} and~\eqref{eq.z6},
all the quantities in~\eqref{eq.appendix 1 4} lie in the field $\bbQ[\exp(\frac{2\pi\rmi}{55})]$,
meaning this equation can be verified via cyclotomy.
In detail, by~\eqref{eq.z1}
\begin{align*}
&(1+2a^2-2a^4)z_{1,1}\\
&\qquad=(1+2a^2-2a^4)[(b^{10}-b)+a(b^7-b^4)+a^2(b^6-b^5)+a^3(b^2-b^9)+a^4(b^8-b^3)]\\
&\qquad=(-b+2b^2+2b^4-2b^7-2b^9+b^{10})\\
&\qquad\qquad+a(-2b^3-b^4+2b^5-2b^6+b^7+2b^8)\\
&\qquad\qquad+a^2(-2b-2b^2-b^5+b^6+2b^9+2b^{10})\\
&\qquad\qquad+a^3(b^2+2b^3-2b^4+2b^7-2b^8-b^9)\\
&\qquad\qquad+a^4(2b-b^3-2b^5+2b^6+b^8-2b^{10}),
\end{align*}
while by~\eqref{eq.z6},
\begin{align*}
&[1+2(b+b^3+b^4+b^5+b^9)]z_{6,1}\\
&=[1+2(b+b^3+b^4+b^5+b^9)]
[(b^{10}+b)+a(b^7+b^4)+a^2(b^6+b^5)+a^3(b^2+b^9)+a^4(b^8+b^3)]\\
&=(2+b+4b^2+2b^3+4b^4+2b^5+2b^6+2b^8+3b^{10})\\
&\qquad+a(2+2b+2b^2+b^4+4b^5+3b^7+4b^8+2b^9+2b^{10})\\
&\qquad+a^2(2+2b^3+2b^4+b^5+3b^6+2b^7+2b^8+4b^9+4b^{10})\\
&\qquad+a^3(2+2b+3b^2+4b^3+2b^5+2b^6+4b^7+b^9+2b^{10})\\
&\qquad+a^4(2+4b+2b^2+b^3+2b^4+4b^6+2b^7+3b^8+2b^9),
\end{align*}
and subtracting the former from the latter gives exactly
\begin{equation*}
[1+2(b+b^3+b^4+b^5+b^9)]z_{6,1}-(1+2a^2-2a^4)z_{1,1}
=2\biggparen{\,\sum_{n=0}^4a^n}\biggparen{\,\sum_{n=0}^{10}b^n}
=0,
\end{equation*}
namely~\eqref{eq.appendix 1 4}.
Having~\eqref{eq.appendix 1 2} and so~\eqref{eq.pf of EITFF(11,11,3) 13},
we turn to the proof of~\eqref{eq.pf of EITFF(11,11,3) 14},
which equates to having
\begin{equation}
\label{eq.appendix 1 5}
(-1+\sqrt{5})[(3-\sqrt{5})\omega^2+(5+\sqrt{5})]z_{1,1}^2+32z_{2,1}=0.
\end{equation}
Here,
squaring~\eqref{eq.appendix 1 2} gives
$-11\omega^2
=(1+2a^2-2a^4)^2
=1-8a+4a^2+4a^3
=-3-12a-4a^4$,
while~\eqref{eq.appendix 1 3} gives $\sqrt{5}=1+2a+2a^4$,
and so the coefficient of $z_{1,1}^2$ in~\eqref{eq.appendix 1 5} can be expressed as:
\begin{multline*}
(-1+\sqrt{5})[(3-\sqrt{5})\omega^2+(5+\sqrt{5})]\\
=(2a+2a^4)[(2-2a-2a^4)\tfrac1{11}(3+12a+4a^4)+(6+2a+2a^4)]
=\tfrac{32}{11}(2-a-2a^3).
\end{multline*}
Thus, \eqref{eq.pf of EITFF(11,11,3) 14} equates to~\eqref{eq.appendix 1 5} which equates to having
\begin{equation}
\label{eq.appendix 1 6}
(2-a-2a^3)z_{1,1}^2+11z_{2,1}=0.
\end{equation}
Here, directly squaring~\eqref{eq.z1} gives
\begin{align*}
z_{1,1}^2
&=(-2-2b+b^2+2b^3+2b^8+b^9-2b^{10})\\
&\qquad+a(-2-2b^2+b^4+2b^5+2b^6+b^7-2b^9)\\
&\qquad+a^2(-2+2b+b^3-2b^4-2b^7+b^8+2b^{10})\\
&\qquad+a^3(-2+2b^2-2b^3+b^5+b^6-2b^8+2b^9)\\
&\qquad+a^4(-2+b+2b^4-2b^5-2b^6+2b^7+b^{10}).
\end{align*}
Multiplying this by $2-a-2a^3$ gives
\begin{align*}
(2-a-2a^3)z_{1,1}^2
&=(2-9b+2b^2+2b^3+2b^4+2b^5+2b^6+2b^7+2b^8+2b^9-9b^{10})\\
&\qquad+a(2+2b-9b^2+2b^3+2b^4+2b^5+2b^6+2b^7+2b^8-9b^9+2b^{10})\\
&\qquad+a^2(2+2b+2b^2+2b^3-9b^4+2b^5+2b^6-9b^7+2b^8+2b^9+2b^{10})\\
&\qquad+a^3(2+2b+2b^2-9b^3+2b^4+2b^5+2b^6+2b^7-9b^8+2b^9+2b^{10})\\
&\qquad+a^4(2+2b+2b^2+2b^3+2b^4-9b^5-9b^6+2b^7+2b^8+2b^9+2b^{10}),
\end{align*}
and comparing this to~\eqref{eq.z2} gives
\smash{$\displaystyle(2-a-2a^3)z_{1,1}^2+11z_{2,1}
=2\sum_{m=0}^4a^m\sum_{n=0}^{10}b^n
=0$},
namely~\eqref{eq.appendix 1 6} and so~\eqref{eq.pf of EITFF(11,11,3) 14}.
\end{proof}

Despite some effort, we have not yet been able to generalize the construction of Theorem~\ref{thm.EITFF(11,11,3)} in a way that yields other EITFFs.
We leave this problem for future work,
but do note one remarkable aspect of its construction that went unmentioned in the above proof:
for any $y\neq0$,
it is straightforward to explicitly compute \smash{$\tfrac{3\sqrt{5}}{11}\bfD_y^{}\bfM_y^{}\bfD_y^*$} from~\eqref{eq.pf of EITFF(11,11,3) 10},
and to then in turn show that the cube of this matrix is $-\bfI$.
Equivalently, for any $y\neq0$,
we have that $\bfM_y$ has trace zero, and has a (real) negative eigenvalue.
Whereas the former is a consequence of Theorem~\ref{thm.EITFF(11,11,3)}(d) and the fact that $\rank(\bfP_x)=1$ for all $x\in\bbF_{11}$, the latter seems truly remarkable.


\begin{thebibliography}{WW}

\bibitem{BalanBCE09}
R.~Balan, B.~G.~Bodmann, P.~G.~Casazza, D.~Edidin,
Painless reconstruction from magnitudes of frame coefficients,
J.\ Fourier Anal.\ Appl.\ 15 (2009) 488--501.

\bibitem{BlokhuisBE18}
A.~Blokhuis, U.~Brehm, B.~Et-Taoui,
Complex conference matrices and equi-isoclinic planes in Euclidean spaces,
Beitr.\ Algebra Geom.\ 59 (2018) 491--500.

\bibitem{CalderbankTX15}
R.~Calderbank, A.~Thompson, Y.~Xie,
On block coherence of frames,
Appl.\ Comput.\ Harmon.\ Anal.\ 38 (2015) 50--71.

\bibitem{CasazzaFMWZ11}
P.~G.~Casazza, M.~Fickus, D.~G.~Mixon, Y.~Wang, Z.~Zhou,
Constructing tight fusion frames,
Appl.\ Comput.\ Harmon.\ Anal.\ 30 (2011) 175--187.

\bibitem{CohnKM16}
H.~Cohn, A.~Kumar, G.~Minton,
Optimal simplices and codes in projective spaces,
Geom.\ Topol.\ 20 (2016) 1289--1357.

\bibitem{ConwayHS96}
J.~H.~Conway, R.~H.~Hardin, N.~J.~A.~Sloane,
Packing lines, planes, etc.: packings in Grassmannian spaces,
Exp.\ Math.\ 5 (1996) 139--159.

\bibitem{DelsarteGS71}
P.~Delsarte, J.-M.~Goethals, J.~J.~Seidel,
Orthogonal matrices with zero diagonal.~II,
Canad.\ J.\ Math.\ 23 (1971) 816--832.

\bibitem{DhillonHST08}
I.~S.~Dhillon, J.~R.~Heath, T.~Strohmer, J.~A.~Tropp,
Constructing packings in Grassmannian manifolds via alternating projection,
Exp.\ Math.\ 17 (2008) 9--35.

\bibitem{DingF07}
C.~Ding, T.~Feng,
A generic construction of complex codebooks meeting the Welch bound,
IEEE Trans.\ Inform.\ Theory 53 (2007) 4245--4250.

\bibitem{EldarKB10}
Y.~C.~Eldar, P.~Kuppinger, H.~Bölcskei,
Block-sparse signals: uncertainty relations and efficient recovery
IEEE Trans.\ Signal Process.\ 58 (2010) 3042--3054.

\bibitem{EtTaoui18}
B.~Et-Taoui,
Infinite family of equi-isoclinic planes in Euclidean odd dimensional spaces and of complex symmetric conference matrices of odd orders,
Linear Algebra Appl.\  556 (2018) 373--380.

\bibitem{EtTaoui20}
B.~Et-Taoui,
Quaternionic equiangular lines,
Adv.\ Geom.\ 20 (2020) 273--284.

\bibitem{EtTaouiM21}
B.~Et-Taoui, A.~Makhlouf,
Complex skew-symmetric conference matrices,
to appear in Linear Multilinear Algebra.

\bibitem{EtTaouiR09}
B.~Et-Taoui, J.~Rouyer,
On $p$-tuples of equi-isoclinic $3$-spaces in the Euclidean space,
Indag.\ Math.\ (N.S.) 20 (2009) 491--525.

\bibitem{FickusJKM18}
M.~Fickus, J.~Jasper, E.~J.~King, D.~G.~Mixon,
Equiangular tight frames that contain regular simplices,
Linear Algebra Appl.\ 555 (2018) 98--138.

\bibitem{FickusM21}
M.~Fickus, B.~R.~Mayo,
Mutually unbiased equiangular tight frames,
IEEE Trans.\ Inform.\ Theory 67 (2021) 1656--1667.

\bibitem{FickusMW21}
M.~Fickus, B.~R.~Mayo, C.~E.~Watson,
Certifying the novelty of equichordal tight fusion frames,
arXiv:2103.03192 (2021).

\bibitem{FickusM16}
M.~Fickus, D.~G.~Mixon, Tables of the existence of equiangular tight frames, arXiv:1504.00253 (2016).

\bibitem{FickusS20}
M.~Fickus, C.~A.~Schmitt,
Harmonic equiangular tight frames comprised of regular simplices,
Linear Algebra Appl.\ 586 (2020) 130--169.

\bibitem{GAP}
The GAP Group,
GAP -- Groups, Algorithms, and Programming,
Version 4.11.1, 2021.

\bibitem{GodsilR09}
C.~Godsil, A.~Roy,
Equiangular lines, mutually unbiased bases, and spin models,
European J.\ Combin.\ 30 (2009) 246--262.

\bibitem{GoyalKK01}
V.~K.~Goyal, J.~Kova\v{c}evi\'{c}, J.~A.~Kelner,
Quantized frame expansions with erasures,
Appl.\ Comput.\ Harmon.\ Anal.\ 10 (2001) 203--233.

\bibitem{Hoggar77}
S.~G.~Hoggar,
New sets of equi-isoclinic $n$-planes from old,
Proc.\ Edinb.\ Math.\ Soc.\ 20 (1977) 287--291.

\bibitem{IversonKM21}
J.~W.~Iverson, E.~J.~King, D.~G.~Mixon,
A note on tight projective 2-designs,
J.\ Combin.\ Des.\ 29 (2021) 809--832.

\bibitem{King16}
E.~J.~King,
New constructions and characterizations of flat and almost flat {G}rassmannian fusion frames,
arXiv:1612.05784 (2016).

\bibitem{Konig99}
H.~K\"{o}nig,
Cubature formulas on spheres,
Math.\ Res.\ 107 (1999) 201--212.

\bibitem{LemmensS73b}
P.~W.~H.~Lemmens, J.~J.~Seidel,
Equi-isoclinic subspaces of Euclidean spaces,
Indag.\ Math.\ 76 (1973) 98--107.

\bibitem{LidlN97}
R.\ Lidl, H.\ Niederreiter,
Finite fields,
Cambridge U.\ Press, 1997.

\bibitem{Renes07}
J.~M.~Renes,
Equiangular tight frames from Paley tournaments,
Linear Algebra Appl.~426 (2007) 497--501.

\bibitem{Strohmer08}
T.~Strohmer,
A note on equiangular tight frames,
Linear Algebra Appl.~429 (2008) 326--330.

\bibitem{StrohmerH03}
T.~Strohmer, R.~W.~Heath,
Grassmannian frames with applications to coding and communication,
Appl.\ Comput.\ Harmon.\ Anal.\ 14 (2003) 257--275.

\bibitem{Turyn65}
R.~J.~Turyn,
Character sums and difference sets,
Pacific J.\ Math.\ 15 (1965) 319--346.

\bibitem{ValeW04}
R.~Vale, S.~Waldron,
Tight frames and their symmetries,
Constr. Approx. 21 (2004) 83--112.

\bibitem{Waldron18}
S.~Waldron,
An introduction to finite tight frames,
Birkh\"{a}user, Basel, 2018.

\bibitem{Waldron20}
S.~Waldron,
Tight frames over the quaternions and equiangular lines,
arXiv:2006.06126.

\bibitem{Welch74}
L.~R.~Welch,
Lower bounds on the maximum cross correlation of signals,
IEEE Trans.\ Inform.\ Theory 20 (1974) 397--399.

\bibitem{XiaZG05}
P.~Xia, S.~Zhou, G.~B.~Giannakis,
Achieving the Welch bound with difference sets,
IEEE Trans.\ Inform.\ Theory 51 (2005) 1900--1907.

\bibitem{Zauner99}
G.~Zauner,
Quantum designs: Foundations of a noncommutative design theory,
Ph.D.\ Thesis, University of Vienna, 1999.

\end{thebibliography}
\end{document}